\pdfoutput=1
\RequirePackage{ifpdf}
\ifpdf 
\documentclass[pdftex]{sigma}
\else
\documentclass{sigma}
\fi

\numberwithin{equation}{section}

\newtheorem{Theorem}{Theorem}[section]
\newtheorem{Corollary}[Theorem]{Corollary}
\newtheorem{Lemma}[Theorem]{Lemma}
\newtheorem{prop}[Theorem]{Proposition}
 { \theoremstyle{definition}
\newtheorem{Definition}[Theorem]{Definition}

\newtheorem{Remark}[Theorem]{Remark} }

\usepackage{mathtools,physics}
\usepackage{bm,braket,faktor,esint}
\usepackage[export]{adjustbox}

\usepackage{enumerate,enumitem}
\usepackage{dsfont}

\usepackage{tabularx,tabulary,booktabs}	

\newcommand{\ee}{\mathrm{e}}
\newcommand{\NN}{\mathbb{Z}_{>0}}
\newcommand{\ZZ}{\mathbb{Z}}

\newcommand{\CC}{\mathbb{C}}

\newcommand{\ind}{\mathds{1}}

\newcommand{\Z}{\mathbf{Z}}

\newcommand{\C}{\mathbb{C}}
\newcommand{\mean}[1]{\left\langle{#1}\right\rangle}
\newcommand{\W}{\mathcal{W}}

\newcommand{\A}{\mathcal{A}}
\newcommand{\B}{\mathcal{B}}

\DeclareMathOperator{\pf}{Pf}

\begin{document}

\newcommand{\arXivNumber}{2312.12229}

\renewcommand{\PaperNumber}{054}

\FirstPageHeading

\ShortArticleName{Fay Identities of Pfaffian Type for Hyperelliptic Curves}

\ArticleName{Fay Identities of Pfaffian Type\\ for Hyperelliptic Curves}

\Author{Ga\"etan BOROT~$^{\rm a}$ and Thomas BUC-D'ALCH\'E~$^{\rm b}$}

\AuthorNameForHeading{G.~Borot and T.~Buc-d'Alch\'e}

\Address{$^{\rm a)}$~Institut f\"ur Mathematik und Institut f\"ur Physik, Humboldt-Universit\"at zu Berlin,\\
\hphantom{$^{\rm a)}$}~Unter den Linden 6, 10099 Berlin, Germany}
\EmailD{\href{mailto:gaetan.borot@hu-berlin.de}{gaetan.borot@hu-berlin.de}}

\Address{$^{\rm b)}$~UMPA UMR 5669, ENS de Lyon, CNRS, 46, all\'ee d'Italie 69007, Lyon, France}
\EmailD{\href{mailto:thomas.buc-dalche@ens-lyon.fr}{thomas.buc-dalche@ens-lyon.fr}}

\ArticleDates{Received January 30, 2024, in final form June 06, 2024; Published online June 23, 2024}

\Abstract{We establish identities of Pfaffian type for the theta function associated with twice or half the period matrix of a hyperelliptic curve. They are implied by the large size asymptotic analysis of exact Pfaffian identities for expectation values of ratios of characteristic polynomials in ensembles of orthogonal or quaternionic self-dual random matrices. We show that they amount to identities for the theta function with the period matrix of a~hyperelliptic curve, and in this form we reprove them by direct geometric methods.}

\Keywords{random matrix theory; theta function; Fay's identity; hyperelliptic curves}

\Classification{60B20; 14H42}

\section{Introduction}

The Fay trisecant identity \cite{Fay} is a property of the Riemann theta function associated to the period matrix $\bm{\tau}$ of a compact Riemann surface $\widehat{C}$ of genus $g > 0$
\begin{gather}
0 = \theta(\boldsymbol{v}|\bm{\tau}) \theta (\boldsymbol{u}(z_1) - \boldsymbol{u}(z_3) + \boldsymbol{c} |\bm{\tau} )\theta (\boldsymbol{u}(z_2) - \boldsymbol{u}(z_4) + \boldsymbol{c} |\bm{\tau} ) \nonumber\\
\hphantom{0=}\quad{}
\times \theta (\boldsymbol{v} + \boldsymbol{u}(z_1) - \boldsymbol{u}(z_2) + \boldsymbol{u}(z_3) - \boldsymbol{u}(z_4) |\bm{\tau} ) \nonumber\\
\hphantom{0=}{} - \theta(\boldsymbol{u}(z_1) - \boldsymbol{u}(z_4) + \boldsymbol{c} |\bm{\tau} )\theta (\boldsymbol{u}(z_3) - \boldsymbol{u}(z_2) + \boldsymbol{c} |\bm{\tau} )\nonumber\\
\hphantom{0=}\quad {}\times\theta (\boldsymbol{v} + \boldsymbol{u}(z_1) - \boldsymbol{u}(z_2) |\bm{\tau} )\theta (\boldsymbol{v} + \boldsymbol{u}(z_3) - \boldsymbol{u}(z_4) |\bm{\tau} ) \nonumber\\
\hphantom{0=}{}+ \theta (\boldsymbol{u}(z_1) - \boldsymbol{u}(z_2) + c |\bm{\tau} )\theta (\boldsymbol{u}(z_3) - \boldsymbol{u}(z_4) + \boldsymbol{c} |\bm{\tau} ) \nonumber\\
\hphantom{0=}\quad {}\times \theta (\boldsymbol{v} + \boldsymbol{u}(z_1) - \boldsymbol{u}(z_4) |\bm{\tau} )\theta (\boldsymbol{v} + \boldsymbol{u}(z_3) - \boldsymbol{u}(z_2) |\bm{\tau} ) ,\label{Fayid}
\end{gather}
where $\boldsymbol{u}(z_i)$ is the image via the Abel map of a point $z_i$ in the universal cover $\widetilde{C}$ of $\widehat{C}$ and $\boldsymbol{v}$ is an arbitrary $g$-dimensional vector, all notations appearing in this formula will be reviewed later. Its proof is an application of basic facts about the geometry of the Jacobian of $\widehat{C}$. This identity admits a generalisation to a determinantal identity involving $2n$ points on $\widetilde{C}$.

The Fay identity is a showcase of the deep relations between the
geometry of Riemann surfaces and integrability. It is responsible for
the existence of the algebro-geometric solutions to the KP hierarchy
\cite{Krichever}, which are associated to any fixed Riemann surface $\widehat{C}$; in
this context~\eqref{Fayid} is an equivalent form of the Hirota equation. The
real-valued solutions among those give rise to the finite-gap
potentials for the associated linear differential system, which
historically have played an important role in the study of the KdV and
KP hierarchies \cite{Matveev}. The Fay identity is also the basis of a solution
to the Schottky problem: as conjectured by Novikov and proved by
Shiota~\cite{Shiota} building on earlier work of Mulase \cite{Mulase}, \eqref{Fayid} characterises
period matrices among complex symmetric matrices $\bm{\tau}$ with
positive-definite imaginary part. In the algebro-geometric solutions,
coupling the moduli of $\widehat{C}$ to the KP flows does not give anymore
an exact solution, but captures the long-time asymptotics of more
general solutions of KP, see, e.g., the conjectures in~\cite{BEInt}.

The $1$-hermitian matrix model exhibits Toda integrability (closely related to KP integrability), which manifests itself by the existence of determinantal formulae to compute various observables, valid for any matrix size $N$. In particular, the average of any ratio of characteristic polynomials (called $2n$-point functions) can be expressed as a determinant of expectation values of the ratio of two characteristic polynomials. The large $N$ asymptotic of these matrix models has been extensively studied, either by Riemann--Hilbert methods relying on integrability~\mbox{\cite{CFWW,CGmL, Pastur06}}, or by probabilistic techniques \cite{APS01,BG11,BGmulti,BGK,Johan,Smulticut}. The multi-cut regime, when the large $N$ spectral density of the random matrix is supported on $g + 1$ segments, is particularly interesting. As observed numerically in \cite{Jur91}, explained heuristically in \cite{BDE} and justified rigorously in \cite{BGmulti,Smulticut}, the asymptotic behavior is of oscillatory nature. In particular, fluctuations of linear statistics in the macroscopic regime are asymptotically described as the independent sum of a Gaussian and a discrete Gaussian living on an $N$-dependent lattice. This can be precisely described through the Riemann theta function of the underlying spectral curve, which is hyperelliptic of genus $g$. Given the parallel with integrability, it should not be a surprise that the exact determinantal formulae in the hermitian matrix model imply, in the large $N$ limit up to~$o(1)$, the Fay identities \eqref{Fayid}. This implication will be shown in Proposition~\ref{corol:fay-from-beta=2}.

The purpose of our work is to generalise this to orthogonal and quaternionic self-dual 1-matrix models. The determinantal formulae of the hermitian case for $2n$-point functions of ratios of characteristic polynomials are then replaced with the Pfaffian formulae found by Borodin and Strahov \cite{borodin_averages}. These models correspond to the $\beta = 1$ and $\beta = 4$ cases of the $\beta$-ensembles, whose asymptotic analysis in the $(g+1)$-cut regime has been established for all $\beta > 0$ by probabilistic techniques in \cite{BGmulti,BGK,Smulticut}. The large $N$ spectral density is described by a hyperelliptic curve of genus $g$ independent of $\beta$ and having a period matrix $\bm{\tau}$. The asymptotics of the partition function and the $2n$-point functions are governed by the theta function associated with the matrix $\frac{\beta}{2} \bm{\tau}$. The appearance of the theta function does not have a geometric origin\footnote{The asymptotic analysis carried via the Riemann--Hilbert method does give a geometric origin to the theta function, but it is only applied to the hermitian case, i.e., $\beta = 2$.} but is rather explained by eigenvalue tunnelling between the different connected components of the support \cite{BDE}. Inserting these asymptotics up to $o(1)$ in the Pfaffian identities for the $2n$-point functions yield identities between these theta functions, which can be expressed solely in terms of the geometry of the underlying spectral curve.

The spectral curves arising from the large $N$ limit of the matrix models we consider must be hyperelliptic, have real Weierstra\ss{} points, and have the Boutroux property. We show that all such curves can be realised as the spectral curve of an off-critical $\beta$-ensemble with polynomial potential (Proposition~\ref{propallhyp}). By analytic continuation we can extend the validity of the resulting identities to all hyperelliptic curves. This gives our main results: Theorem~\ref{thm:formula_beta_2} for $\beta = 2$, Theorem~\ref{thm:formula_beta_1} for $\beta = 1$ and Theorem~\ref{thm:formula_beta_4} for $\beta = 4$. It turns out that all three identities can be reformulated in terms of theta functions for the matrix of periods $\boldsymbol{\tau}$ \big(instead of $\frac{\beta}{2}\bm{\tau}$\big) and in this form we are able to give them a second proof by direct algebraic methods. Interestingly, the $\beta = 1$ and $\beta = 4$ identities are equivalent via the modular properties of theta functions, and the $\beta = 2$ identity implies the Fay identity in the special case of hyperelliptic curves.

As a byproduct of our proofs, we obtain a seemingly new formula (Proposition~\ref{lem:Ef0} proved in Section~\ref{sec:can}) for the equilibrium energy of the $\beta$-ensembles in the multi-cut regime in terms of the geometry of the spectral curve. Although the ingredients are the same, at first sight it does not have exactly the same form as the 1-matrix model specialisation of the formula known in the context of the $2$-matrix model \cite{Bertola}. Independently of our analysis, we also establish (Proposition~\ref{lem:compute-nu} proved in Appendix~\ref{AppA}) an explicit formula for the derivative with respect to filling fractions of the equilibrium entropy. For $\beta \neq 2$, the equilibrium entropy appears as the order $N$ term in the free energy of the $\beta$-ensemble, and its derivatives with respect to filling fractions appear both in the asymptotics of the partition function (Theorem~\ref{Asymppart}) and in the centering in the generalised central limit theorem (Theorem~\ref{thm:clt}).

The strategy of proof via asymptotics in integrable random matrix ensembles is somehow more interesting than the resulting identities in the particular case we studied, and constitutes the originality of this study. In principle this strategy can be applied to any random matrix ensemble:
\begin{itemize}\itemsep=0pt
\item[(i)] which is amenable to asymptotic analysis up to $o(1)$ in the large size limit and in the multi-cut regime;
\item[(ii)] in which exact formulae for $2n$-point functions in terms of $k$-point functions (with $k$ independent of $n$) are available.
\end{itemize}
In principle, more general algebraic spectral curves can be obtained in two-matrix models or in linearly coupled chain of matrices. In such models, (ii) is addressed by the Eynard--Mehta formulae \cite{EM} but obtaining (i) already for the two-matrix model away from critical points in the multi-cut regime is a notoriously hard open problem. More general algebraic spectral curves can also be obtained in repulsive particle systems with $d$ groups of particles, in which the repulsion intensity between particles of groups $i$ and $j$ is $\beta_{i,j}$. The latter appear naturally in various situations (see, e.g., \cite{BEO,BEW}) and (i) has been addressed by \cite{BGK}. Their discrete counterpart appears in models related to random two-dimensional tilings and is also amenable to asymptotic analysis \cite{BGG}. We expect that in some of these models, integrability properties (ii) should exist, and will therefore imply identities between theta functions with a more complicated structure and that depend on a larger class of algebraic curves. By this we mean the matrix $\boldsymbol{\tau}$ in the theta function will be specified from the geometry of these curves, although it may not exactly be the matrix of periods. This would in fact be the interesting situation, as the corresponding theta function is then associated to an abelian variety which may not be a Jacobian. Whether genuinely new identities for theta functions of certain abelian varieties, i.e., for which a proof by direct algebraic methods is not easily available, as consequences of the integrability of probabilistic models can be obtained is a question left to future investigations.

\section[beta-ensembles and their properties]{$\boldsymbol{\beta}$-ensembles and their properties}

We recall a few facts about the $\beta$-ensembles and review the
determinantal and Pfaffian formulae of Borodin and Strahov \cite{borodin_averages}. We use the notation $[g]$ for the integer set $\{1,\ldots,g\}$.

\subsection{The unconstrained model}

Fix a finite union $A$ of compact intervals of $\mathbb{R}$, a positive integer $N$, a real number
$\beta > 0$, and an even-degree polynomial $V$ (the potential) with real coefficients and positive top coefficient. We consider the probability measure
$\mathbb{P}_{N}^{V}$ on $A^{N}$ defined by
\begin{equation}\label{eq:measure-free}
 \dd \mathbb{P}_{N}^V(\bm{\lambda}) = \frac{1}{Z_{N}^V}\,|\Delta(\bm{\lambda})|^{\beta}\ee^{-\frac{\beta N}{2}\sum_{i = 1}^N V(\lambda_{i})}\prod_{i=1}^{N}\ind_{A}(\lambda_{i})\dd\lambda_{i},
\end{equation}
where
$\bm{\lambda} = (\lambda_{1}, \ldots, \lambda_{N}) \in A^{N}$,
$\Delta(\bm{\lambda}) = \prod_{i < j}(\lambda_{j} - \lambda_{i})$ is
the Vandermonde determinant, and
\begin{equation*}
 Z_{N}^V = \int_{A^{N}} |\Delta(\bm{\lambda})|^{\beta}\exp\Biggl(-\frac{\beta N}{2}\sum_{i = 1}^N V(\lambda_{i})\Biggr) \prod_{i=1}^{N}\dd\lambda_{i}
\end{equation*}
is the partition function. Many results that we quote are formulated with $A = \mathbb{R}$ but their validity trivially extends to the case of $A$ compact. We choose to work from the start with a~compact $A$ as it facilitates the statement of the asymptotic results we will need and does not lead to any loss of generality.

 When $\beta = 1, 2 \text{ or }4$,
 $\mathbb{P}_{N}^{V}$ is the distribution of the $N$ eigenvalues of
 a random matrix whose law~is proportional to \smash{$\exp\bigl(-\frac{\beta N}{2}V(\bm{M})\bigr)\dd \bm{M}$},
  where $\bm{M}$ is a matrix which is real symmetric ($\beta = 1$), Hermitian ($\beta = 2$) or
 quaternionic self-dual ($\beta = 4$) and which is conditioned to have spectrum in $A$. The measure $\dd \bm{M}$ is
 the product of Lebesgue measure on the $\mathbb{R}$-linearly independent entries of $\bm{M}$. In particular, when $V(\bm{M}) = \frac{1}{2}\bm{M}^2$, the entries $M_{ij}$, $i \leq j$
 of the matrix are independent Gaussian random variables. These matrix ensembles
are known under the name of Gaussian orthogonal ensemble
 (GOE) for $\beta = 1$, Gaussian unitary ensemble (GUE) for
 $\beta = 2$, and Gaussian symplectic ensemble (GSE) for $\beta = 4$,
 see \cite{Mehtabook}. The $\beta$-ensembles \eqref{eq:measure-free} constitute a
 generalisation of these models.

\subsection{The model with fixed filling fractions}

Let us write $A = \bigsqcup_{h=0}^{g} A_{h}$ where $A_h$ are the connected components of $A$. In addition to the measure \eqref{eq:measure-free}, we define the
$\beta$-ensemble with fixed filling fractions as follows. Let
$\bm{N} = (N_{h})_{h = 1}^{g} \in \NN^{g}$ such that $N_1 + \cdots + N_g < N$ and introduce $N_0 \in \mathbb{Z}_{> 0}$ such that
\[
 N_{0} + \cdots + N_{g} = N.
\]
We call $N_h/N$ the filling fraction of $A_h$. We define the measure with fixed filling fractions by
\begin{equation*}
 \dd\mathbb{P}_{N,\bm{N}/N}^V = \frac{1}{Z^{V}_{N,\bm{N}/N}}|\Delta(\bm{\lambda})|^{\beta}\exp\Biggl(-\frac{\beta N}{2}\sum_{h = 0}^{g} \sum_{i = 1}^{N_h} V(\lambda_{h, i})\Biggr)\prod_{h=0}^{g}\prod_{i=1}^{N_h}\ind_{A_{h}}(\lambda_{h, i})\dd\lambda_{h, i},
\end{equation*}
where
$\bm{\lambda} = (\lambda_{h, 1}, \ldots, \lambda_{h, N_{h}})_{h = 0}^{g}$ is a $N$-tuple and
\begin{equation*}
 Z^{V}_{N,\bm{N}/N} = \int_{A^{N}} |\Delta(\bm{\lambda})|^{\beta}\ee^{-\frac{\beta N}{2}\sum_{h = 0}^{g} \sum_{i = 1}^{N_h} V(\lambda_{h, i})}\prod_{h=0}^{g}\prod_{i=1}^{N_h}\ind_{A_{h}}(\lambda_{h, i})\dd\lambda_{h, i}
\end{equation*}
is the partition function for fixed filling fractions. To distinguish it from the model with fixed filling fractions, we refer to \eqref{eq:measure-free} as the unconstrained model.

\subsection{Equilibrium measures and their Stieltjes transform}
\label{eqmessec}
We define the empirical measure as $L_N = \frac{1}{N}\sum_{i=1}^{N}\delta_{\lambda_{i}}$. It belongs to the space of probability measures on $A$, which we equip with the weak topology. We first consider $L_N$ in the original model. The following result comes large deviation arguments \cite[Theorem~2.6.1 and Corollary~2.6.3]{BookAG}, but see also \cite{APS01,Deift,Johan}.
\begin{Theorem}\label{thm:equilibrium}
 Assume that $V$ is an even-degree real polynomial with positive top coefficient. As $N \rightarrow \infty$, $L_{N}$
 converges under $\mathbb{P}_{N}^{V}$ almost surely, and in
 expectation $($when tested against continuous bounded functions$)$ to the unique probability measure
 $\mu_{{\rm eq}}$ on $A$ maximising
 \begin{equation}\label{energy}
 \mathcal{E}[\mu] = \frac{\beta}{2}\iint_{A^2} \biggl( \ln|\xi - \eta| -\frac{V(\xi) + V(\eta)}{2}\biggr)\dd\mu(\xi)\dd\mu(\eta).
 \end{equation}
 Furthermore, $\mu_{{\rm eq}}$ has compact
 support $S$ consisting in a finite union of segments. It is characterised by the existence of a constant $c$ such
 that
 \begin{equation}\label{eq:charc-eq}
	\forall x \in A \qquad 2\int_{A}\ln|x - \xi|\dd\mu_{{\rm eq}}(\xi) - V(x) \leq c
 \end{equation}
 with equality $\mu_{{\rm eq}}$-almost everywhere.
\end{Theorem}

We will only need to consider $S = \bigsqcup_{h = 0}^{g} S_h$ where $S_h$ is a segment contained in the interior\footnote{This is usually called the `soft edge' case, by opposition to hard edges that are endpoints of $S$ in the boundary of~$A$.}~$\mathring{A}_h$ of~$A_h$.
 Without loss of generality one can and one will restrict $A_h$ to be a small enlargement of~$S_h$. The choice of this enlargement will be irrelevant for our purposes, as it does not change the equilibrium measure and only affects the model by corrections which are exponentially small in~$N$, see, e.g., \cite[Proposition~2]{APS01} or the discussion in \cite[Section~2]{BG11} and references therein. In~the model with fixed filling fractions, Theorem~\ref{thm:equilibrium} has the following adaptation.

\begin{Theorem}[{\cite[Theorem 1.2]{BGmulti}}]
 \label{mueqfixed} Consider a sequence indexed by $N$ of $g$-tuples of nonnegative
 integers $\bm{N} = (N_{1}, \ldots, N_{g})$ with $\sum_{h=1}^{g}N_{h} < N$ and
 assume there exists
 $\bm{\epsilon} = (\epsilon_{h})_{1 = 0}^{g}$ such that
 $N_{h}/N \to \epsilon_{h}$ for all $h \in [g]$.
 Then, $L_{N} = \frac{1}{N} \sum_{h = 0}^{g} \sum_{i=1}^{N_{h}}\delta_{\lambda_{h, i}}$
 converges almost surely and in expectation under $\mathbb{P}^{V}_{N,\bm{N}/N}$ towards a deterministic probability measure $\mu_{{\rm eq},\bm{\epsilon}}$, which is the maximiser of \eqref{energy} among probability measures giving mass $\epsilon_h$ to the segment $A_h$ for each $h \in [0,g]$. It is characterised by the existence of constants $(c_h)_{h = 0}^{g}$ such that
\[
\forall h \in [0,g], \quad \forall x \in A_h \qquad 2\int_{A}\ln|x - \xi|\dd\mu_{{\rm eq},\bm{\epsilon}}(\xi) - V(x) \leq c_h
\]
with equality $\mu_{{\rm eq},\bm{\epsilon}}|_{A_h}$-almost everywhere.
\end{Theorem}

The filling fractions at equilibrium
$\bm{\epsilon}^{*} = (\epsilon_{h}^{*})_{h=1}^{g}$ are defined as $\epsilon_{h}^{*} = \mu_{{\rm eq}}(A_h)$, and one can show that $\mu_{{\rm eq}} = \mu_{{\rm eq}, \bm{\epsilon}^{*}}$, see \cite[Section~1.4]{BGmulti}.

Let us now discuss the properties of the equilibrium measure, both in the unconstrained case (Theorem~\ref{thm:equilibrium}) or fixed filling fraction case (Theorem~\ref{mueqfixed}) . We introduce the Stieltjes transform of the equilibrium measure
\[
 W_{1}(x) = \int_{A}\frac{\dd \mu_{{\rm eq}}(\xi)}{x - \xi},
\]
defined for $x \in \C \setminus S$. In \cite{Johan}, Johansson introduces the polynomial
\[
 P(x) = \int_{A}\frac{V'(x) - V'(\xi)}{x - \xi}\,\dd \mu_{{\rm eq}}(\xi),
\]
and derives the equation
\begin{equation}\label{eq:variational}
 W_{1}(x)^{2} - V'(x)W_{1}(x) + P(x) = 0
\end{equation}
for all $x \in \C\setminus S$. This equation is the large $N$ limit of the first
Dyson--Schwinger equation of the model, and its origin can be traced back to
\cite{BIPZ,migdal_loop_1983}. In particular, it implies that
\begin{equation*}
 W_{1}(x) = \frac{V'(x)}{2} \pm \frac{\sqrt{V'(x)^{2} - 4P(x)}}{2}.
\end{equation*}

The determination of the squareroot should be chosen such that $W_1(x) \sim \frac{1}{x}$ as $x \rightarrow \infty$ and~$W_1$ is holomorphic in $\mathbb{C} \setminus S$. As this determination plays an important role in our discussion, it is worth reviewing in detail how this can be achieved. The standard determination of the squareroot gives a holomorphic function $x \mapsto \sqrt{x}$ on $\mathbb{C} \setminus \mathbb{R}_{\leq 0}$ such that $\sqrt{\mathbb{R}_{> 0}} = \mathbb{R}_{> 0}$ and $\bigl(\sqrt{x}\bigr)^2 = x$. We decompose $V'(x)^2 - 4P(x) = M(x)^2 \sigma(x)$, where $\sigma$ is a monic polynomial with simple real roots and $M$ is a real polynomial with positive top coefficient. We write further
\[
\sigma(x) = \prod_{h = 0}^{g} (x - a_h)(x - b_h)
\]
with
\[
a_0 < b_0 < a_1 < b_1 < \cdots < a_g < b_g.
\]
The locus
$\sigma^{-1}(\mathbb{C} \setminus \mathbb{R}_{\leq 0})$ is a union of
$g + 2$ connected components, labelled from left to right: $C_0$~contains $a_0$ in its closure, $C_h$ contains $b_{h}$
and $a_{h + 1}$ in its closure for $h \in [g - 1]$, and $C_{g + 1}$
contains $b_g$ in its closure. For $x \in C_{h}$ with $h \in [0,g+1]$, we set $s(x) = (-1)^{g + 1 -h}\sqrt{\sigma(x)}$. This
definition makes $s(x)$ a continuous (thus holomorphic) function of
$x \in \mathbb{C} \setminus S$. It is discontinuous on $(a_h,b_h)$ because by
crossing this segment we stay in the same component, so
we keep the same global sign coming from the component we are in
while the standard determination of the squareroot does get a sign
change since $\sigma(x)$ crosses $\mathbb{R}_{< 0}$. Then, $M(x)s(x)$ is a holomorphic function of $x \in \mathbb{C}\setminus S$, and $M(x)s(x) \sim t x^{d - 1}$ for some $d \geq 2$ and $t > 0$. The constraint $W_1(x) \sim \frac{1}{x}$ as $x \rightarrow \infty$ leads to the formula
\begin{equation}
\label{W1SM}W_1(x) = \frac{V'(x) - M(x)s(x)}{2}.
\end{equation}

The fact that $W_1$ is the Stieltjes transform of the equilibrium measure puts some constraints on the polynomial $M$.

\begin{Lemma}
\label{lemMmm} The support of $\mu_{{\rm eq}}$ is $S = \bigsqcup_{h = 0}^{g} S_h$ with $S_h = [a_h,b_h]$ and we have
\[
\frac{\dd \mu_{{\rm eq}}}{\dd x} = \frac{M(x)  \operatorname{Im}(s(x + {\rm i}0))}{2\pi} \ind_S(x).
\]
For each $h \in [g]$, the number of zeros $($with multiplicity$)$ of $M$ in $[b_{h - 1},a_h]$ is odd. For each $h \in [0,g]$, the zeros of $M$ in $(a_h,b_h)$ have even multiplicity $($if there is any$)$.
\end{Lemma}

\begin{proof}
By construction, for any $h \in [0,g]$ and $x \in (a_h,b_h)$ we have $s(x + {\rm i}0) \in (-1)^{g -h}{\rm i}\mathbb{R}_{> 0}$ and for any $h \in [0,g + 1]$ and $x \in (b_{h - 1},a_h)$ we have $s(x) \in (-1)^{g + 1 - h}\mathbb{R}_{> 0}$, with the conventions $b_{-1} = -\infty$ and $a_{g + 1} = +\infty$. By definition of the Stieltjes transform, the function $W_1(x)$ has a discontinuity in the interior of the support of $\mu_{{\rm eq}}$. It is identified as the \emph{real} locus where the polynomial $V'(x)^2 - 4P(x)$ takes nonpositive values, and thus coincides with $S = \bigsqcup_{h = 0}^{g} [a_h,b_h]$. The density of the equilibrium measure is reconstructed from the jump:
\begin{gather*}
\forall x \in \mathbb{R} \qquad \frac{\dd\mu_{{\rm eq}}}{\dd x} = \frac{W_1(x - {\rm i}0) - W_1(x + {\rm i}0)}{2{\rm i} \pi} = \frac{M(x)s(x + {\rm i}0)}{2{\rm i}\pi}\ind_{S}(x) \\
\hphantom{\forall x \in \mathbb{R}\qquad \frac{\dd\mu_{{\rm eq}}}{\dd x}}{}
 = \frac{M(x) \operatorname{Im}(s(x + {\rm i}0))}{2\pi}\ind_S(x).
\end{gather*}
Since $\mu_{{\rm eq}}$ is a positive measure and since $\operatorname{Im}(s(x + {\rm i}0))$ has a constant sign in each $(a_h,b_h)$, $M$ should have constant sign in $S_h$ and thus have zeros of even multiplicity there (if there is any). Likewise, since the sign of $\operatorname{Im}(s(x + {\rm i}0))$ changes between two consecutive segments in the support, $M$ should have at least a sign change in the closure of the interval between these two segments, hence an odd number of zeros.
\end{proof}

Lemma~\ref{lemMmm} allows us to give an expression for the
density $\rho$ of the equilibrium measure. Indeed,
\begin{align*}
 \rho(x) & = \frac{W_{1}(x -{\rm i}0) - W_{1}(x + {\rm i}0)}{2{\rm i}\pi} = M(x)\frac{s(x+{\rm i}0)-s(x-{\rm i}0)}{4\pi}\ind_{S}(x) \\
 & = \frac{\sqrt{-M(x)\sigma(x)}}{2\pi}\ind_{S}(x).
\end{align*}

\begin{Definition}
The effective potential is defined for $x \in A$ by
\[
U(x) := V(x) - 2 \int_{A} \ln|x - \xi|\dd \mu_{{\rm eq}}(\xi).
\]
 It satisfies
\begin{gather}
\forall x \in A \setminus S  \qquad U'(x) = V'(x) - 2W_1(x) = M(x)s(x), \nonumber\\
\forall x \in \mathring{S}  \qquad U'(x) = V'(x) - W_1(x + {\rm i}0) - W_1(x - {\rm i}0) = 0.\label{effpot}
\end{gather}
\end{Definition}
\begin{Remark}
\label{Rem:vanishint} For the equilibrium measure in the unconstrained case, Theorem~\ref{thm:equilibrium} says there exists a constant $c$ such that $U(x) = c$ for all $x \in S$. In view of \eqref{effpot}, the latter property is equivalent to
\begin{equation*}
\forall h \in [g]\qquad \int_{b_{h - 1}}^{a_h} M(x)s(x)\,\dd x = 0.
\end{equation*}
\end{Remark}

\subsection{Determinantal and Pfaffian
 formulae}

Expectation values of ratios of
characteristic polynomials, also called kernels, are quantities of interest in random matrix theory. Let us introduce the notation
 $\mean{\cdot}^{V}_{N}$ for the expectation value with respect to $\mathbb{P}_{N}^V$ (the value of $\beta$ will be specified in each case), and $\Lambda = \operatorname{diag}(\bm{\lambda})$. Given $c_1,\ldots,c_m \in \mathbb{Z}$, and $x_1,\ldots,x_m \in \mathbb{C}$ with the condition $x_j \notin A$ if $c_j < 0$, the $m$-point kernel is defined as
\[
 \Biggl\langle \prod_{j=1}^{m} \det (x_{j} - \Lambda)^{c_j}\Biggr\rangle_{N}^{V} = \Biggl\langle \prod_{j=1}^{m} \prod_{i = 1}^{N} (x_{j} - \lambda_i)^{c_j}\Biggr\rangle_{N}^{V}.
\]

In \cite{borodin_averages}, Borodin and Strahov derive formulae
to compute the kernels. In what follows, we will always consider the ``balanced'' case, that
is, when there are as many characteristic polynomials in the numerator
as in the denominator. Given two tuples of complex numbers $\bm{x} = (x_{1}, \ldots, x_{m_1})$ and
$\bm{\tilde{x}} = (\tilde{x}_{1}, \ldots, \tilde{x}_{m_2})$, we write
\begin{equation*}
 \Delta(\bm{x},\bm{\tilde{x}}) = \prod_{i=1}^{m_1}\prod_{j=1}^{m_2} (x_{i} - \tilde{x}_{j}).
\end{equation*}

\subsubsection[The determinantal case: beta = 2]{The determinantal case: $\boldsymbol{\beta = 2}$}
In that case, we have the following formulae.
\begin{Theorem}[{\cite[Theorem 4.1.1]{borodin_averages}}]
 \label{beta2form} Let $N$, $m_1$, $m_2$ be positive integers, and sets of complex numbers
 \begin{alignat*}{3}
 &\bm{x} = \{x_{1}, \ldots, x_{m_{1}}\},\qquad && \bm{x}' = \bigl\{x'_{1}, \dots, x'_{m_{1}}\bigr\}, &\\
 &\tilde{\bm{x}} = \{\tilde{x}_{1}, \ldots, \tilde{x}_{m_{2}}\},\qquad && \tilde{\bm{x}}' = \bigl\{\tilde{x}'_{1}, \dots, \tilde{x}'_{m_{2}}\bigr\}, &
 \end{alignat*}
 such that
 \begin{equation*}
 \bm{x} \cap \bm{x}' = \varnothing,\qquad\bm{\tilde{x}} \cap \bm{\tilde{x}}'=\varnothing,\qquad\bm{x}' \cap A = \varnothing,\qquad \bm{\tilde{x}}' \cap A = \varnothing.
 \end{equation*}
 We have
 \begin{gather*}
 \Biggl\langle\prod_{j=1}^{m_{1}} \frac{\det(x_{j} - \Lambda)}{\det\bigl(x'_{j} - \Lambda\bigr)}\prod_{j=1}^{m_{2}}\frac{\det(\tilde{x}_{j} - \Lambda)}{\det\bigl(\tilde{x}'_{j} - \Lambda\bigr)}\Biggr\rangle_{N}^{V}\\
 \qquad  = (-1)^{\frac{1}{2}((m_{1}+m_{2})^{2} + m_{2} - m_{1})}\frac{\Delta(\bm{x}, \bm{x}')}{\Delta(\bm{x})\Delta(\bm{x}')}\frac{\Delta(\tilde{\bm{x}},\tilde{\bm{x}}')}{\Delta(\tilde{\bm{x}})\Delta(\tilde{\bm{x}}')}
 \det\bigl(\mathcal{M}^{(2)}(\bm{x},\bm{x}';\tilde{\bm{x}},\tilde{\bm{x}}')\bigr),
 \end{gather*}
 in terms of the block matrix of size $(m_1 + m_2)$:
 \[
 \mathcal{M}^{(2)}(\bm{x},\bm{x}';\tilde{\bm{x}},\tilde{\bm{x}}') = \begin{pmatrix} \mathcal{M}^{(2)}_{++}(x_i,\tilde{x}_j) & \mathcal{M}^{(2)}_{+-}\big(x_i,x_j'\big) \\ \mathcal{M}_{-+}^{(2)}(\tilde{x}_i',\tilde{x}_j) & \mathcal{M}^{(2)}_{--}\big(\tilde{x}_i',x_j'\big)\end{pmatrix},
 \]
where $i$ is a row index, $j$ a column index, and the entries are
\begin{gather*}
 \mathcal{M}_{++}^{(2)}(x, \tilde{x}) = N\frac{Z_{N-1}^{\frac{N}{N-1}V}}{Z_{N}^{V}}\mean{\det(x - \Lambda)\det(\tilde{x} - \Lambda)}_{N-1}^{\frac{N}{N-1}V},\\
 \mathcal{M}^{(2)}_{+-}(x, x') = \frac{1}{x-x'}\mean{\frac{\det(x - \Lambda)}{\det(x' - \Lambda)}}_{N}^{V},\\
 \mathcal{M}^{(2)}_{-+}(\tilde{x}',\tilde{x}) = \frac{1}{\tilde{x}' - \tilde{x}}\mean{\frac{\det(\tilde{x} - \Lambda)}{\det(\tilde{x}' - \Lambda)}}_{N}^{V},\\
 \mathcal{M}^{(2)}_{--}(\tilde{x}', x') = \frac{1}{N+1}\frac{Z_{N+1}^{\frac{N}{N+1}V}}{Z_{N}^{V}}\mean{\frac{1}{\det(\tilde{x}' - \Lambda)\det(x' - \Lambda)}}_{N+1}^{\frac{N}{N+1}V}.
 \end{gather*}
\end{Theorem}

\subsubsection[Orthogonal ensembles: beta = 1]{Orthogonal ensembles: $\boldsymbol{\beta = 1}$}

Recall that for an antisymmetric matrix $A$ of size $2m$, the Pfaffian is defined as
\[
\pf(A) = \frac{1}{m! 2^m} \sum_{\sigma \in \mathfrak{S}_{2m}} \operatorname{sgn}(\sigma) \prod_{i = 1}^{m} A_{\sigma(2i - 1),\sigma(2i)}.
\]

\begin{Theorem}[{\cite[Theorem 1.2.1]{borodin_averages}}]
 \label{th:idbeta1} Let $N$, $m$ be positive integers and set of complex numbers
\[
\bm{x} = \{x_{1}, \ldots, x_{m}\},\qquad \bm{x}' = \{x'_{1}, \ldots,x_{m}'\},
\]
such that $\bm{x}' \cap A = \varnothing$. We have
 \begin{equation*}
 \Biggl\langle\prod_{j=1}^{m}\frac{\det(x_{j} - \Lambda)}{{\det} (x_{j}' - \Lambda)}\Biggr\rangle_{2N}^{V}
 = \frac{\Delta(\bm{x},\bm{x}')}{\Delta(\bm{x})\Delta(\bm{x}')}\pf\bigl(\mathcal{M}^{(1)}(\bm{x},\bm{x}')\bigr),
 \end{equation*}
 in terms of the antisymmetric matrix of size $2m$
 \begin{equation}
 \label{blockm1}
 \mathcal{M}^{(1)}(\bm{x},\bm{x}') = \begin{pmatrix} \mathcal{M}^{(1)}_{++}(x_i,x_j) & \mathcal{M}^{(1)}_{+-}\big(x_i,x_j'\big) \\ \mathcal{M}^{(1)}_{-+}(x_i',x_j) & \mathcal{M}^{(1)}_{--}\big(x_i',x_j'\big) \end{pmatrix},
\end{equation}
 with entries
 \begin{gather*}
 \mathcal{M}^{(1)}_{++}(x,\tilde{x}) = (2N -1)2N(x - \tilde{x})\frac{Z^{\frac{2N}{2N-2}V}_{2N-2}}{Z^{V}_{2N}}\mean{\det(x - \Lambda)\det(\tilde{x}- \Lambda)}^{\frac{2N}{2N - 2}V}_{2N-2},\\
 \mathcal{M}^{(1)}_{+-}(x,x') = \frac{1}{x-x'}\mean{\frac{\det(x - \Lambda)}{\det(x' - \Lambda)}}^{V}_{2N} = - \mathcal{M}^{(1)}_{-+}(x',x), \\
 \mathcal{M}_{--}^{(1)}(x',\tilde{x}') = \frac{x'- \tilde{x}'}{(2N+1)(2N+2)}\frac{Z^{\frac{2N}{2N + 2}V}_{2N+2}}{Z^{V}_{2N}} \mean{\frac{1}{\det(x' - \Lambda)\det(\tilde{x}' - \Lambda)}}^{\frac{2N}{2N+2}V}_{2N+2}.
 \end{gather*}
\end{Theorem}

\subsubsection[Symplectic ensembles: beta = 4]{Symplectic ensembles: $\boldsymbol{\beta = 4}$}
The case $\beta = 4$ is very similar to the case $\beta = 1$.

\begin{Theorem}[{\cite[Theorem 1.2.1]{borodin_averages}}]
 \label{th:idbeta4} Let $N$, $m$ be positive integers and two sets of complex numbers
\[
\bm{x} = \{x_{1}, \ldots, x_{m}\},\qquad \bm{x}' = \{x'_{1}, \ldots, x'_{m}\},
\]
 such that $\bm{x}' \cap A= \varnothing$. We have
 \begin{equation*}
 \Biggl\langle\prod_{j=1}^{m}\frac{\det(x_{j} - \Lambda)^{2}}{\det\big(x'_{j} - \Lambda\big)^{2}}\Biggr\rangle_{N}^{V}
 = \frac{\Delta(\bm{x},\bm{x}')}{\Delta(\bm{x})\Delta(\bm{x}')}\pf\bigl(\mathcal{M}^{(4)}(\bm{x},\bm{x}')\bigr),
 \end{equation*}
 with the same block structure as \eqref{blockm1} but entries
 \begin{gather*}
 \mathcal{M}^{(4)}_{++}(x,\tilde{x}) = N\frac{Z^{\frac{N}{N-1}V}_{N-1}}{Z^{V}_{N}}(x- \tilde{x})\bigl\langle\det(x - \Lambda)^{2}\det(\tilde{x} - \Lambda)^{2}\bigr\rangle^{\frac{N}{N - 1}V}_{N-1},\\
 \mathcal{M}^{(4)}_{+-}(x,x')  = \frac{1}{x-x'}\Biggl\langle\frac{\det(x - \Lambda)^{2}}{\det(x' - \Lambda)^{2}}\Biggr\rangle^{V}_{N} = -\mathcal{M}^{(4)}_{-+}(x',x), \\
 \mathcal{M}^{(4)}_{--}(x',\tilde{x}') = \frac{1}{N+1}\frac{Z^{\frac{N}{N+1}V}_{N+1}}{Z^{V}_{N}}(x'- \tilde{x}')\mean{\frac{1}{\det(x' - \Lambda)^{2}\det(\tilde{x}' - \Lambda)^{2}}}^{\frac{N}{N+1}V}_{N}.
 \end{gather*}
\end{Theorem}

\section{Geometry of the spectral curves}

This section collects the information on theta functions and geometry of the spectral curve that will be needed later to present the large $N$ asymptotics in the $\beta$-ensembles. We only give the details necessary to understand the formulae of Sections~\ref{sec:expans-part} and \ref{sec:derivation-formulae} in a self-contained way. We refer to the many textbooks address in details theta functions and the geometry of Riemann surfaces, for instance \cite{Bertolac,Farkas-Kra, Fay,Tata1}.

\subsection{Theta functions}
\label{sec:theta-function}

Let us recall the definition and properties of the theta function.
\begin{Definition}
 Let $\bm{\tau}$ be a complex $g \times g$ symmetric matrix such that $\Im \boldsymbol{\tau}$ is positive definite. The
 theta function with characteristics
 $\boldsymbol{\mu},\boldsymbol{\nu} \in \mathbb{R}^{g}$ is the function defined by
\[
 \forall \boldsymbol{z} \in \mathbb{C}^{g} \qquad \vartheta_{\bm{\mu},\bm{\nu}}(\bm{z}|\boldsymbol{\tau}) = \sum_{\bm{n}\in\ZZ^{g}}\exp({\rm i}\pi(\bm{n}+\bm{\mu})\cdot \bm{\tau}(\bm{n}+\bm{\mu}) + 2{\rm i}\pi(\bm{n}+\bm{\mu})\cdot(\bm{z} + \bm{\nu})).
\]
We set $\theta := \vartheta_{\boldsymbol{0},\boldsymbol{0}}$.
\end{Definition}
The condition $\Im\bm{\tau} > 0$ ensures that the function is well
defined. Let us define the period lattice associated to $\bm{\tau}$ as $\mathbb{L} = \ZZ^{g} \oplus \bm{\tau}(\ZZ^{g})$. The theta function is quasi-periodic: for $\bm{m}, \bm{n}\in \ZZ^{g}$ we have
$\bm{m} + \bm{\tau}(\bm{n})\in\mathbb{L}$, and for any $\bm{z} \in \CC^{g}$
\begin{equation*}
 \vartheta_{\bm{\mu},\bm{\nu}} (\bm{z} + \bm{m} + \bm{\tau}(\bm{n}) |\bm{\tau} ) ={\rm e}^{2{\rm i}\pi \bm{m}\cdot \bm{\mu} - {\rm i}\pi\bm{n} \cdot (\bm{\tau}(\bm{n}) + 2\bm{z} + 2\bm{\nu})}\vartheta_{\bm{\mu},\bm{\nu}}(\bm{z}|\bm{\tau}).
\end{equation*}
\begin{Definition}
An odd half-integer characteristic is $\bm{c} = \frac{1}{2}\bm{e} + \frac{1}{2}\bm{\tau}(\bm{e}')$, with $\bm{e},\bm{e}' \in \ZZ^{g}$ such that $\bm{e} \cdot \bm{e}' \in 2\mathbb{Z} + 1$.
 \end{Definition}
 By direct computation, if $\bm{c}$ is a odd half-integer characteristic, then $\theta(\bm{c}|\bm{\tau}) = 0$.

\subsection{Geometry of Riemann surfaces}

\subsubsection{Basis of cycles and forms}
\label{sec:base-cycles-forms}
Let $\widehat{C}$ be a compact Riemann surface of genus $g$. The first
homology group $H_{1}(\widehat{C}; \mathbb{Z})$ has a basis $(\A_{h}, \B_{h})_{h = 1}^{g}$ which can be chosen to have the following properties under the intersection pairing
\begin{equation*}
 \forall h,k \in [g]\qquad \A_{h}\cap\A_{k} = 0,\qquad \B_{h}\cap\B_{k} = 0,\qquad \A_{h}\cap \B_{k} = \delta_{h,k}.
\end{equation*}
Such a basis is called a symplectic basis of homology, and $\widehat{C}$ equipped with such a basis is called a marked Riemann surface. We can for instance choose a point $p_{0}$ and simple closed curves on $\widehat{C}$ representing the $2g$ classes $(\A_{h}, \B_{h})_{h = 1}^g$
such that all the curves intersect each other at $p_{0}$ only. For spectral curves of $\beta$-ensembles, we will later work with another set of representatives (Section~\ref{sec:constr-spectr-curve}). We keep the same notation for homology classes and their representatives. The
surface $\widehat{C}^{0} = \widehat{C} \setminus \bigcup_{h=1}^{g}(\A_{h}\cup\B_{h})$ is
then simply-connected. The $\A$-cycles $(\A_h)_{h = 1}^{g}$ determine a~dual basis of holomorphic 1-forms
$(\dd u_{h})_{h = 1}^{g}$, such that
\begin{equation*}
 \begin{split}
\forall h,k \in [g]\qquad \oint_{\A_{h}}\dd u_{k} = \delta_{h,k}.
 \end{split}
\end{equation*}
The matrix of periods $\bm{\tau}$ is then defined by
\begin{equation}
\label{tauhk} \forall h,k \in [g] \qquad \tau_{h,k} = \oint_{\B_{h}}\dd u_{k}.
\end{equation}
It is symmetric and $\operatorname{Im}(\bm{\tau})$ is definite positive, in particular we can consider the theta function with matrix $\frac{\beta}{2} \bm{\tau}$ for any $\beta > 0$. The theta function with matrix equal to \eqref{tauhk} is called the Riemann theta function.

\subsubsection{Abel map}
\label{SecAbel} With the 1-forms $(\dd u_{h})_{h = 1}^{g}$ defined in Section
\ref{sec:base-cycles-forms}, we can introduce the Abel map.
\begin{Definition}
 Choose a base point $p_{0}$ in $\widehat{C}$. The Abel map
 $\bm{u} \colon \widehat{C}^{0} \to \C^{g}$ is defined by
 \begin{equation*}
 \begin{split}
 u_{i}(z) &= \int_{p_{0}}^{z}\dd u_{i},
 \end{split}
 \end{equation*}
 where the path of integration is in $\widehat{C}^0$.
 \end{Definition}

The definition of the Abel map depends on a choice of base point $p_{0}$.
However, we will often consider differences $\bm{u}(z) - \bm{u}(w)$ of Abel
maps, which are independent of $p_{0}$. Depending on the context, we may also consider the Abel map as a map $\bm{u} \colon \widetilde{C} \rightarrow \mathbb{C}^g$ defined on the universal cover~$\widetilde{C}$ of $\widehat{C}$ based at $p_0$. We say that $\boldsymbol{c}$ is non-singular if $\theta(\bm{c} + \bm{u}(z) - \bm{u}(w)|\bm{\tau})$ is not identically~$0$ when $z,w \in \widetilde{C}$. Non-singular odd half-integer characteristics exist, and in what follows we fix one. \looseness=-1

\subsubsection{Prime form}

We introduce the holomorphic 1-form
\begin{equation*}
 \omega_{\bm{c}} = \sum_{h=1}^{g}\partial_{z_{h}}\theta(\bm{z}|\bm{\tau})\big|_{\bm{z} = \bm{c}}\,\dd u_{h}.
\end{equation*}
The prime form is
\begin{equation}
\label{Eprimtheta}
E(z_1,z_2) = \frac{\theta(\bm{c} + \bm{u}(z_1) - \bm{u}(z_2)|\bm{\tau})}{\sqrt{\omega_{\bm{c}}(z_1)}\sqrt{\omega_{\bm{c}}(z_2)}}.
\end{equation}
It is defined as a holomorphic bispinor on $\widetilde{C} \times \widetilde{C}$, i.e., a $\bigl(-\frac{1}{2}\bigr) \otimes \bigl(-\frac{1}{2}\bigr)$ form. It has zeros only at $z_1 = z_2$, and in local coordinates $\zeta$, we have
\[
E(z_1,z_2) \mathop{\sim}_{z_1 \rightarrow z_2} \frac{\zeta(z_1) - \zeta(z_2)}{\sqrt{\dd \zeta(z_1) \dd \zeta(z_2)}}.
\]
The prime form on the Riemann sphere $\widehat{\mathbb{C}}$ reads
\[
E_0(x_1,x_2) = \frac{x_1 - x_2}{\sqrt{\dd x_1 \dd x_2}}.
\]
Given a meromorphic function $X \colon \widehat{C} \rightarrow \widehat{\mathbb{C}}$, we can define the relative prime form
\begin{equation*}
\tilde{E}(z_1,z_2) = \frac{E(z_1,z_2)}{E_0(X(z_1),X(z_2))}.
\end{equation*}
We observe that $\tilde{E}(z_1,z_2)$ is a function on $\widetilde{C} \times \widetilde{C}$, such that
\begin{equation}
\label{limEtilde} \lim_{z_2 \rightarrow z_1} \tilde{E}(z_1,z_2) = 1.
\end{equation}

\subsubsection{Fundamental bidifferential}
\begin{Definition}
 The fundamental bidifferential $B(z, w)$ is the unique bidifferential \big(i.e., a $1 \otimes 1$ form on \smash{$\widehat{C} \times \widehat{C}$\big)} such that
 \begin{enumerate}\itemsep=0pt
 \item Symmetry: $B(z, w) = B(w, z)$.
 \item Normalisation: $\forall h \in [g]$ $\oint_{\A_{h}}B(\cdot, w) = 0$.
 \item Singularities: $B(z, w)$ is meromorphic with only a double pole at $z = w$, and if $\zeta$ is a local coordinate, we have
\[
 B(z, w) \mathop{=}_{z \rightarrow w} \biggl(\frac{1}{(\zeta(z) - \zeta(w))^{2}} + S_{B,\zeta}(w) + \order{\zeta(z) - \zeta(w)}\biggr)\dd \zeta(z) \dd \zeta(w).
\]
for some function $S_{B,\zeta}$ locally defined on $\widehat{C}$.
 \end{enumerate}
\end{Definition}

The fundamental bidifferential can be expressed as
\begin{equation*}
 \begin{split}
 B(z,w) &= \dd_{z}\dd_{w}\ln \theta(\bm{c} + \bm{u}(z) - \bm{u}(w)),
 \end{split}
\end{equation*}
and this expression is independent on the choice of a non-singular odd half-integer characteristics~$\bm{c}$. Given $p,q \in \widehat{C}$ and a choice of path $\gamma_{p,q}$ from $p$ to $q$, we define the meromorphic~form
\begin{equation*}
 \begin{split}
 \dd S_{p,q}(z) = \int_{\gamma_{p,q}} B(z, \cdot).
 \end{split}
\end{equation*}
It has two poles of order 1 in $p$ and $q$, with respective residue
$-1$ and $+1$. Equivalently, we can consider that it is specified by the choice of two points $p$ and $q$ in the universal cover $\widetilde{C}$. The prime form appears in the following computation.
\begin{Lemma}
\label{IntB} For $i = 1,2$, let $z_i,\tilde{z}_i \in \widetilde{C}$ and $\gamma_i$ a path from $\tilde{z}_i$ to $z_i$. Then
\[
\int_{\gamma_1} \int_{\gamma_2} B = \int_{\gamma_1} \dd S_{\tilde{z}_2,z_2} = \ln\biggl(\frac{E(z_1,z_2)E(\tilde{z}_1,\tilde{z}_2)}{E(z_1,\tilde{z}_2)E(\tilde{z}_1,z_2)}\biggr).
\]
\end{Lemma}

\subsubsection{Decomposition of meromorphic forms}
\label{sec:decomp-merom-forms}

One distinguishes between three kinds of meromorphic forms:
\begin{itemize}\itemsep=0pt
 \item holomorphic 1-forms (first kind);
 \item meromorphic 1-forms with vanishing residues (second kind);
 \item meromorphic 1-form with non-vanishing residues (third kind).
\end{itemize}

The space of first kind differentials has for basis $(\dd u_{h})_{h = 1}^{g}$, which is dual to the basis of the $\A$-cycles.
Assume that a choice of local coordinate $\zeta_p$ near each point $p \in \widehat{C}$ has been made. A~basis of the space of second kind differentials is then given by
\begin{equation*}
 \dd B_{p,k}(z) = \mathop{\mathrm{Res}}_{z' = p}\zeta_{p}(z')^{-k}B(z', z)
\end{equation*}
for $p \in \widehat{C}$ and $k \in \mathbb{Z}_{> 0}$. Given the properties of the fundamental bidifferential, its only pole is at $p$ with order $(k + 1)$, where it behaves like
\[
\dd B_{p,k}(z) \mathop{=}_{z \rightarrow p} -\dd\bigl(\zeta_{p}(z)^{-k}\bigr) + \order{\dd \zeta_p(z)}.
\]
Besides, we have $\oint_{\A_h} \dd B_{p,k} = 0$ for any $h \in [g]$.

Assume that for each $p \in \widehat{C}^0$, a choice of path $\gamma_{p_0,p}$ from $p_0$ to $p$ in $\widehat{C}^0$ has been made. An example of third kind differential is
\begin{equation}
 \label{Spp} \dd S_{p_0,p}(z) = \int_{\gamma_{p_0,p}} B(z, \cdot).
\end{equation}
It has two poles of order 1 in $p_0$ and $p$, respectively of residue $-1$ and $+1$, and has zero $\A$-periods.

Every meromorphic $1$-form $\phi$ can be decomposed uniquely as a sum of first kind, second kind and third kind differentials:\footnote{Given $\phi$, we can always perturb the representatives of $\A$ and $\B$-cycles so that all poles are contained in $\widehat{C}^0$ and we can use \eqref{Spp}.}
\begin{gather*}
\phi(z) = \Biggl(\sum_{h = 1}^{g} \oint_{\A_h} \phi\Biggr) \dd u_h(z) + \sum_{\substack{p\ \text{simple pole} \\ \text{of}\ \phi}} \bigl(\mathop{\text{Res}}_{p} \phi\bigr) \dd S_{p_0,p}(z) \\
\hphantom{\phi(z) =}{}
 + \sum_{k \geq 1} \sum_{\substack{p\ \text{pole of} \\ \text{order}\ (k + 1)\ {\rm of}\ \phi}} \Bigl(\mathop{\mathrm{Res}}_{p} \zeta_p^{k} \phi \Bigr)\frac{\dd B_{p,k}(z)}{k}.
\end{gather*}

\subsection{The spectral curve}

We elaborate on Section~\ref{eqmessec} and construct the spectral curve associated to the equilibrium measure of ${\beta}$-ensembles, for the moment indifferently in the unconstrained case or the fixed filling fraction case. This prepares us for Section~\ref{sec:expans-part} where the asymptotics in the $\beta$-ensembles is described solely in terms of the geometry of this spectral curve.

\subsubsection{Construction of the marked Riemann surface}
\label{sec:constr-spectr-curve}

The equation \eqref{eq:variational} satisfied by the Stieltjes transform of the equilibrium measure has two solutions:
\begin{gather}
\nonumber
 F_{+}(x) = \frac{V'(x)}{2} + y(x) = W_{1}(x),\\
 F_{-}(x) = \frac{V'(x)}{2} - y(x) = V'(x) - W_{1}(x),\label{Fpm}
\end{gather}
where $y(x) = -\frac{1}{2}M(x)s(x)$ and $y(x)^2 = \frac{1}{4}V'(x)^2 - P(x)$. After the birational transformation $(x,y) \mapsto \bigl(x,s = -\frac{2y}{M(x)}\bigr)$, we have the equation of an hyperelliptic curve
\[
 s^{2} = \sigma(x) = \prod_{h = 0}^{g} (x - a_h)(x - b_h),
\]
where the Weierstra\ss{} points $a_0,b_0,\ldots,a_g,b_g$ are real. This curve is constructed from two sheets homeomorphic to
$\CC\setminus S$, which are glued together along $S$. The two sheets are embedded into the curve
by
\[
 \iota_{\pm }\colon \   \CC \setminus S   \longrightarrow   C \subset \CC \times \CC,\qquad  x   \longmapsto   (x,\pm s(x))  .
\]
We denote the sheets by $C_{\pm} = \iota_{\pm}(\C \setminus S)$, and $\widehat{C}_{\pm}$ are the sheets including their point at infinity. Adding these points at infinity $\infty_{\pm}$ to $C$, we get a compact Riemann surface $\widehat{C}$. We define the projection map as the meromorphic function
\[
 X \colon \  \widehat{C}   \longrightarrow   \widehat{\mathbb{C}},\qquad (x, s)   \longmapsto   x .
\]
This function has simple poles at $\infty_{\pm}$ and it defines a degree $2$ branched covering of the Riemann sphere~$\hat{\CC}$, whose branch points are the zeros of $\sigma$. It allows us to define local coordinates $\zeta_{p}$ around any point $p \in \widehat{C}$, which is used in Section~\ref{sec:decomp-merom-forms} to define a basis of meromorphic differentials:
\begin{itemize}\itemsep=0pt
 \item If $p$ is a ramification point, we take $\zeta_{p}(z) = \sqrt{X(z) - X(p)}$ for some choice of sign for the squareroot.
 \item If $p = \infty_{\pm}$, then $\zeta_{\infty_{\pm}} = X(z)^{-1}$.
 \item In all the other cases, $\zeta_{p}(z) = X(z) - X(p)$.
\end{itemize}
We shall take $p_0 = \infty_+$ as reference point for the definition of the Abel map (Section~\ref{SecAbel}). For later use, we analyse the prime form over $\widehat{\mathbb{C}}$ near $\infty_{\pm}$.
\begin{Lemma}
\label{LemE0in}We have
\[
E_0(z,\tilde{z})\sqrt{\dd \zeta_{\infty_\pm}(\tilde{z})}\Big|_{\tilde{z} = \infty_\pm} = \frac{-1}{\sqrt{-\dd X(z)}}.
\]
\end{Lemma}
\begin{proof}
Since $\dd X(\tilde{z}) = -X(\tilde{z})^2 \dd\zeta_{\infty_\pm}(\tilde{z})$, we have
\[
\lim_{\tilde{z} \rightarrow \infty_\pm} E_0(z,\tilde{z})\sqrt{\dd \zeta_{\infty_\pm}(\tilde{z})} = \lim_{\tilde{z} \rightarrow \infty_\pm} (X(z) - X(\tilde{z})) \frac{\sqrt{-X(\tilde{z})^{-2}\dd X(\tilde{z})}}{\sqrt{\dd X(z)\dd X(\tilde{z})}} = \frac{-1}{\sqrt{-\dd X(z)}}.
\tag*{\qed}
\]
\renewcommand{\qed}{}
\end{proof}

We choose representatives for a symplectic basis of homology on $\widehat{C}$ like in Figure~\ref{fig:cycles}. Namely, we take $\A_{h}$ representing a counterclockwise loop in $C_+$ going around the cut $S_h$, for $h \in [g]$. For convenience we fix a representative $\A_0$ of a counterclockwise loop surrounding $S_0$ in $C_+$, whose homology class is $-(\A_1 + \cdots + \A_g)$. We take $\B_{h}$ representing a loop in $\widehat{C}$ travelling from $S_{0}$
to~$S_{h}$ in $C_+$ and in the opposite direction in $C_-$.

\begin{figure}[ht] \centering
 \includegraphics[width=0.8\textwidth]{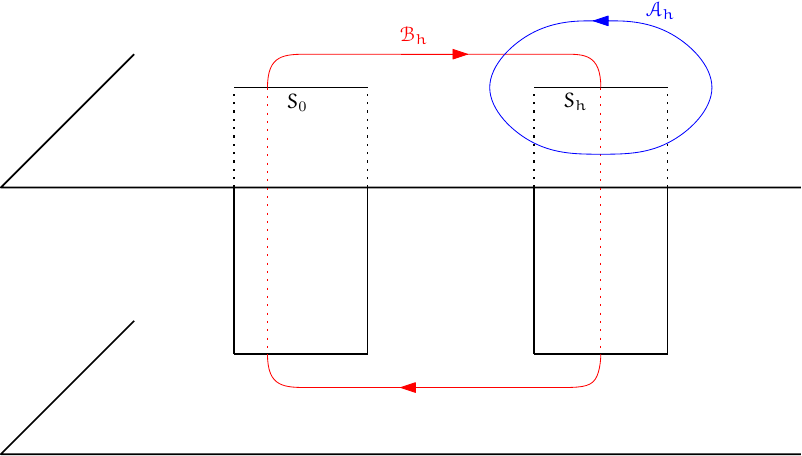}
 \caption{Two cycles $\A_{h}$ and $\B_{h}$.}\label{fig:cycles}
\end{figure}

\subsubsection[The 1-form phi and the Stieltjes transform]{The 1-form $\boldsymbol{\phi}$ and the Stieltjes transform}

The branched covering $X \colon \widehat{C} \rightarrow \widehat{\mathbb{C}}$ alone does not determine the equilibrium measure: we also need to specify the meromorphic function
 $Y\colon (x,s) \mapsto -\frac{1}{2}M(x)s$ on $\widehat{C}$. It is such that $Y(x) = \mp \frac{1}{2} M(x)s(x)$ for $x \in \widehat{C}_\pm$. The advantage to work with the Riemann surface $\widehat{C}$ is that the function~$W_1$, originally defined on $\CC \setminus S$ (Section~\ref{eqmessec}) can be analytically continued to a~meromorphic function on the whole $\widehat{C}$. Indeed, the meromorphic function ${\W_{1}(z) = \frac{V'(X(z))}{2} + Y(z)}$ coincides with $W_1(X(z))$ for $z \in \widehat{C}_+$. We can also see this by noticing that $F_{\pm}$ in \eqref{Fpm} are two solutions of the Riemann--Hilbert problem
\begin{equation*}
\forall x \in \mathring{S}\qquad F(x + {\rm i}0) + F(x - {\rm i}0) = V'(x).
\end{equation*}

\begin{Definition} We equip $\widehat{C}$ with the meromorphic 1-form $\phi(z) = \W_{1}(z)\dd X(z)$.
\end{Definition}
The previous discussion shows that $\phi$ has a simple pole at $\infty_+$ with residue $-1$, and a higher order pole at $\infty_-$ with
\begin{equation*}
 \phi(z)  \mathop{=}_{z \rightarrow \infty_-} \dd V(X(z)) - \frac{\dd X(z)}{X(z)} + \order{\frac{\dd X(z)}{X(z)^2}},
\end{equation*}
where we recall that $\zeta_{\infty_\pm} = 1/X$ is a local coordinate near $\infty_\pm$. As $X$ has two simple poles at~$\infty_\pm$, the form $\dd X$ thus
has double poles at these points with
\[
\dd X = -\zeta_{\infty_\pm}^{2} \dd \zeta_{\infty_\pm}.
\]

The $1$-form $\phi$ therefore decomposes as in Section~\ref{sec:decomp-merom-forms}:
\begin{equation}\label{eq:decomp-omega}
 \phi = \sum_{h=1}^{g} 2{\rm i}\pi\epsilon^{*}_{h}\dd u_{h} + \dd S_{\infty_{+}, \infty_{-}} - \sum_{k=1}^{d}\frac{t_{k}}{k} \dd B_{\infty_{-}, k},
\end{equation}
where the potential is $V(x) = \sum_{k = 1}^{d} \frac{t_kx^k}{k}$. The path from $\infty_+$ to $\infty_-$ used in the definition of the $1$-form $\dd S_{\infty_+,\infty_-}$ is chosen so that it does not intersect $(\mathcal{A}_h,\mathcal{B}_h)_{h = 1}^{g}$. For instance, one can take it to be $\iota_+\bigl(a_0 + {\rm i}\mathbb{R}_{\geq 0}\bigr) \cup \iota_-\bigl(a_0 - {\rm i}\mathbb{R}_{\leq 0}\bigr)$.

\begin{Remark}\label{remB0} For the equilibrium measure in the unconstrained case, the property noticed in Remark~\ref{Rem:vanishint} can be equivalently rewritten as
\[
\forall h \in [g] \qquad \oint_{\mathcal{B}_h} \phi = 0.
\]
As the filling fractions are real, this implies that for any $\gamma \in H_1\bigl(\widehat{C},\mathbb{Z}\bigr)$, we have $\operatorname{Re}\bigl(\oint_{\gamma} \phi\bigr) = 0$. Pairs $\bigl(\widehat{C},\phi\bigr)$ satisfying this property are called Boutroux curves, see \cite{bertola_boutroux_2007}.
\end{Remark}

\subsubsection{The fundamental bidifferential and the second correlator}
\label{sec:second-corr-fund}

Under the assumptions discussed in Section~\ref{sec:expans-part-funct}, \cite{BGmulti} shows that the $N \to \infty$ limit of the second correlator in the model with
fixed filling fractions $\bm{N}/N \rightarrow \bm{\epsilon}$ exists
\begin{equation*}
 W_{2}(x_1,x_2) = \lim_{N \to \infty} \frac{\beta}{2}\Biggl(\bigg\langle\Tr\biggl(\frac{1}{x_1 - \Lambda}\biggr)\Tr\biggl(\frac{1}{x_2 - \Lambda}\biggr)\bigg\rangle^{V}_{N,\bm{N}/N} - W_{1}(x_1)W_{1}(x_2)\Biggr).
\end{equation*}
It can be shown to satisfy the Riemann--Hilbert problem
\begin{equation}
\label{RH2}
 \forall (x_1,x_2) \in (\CC \setminus S) \times \mathring{S},\qquad W_{2}(x_1,x_2 + {\rm i}0) + W_{2}(x_1,x_2 - {\rm i}0) = -\frac{1}{(x_1 - x_2)^{2}}
\end{equation}
for $x \in C\setminus S$ and $y \in S$, see for instance \cite[Chapter
3]{CoursEynard} or \cite{BEO}. Besides, $W_2(x_1,x_2) = \order{1/x_i^2}$ as $x_i \rightarrow \infty$ since the total number of particles is deterministic, and
\begin{equation}
\label{Aw2}
\forall h \in [g]\qquad \oint_{\A_h} W_2(x_1,x_2)\,\dd x_1 = 0,
\end{equation}
since the filling fraction of the segment $A_h$ is fixed.

The Riemann--Hilbert problem \eqref{RH2} implies that we can define a meromorphic
function $\W_{2}(z_1,z_2)$ on $\widehat{C} \times \widehat{C}$ such that
$\W_{2}(z_1,z_2) = W_{2}(X(z_1),X(z_2))$ when
$z_1,z_2 \in \widehat{C}_{+}$. By examining the behavior of $\W_2$ at the poles and considering the $\A$-period conditions \eqref{Aw2}, one can identify it in terms of the fundamental bidifferential:
\begin{equation}
 \label{BW2}
 B(z_1,z_2) = \W_{2}(z_1,z_2)\dd X(z_1)\dd X(z_2) + \frac{\dd X(z_1)\dd X(z_2)}{(X(z_1) - X(z_2))^{2}}.
\end{equation}

\begin{Definition}
The spectral curve of a $\beta$-ensemble is the marked compact Riemann surface $\big(\widehat{C},\mathcal{A},\mathcal{B}\big)$ equipped with the meromorphic functions $X$, $Y$ and the bidifferential $B$.
\end{Definition}

\subsection[Characterisation of spectral curves of beta-ensembles]{Characterisation of spectral curves of $\boldsymbol{\beta}$-ensembles}

In Section \ref{sec:constr-spectr-curve}, we explained that the Riemann surface $\widehat{C}$ underlying the spectral curve of \mbox{a~$\beta$-ensemble} is hyperelliptic with real Weierstra\ss{} points. We now prove the converse, namely that all such Riemann surfaces can be realised (non uniquely) as the underlying Riemann surface of the spectral curve of an unconstrained $\beta$-ensemble.

 \begin{Definition}
 If $G$ is a meromorphic function in a neighborhood of $\infty$ in $\widehat{\mathbb{C}}$, we define its polynomial part $\mathcal{V}[G](x)$, which is the unique polynomial such that $G(x) = \mathcal{V}[G](x) + \order{\frac{1}{x}}$ as $x \rightarrow \infty$.
 \end{Definition}

\begin{prop}
\label{propallhyp} For any $a_0 < b_0 < \cdots < a_g < b_g$, there exists a polynomial $V$ of degree $(2g + 2)$ with top coefficient $\frac{t_{2g + 2}}{2g + 2} > 0$ and there exist for each $h \in [0,g]$ a segment $A_h$ which is a neighborhood of $[a_h,b_h]$ in $\mathbb{R}$, such that the unconstrained $\beta$-ensemble with potential $V$ on $A = \bigsqcup_{h = 0}^{g} A_h$ admits an equilibrium measure with support $S = \bigsqcup_{h = 0}^{g} [a_h,b_h]$ and in \eqref{W1SM} we have $M(x) = t_{2g + 2}\prod_{h = 1}^g (x - z_h)$ having roots outside $A$ and such that $b_{h - 1} < z_h < a_h$ for any $h \in [g]$.
\end{prop}
\begin{proof} Take $2g + 2$ real points $a_0 < b_0 < \cdots < a_g < b_g$ and introduce polynomials $\sigma(x) = \prod_{h = 0}^{g}(x - a_h)(x - b_h)$. We have seen in Section~\ref{eqmessec} that there exists a unique holomorphic function $s(x)$ on $\mathbb{C} \setminus \bigsqcup_{h = 0}^{g} [a_h,b_h]$ such that $s(x)^2 = \sigma(x)$ and $s(x) \sim x^{g + 1}$
as $x \rightarrow \infty$ in the complex plane. Take $h \in [g]$ and introduce the continuous function
\[
\forall \boldsymbol{\lambda}\in [0,1]^{g}\qquad J_h(\boldsymbol{\lambda}) = \int_{b_{h - 1}}^{a_h} s(x) \prod_{k = 1}^{g}(x - (\lambda_k b_{k - 1} + (1 - \lambda_k) a_k))\,\dd x.
\]
By continuity, $s(x)$ has constant sign for $x \in (b_{h - 1},a_h)$. Besides, for $\lambda_h \in \{0,1\}$ we have $\operatorname{sgn}(J_{h}(\boldsymbol{\lambda})) = (-1)^{g - h + 1 - \lambda_h}$. The Poincar\'e--Miranda theorem \cite{Miranda} then implies the existence of $\boldsymbol{\lambda}^* \in [0,1]^{g}$ such that $J_h(\boldsymbol{\lambda}^*)$ vanishes for any $h \in [g]$. Since $J_h(\boldsymbol{\lambda})$ does not vanish when $\lambda_h \in \{0,1\}$, this $\boldsymbol{\lambda}^*$ must be in $(0,1)^{g}$. We let $z_h = \lambda_h^*b_{h - 1} + \big(1 - \lambda_h^*\big)a_h$ for $h \in [g]$ and introduce the polynomial $M(x) = t_{2g + 2}\prod_{h = 1}^{g} (x - z_h)$, where the constant $t_{2g + 2}$ is chosen such that
\begin{equation}
\label{sumhb}
\sum_{h = 0}^{g} \int_{a_h}^{b_h} \frac{M(x) \operatorname{Im}(s(x + {\rm i}0))}{2\pi} = 1.
\end{equation}
The sign discussion for $s$ in Section~\ref{eqmessec} reveals that all terms in \eqref{sumhb} are positive, thus $t_{2g + 2} > 0$. Then, $V(x) = \int_{0}^{x} \mathcal{V}[M \cdot s](\xi) \dd \xi$ is a polynomial of degree $2g + 2$ with top coefficient $\frac{t_{2g + 2}}{2g + 2}$, and
\[
\mathop{\text{Res}}_{x = \infty} \frac{s(x)M(x)}{2}]\, \dd x = -\frac{1}{2{\rm i}\pi} \sum_{h = 0}^{g} \oint_{\mathcal{A}_h} \frac{s(x)M(x)}{2}\, \dd x = \sum_{h = 0}^{g} \int_{a_h}^{b_h} \frac{M(x) \operatorname{Im}(s(x + {\rm i}0))}{2\pi} = 1,
\]
where $\mathcal{A}_h$ is a counterclockwise loop around $[a_h,b_h]$. Therefore,
\[
\frac{V'(x)}{2} \mathop{\mathop{=}}_{x \rightarrow \infty} \frac{M(x)s(x)}{2} + \frac{1}{x} + \mathcal{O}\biggl(\frac{1}{x^2}\biggr).
\]
This $V$ defines the potential in a $\beta$-ensemble which we consider over the domain $A = \bigsqcup_{h = 0}^{g} A_h$, where $A_h = \bigl[a_h',b_h'\bigr]$ and $z_{h} < a_h' < a_h$ and $b_h < b_h' < z_{h + 1}$ for any $h \in [0,g]$, with the conventions $z_0 = -\infty$ and $z_{g + 1} = +\infty$. It remains to check that $W(x) := \frac{1}{2}(V'(x) - M(x)s(x))$ is the Stieltjes transform of the equilibrium measure of this (unconstrained) $\beta$-ensemble.

We define the measure $\mu$ with support $S = \bigsqcup_{h = 0}^{g} [a_h,b_h]$ and density
\[
\frac{\dd\mu}{\dd x} = \frac{W(x - {\rm i}0) - W(x + {\rm i}0)}{2{\rm i}\pi} = \frac{M(x)s(x)}{2\pi} \ind_S(x).
\]
By construction, $W$ is the Stieltjes transform of $\mu$. Since $M$ has a single zero between each components of the support, $\mu$ is a positive measure, see the sign discussion in Section~\ref{eqmessec}. Define $U(x) = V(x) - 2\int_{S} \ln|x - \xi|\dd \mu(\xi)$. We clearly have
\[
\forall x \in S \qquad U'(x) = V'(x) - W_1(x + {\rm i}0) - W_1(x - {\rm i}0) = 0.
\]
Integrating this from $a_{h}$ to $x \in (a_h,b_h)$, we find a constant $c_h$ such that $U(x) = c_h$ for any $x \in (a_h,b_h)$. Besides, for any $h \in [g]$ we compute
\[
c_{h} - c_{h - 1} = \int_{b_{h - 1}}^{a_h} U'(x)\, \dd x = \int_{b_{h - 1}}^{a_h}\bigl(V'(x) - 2W_1(x)\bigr)\dd x = \int_{b_{h - 1}}^{a_{h}} M(x)s(x) \,\dd x.
\]
Since we have chosen $(z_1,\ldots,z_g)$ so that this integral vanishes, $c_h$ is independent of $h$. As a result, $\mu$ satisfies the characterisation of the equilibrium measure from Theorem~\ref{thm:equilibrium} (unconstrained case). By uniqueness, this must be the equilibrium measure: $\mu = \mu_{{\rm eq}}$.
\end{proof}

\subsection{Deformations of the curve}

We consider real and complex deformations of the complex curves, that will be used in Section~\ref{sec:expans-part} to extend the validity of our formulae beyond their realisation for spectral curves of $\beta$-ensembles. We first show that within the class of spectral curves of $\beta$-ensembles, we can always realise any vector of filling fractions in a small neighborhood of a given one by perturbation of the support.

\begin{Lemma}
 \label{periodma}Let $a_0 < b_0 < \cdots < a_g < b_g$ and take a
 corresponding $M(x) = t_{2g + 2} \prod_{h=1}^{g}(x - z_{h})$ with $z_h \in (b_{h - 1},a_h)$ as in
 Proposition~$\ref{propallhyp}$. There exists a small neighborhood
 $\Omega \subset \mathbb{R}^{2g}$ of $(a_h,z_h)_{h = 1}^{g}$ such
 that the map $\Pi \colon \Omega \rightarrow \mathbb{R}^{2g}$ given by
\[
\Pi(\tilde{a}_1,\tilde{z}_1,\ldots,\tilde{a}_g,\tilde{z}_g) = \biggl(\int_{\tilde{a}_h}^{b_h} \tilde{M}(x)\sqrt{-\tilde{\sigma}(x)}\,\dd x,\,\int_{b_{h - 1}}^{\tilde{a}_h} \tilde{M}(x) \sqrt{\tilde{\sigma}(x)}\, \dd x\biggr)_{h = 1}^{g}
\]
is a diffeomorphism onto its image, where we have set
\[
\tilde{M}(x) = t_{2g + 2} \prod_{k=1}^{g}(x - \tilde{z}_k) \qquad \text{and} \qquad \tilde{\sigma}(x) = (x - a_0)(x - b_0)\prod_{k = 1}^{g} (x - \tilde{a}_k)(x -b_k).
\]
\end{Lemma}
This will be used in the following form.
\begin{Corollary}
\label{codense}There is a dense set of $a_0 < b_0 < \cdots < a_g < b_g$ for which there exists \mbox{a~$\beta$-ensemble} whose associated equilibrium measure of Theorem~$\ref{thm:equilibrium}$ has filling fractions $\boldsymbol{\epsilon}^*$ whose components $\epsilon_{1}^*,\ldots,\epsilon_g^*$ are $\mathbb{Q}$-linearly independent.
\end{Corollary}

\begin{proof}[Proof of Lemma~\ref{periodma}] $\Pi$ is a smooth
 function of $(\tilde{a}_h,\tilde{z}_h)_{h = 1}^{g}$ in the range
 $b_0 < \tilde{z}_{1} < \tilde{a}_1 < b_1 < \tilde{z}_2 < \tilde{a}_2 < \cdots < \tilde{z}_g < \tilde{a}_g < b_g$.
 We compute its Jacobian
\begin{gather}
 \det \begin{pmatrix} \displaystyle\int_{\tilde{a}_h}^{b_h} \frac{-\tilde{M}(x)\sqrt{-\tilde{\sigma}(x)}}{2(x - \tilde{a}_k)}\dd x & \displaystyle\int_{\tilde{a}_{h}}^{b_h} \frac{-\tilde{M}(x)\sqrt{-\tilde{\sigma}(x)}}{(x - \tilde{z}_k)} \dd x \vspace{2mm}\\ \displaystyle\int_{b_{h - 1}}^{\tilde{a}_h} \frac{-\tilde{M}(x)\sqrt{\tilde{\sigma}(x)}}{2(x - \tilde{a}_k)}\dd x & \displaystyle\int_{b_{h - 1}}^{\tilde{a}_h} \frac{-\tilde{M}(x)\sqrt{\tilde{\sigma}(x)}}{(x - \tilde{z}_k)} \dd x \end{pmatrix}_{1 \leq h,k \leq g} \nonumber\\
\qquad{} = \frac{1}{2^g} \int_{\tilde{a}_1}^{b_1}    \cdots \int_{\tilde{a}_{g}}^{b_g} \prod_{h = 1}^{g} \dd x_h \tilde{M}(x_h)\sqrt{-\tilde{\sigma}(x_h)}\nonumber \\
\qquad\quad{}\times
\int_{b_0}^{\tilde{a}_1}    \cdots \int_{b_{g - 1}}^{\tilde{a}_g} \dd \xi_h \tilde{M}(\xi_h) \sqrt{\tilde{\sigma}(\xi_h)} \cdot \det\begin{pmatrix} \dfrac{1}{x_h - \tilde{a}_k} & \dfrac{1}{x_h - \tilde{z}_k} \vspace{2mm}\\ \dfrac{1}{\xi_h - \tilde{a}_k} & \dfrac{1}{\xi_h - \tilde{z}_k}\end{pmatrix}_{1 \leq h,k \leq g}, \label{injt}
\end{gather}
where we used the fact that $\sqrt{\pm \tilde{\sigma}(x)}$ vanishes at the endpoints of the integration intervals. The determinant in the integrand is a Cauchy determinant and can be readily evaluated
\begin{gather*}
	\det
	\begin{pmatrix}
	 \dfrac{1}{x_h - \tilde{a}_k} & \dfrac{1}{x_h - \tilde{z}_k} \vspace{2mm}\\
	 \dfrac{1}{\xi_h - \tilde{a}_k} & \dfrac{1}{\xi_h - \tilde{z}_k}
	\end{pmatrix}_{1 \leq h,k \leq g}
\\
\qquad{}	= \Delta(\boldsymbol{\tilde{a}})\Delta(\boldsymbol{\tilde{z}}) \Delta(\boldsymbol{x}) \Delta(\boldsymbol{\xi}) \prod_{h,k = 1}^{g} \frac{ (\tilde{z}_h - \tilde{a}_k)(\xi_h - x_k) }{ (x_h - \tilde{a}_k)(\xi_h - \tilde{a}_k)(x_h - \tilde{z}_k)(\xi_h - \tilde{z}_k)}\\
\qquad{} = \frac{t_{2g + 2}^{2g} \Delta(\boldsymbol{\tilde{a}})\Delta(\boldsymbol{\tilde{z}}) \Delta(\boldsymbol{x}) \Delta(\boldsymbol{\xi})}{\prod_{h=1}^{g} \tilde{M}(x_{h}) \tilde{M}(\xi_{h})} \prod_{h,k = 1}^{g} \frac{(\tilde{z}_h - \tilde{a}_k)(\xi_h - x_k)}{(x_h - \tilde{a}_k)(\xi_h - \tilde{a}_k)}.
\end{gather*}
For $(\tilde{a}_h,\tilde{z}_h)_{h = 1}^{g}$ close enough to
$(a_h,z_h)_{h = 1}^{g}$, the zeros of $\tilde{M}$ are outside
$\bigsqcup_{h = 1}^{g} \bigl[\tilde{a}_{h},\tilde{b}_h\bigr]$, so that the
sign of the integrand in \eqref{injt} remains constant in the whole integration range. The determinant of the
Jacobian of $\Pi$ is thus nonzero, and $\Pi$ is a local
diffeomorphism.
\end{proof}

\begin{proof}[Proof of Corollary~\ref{codense}]
If $(\tilde{a}_h,\tilde{z}_h)_{h = 1}^{g} \in \Omega$, call $\tilde{\mu}$ the measure supported on $\tilde{S} = [a_0,b_0] \cup \bigsqcup_{h = 1}^{g} [\tilde{a}_h,b_h]$ with density $\frac{1}{2\pi} \tilde{M}(x)\sqrt{-\tilde{\sigma}(x)}$. At $(a_h,z_h)_{h = 1}^{g}$ this $\tilde{\mu}$ is by construction the equilibrium measure of a $\beta$-ensemble, which we simply denote $\mu$: it is in particular a probability measure with vector of filling fractions \smash{$\big(\epsilon_h^*\big)_{h = 1}^{g}$} and the $h$-th second component of $\Pi(a_1,z_1,\ldots,a_g,z_g)$ is equal to $(-1)^{g - h}(U(a_h) - U(b_{h - 1})) = 0$ for $h \in [g]$. So, $\Pi$ induces a homeomorphism from~$\Omega$ to a neighborhood $\Omega' \subset \mathbb{R}^{2g}$ of $\bigl((-1)^{g-h}2\pi \epsilon_h^* ,0\bigr)_{h = 1}^{g}$. By continuity with respect to the parameters, $\tilde{\mu}$ remains a positive measure on each component of \smash{$\tilde{S}$} for all parameters in a (possibly smaller)~$\Omega$, and that the total mass of $\tilde{\mu}$ defines a positive continuous function on $\Omega$. In particular, $\tilde{\mu}' = \tilde{\mu}/\tilde{\mu}\big(\tilde{S}\big)$ is a probability measure on $\tilde{S}$.

If $\epsilon_1^*,\ldots,\epsilon_g^*$ are $\mathbb{Q}$-linearly dependent, we can approximate $\bigl((-1)^{g-h}2\pi \epsilon_h^* ,\,0\bigr)_{h = 1}^{g}$ to arbitrary precision by $2g$-tuples $\bigl((-1)^{g-h}2\pi \tilde{\epsilon}_h ,\,0\bigr)_{h = 1}^{g} \in \Omega'$ such that $\tilde{\epsilon}_1,\ldots,\tilde{\epsilon}_g$ are $\mathbb{Q}$-linearly independent. Applying $\Pi^{-1}$, we get an approximation $(\tilde{a}_h,\tilde{z}_h)_{h = 1}^{g}$ of $(a_h,z_h)_{h = 1}^{g}$ at arbitrary precision whose associated probability measure $\tilde{\mu}'$ is by construction (follow the proof of Proposition~\ref{propallhyp}) the equilibrium measure of the $\beta$-ensemble with potential
\[
\tilde{V}(x) = \frac{1}{\tilde{\mu}\big(\tilde{S}\big)} \int_0^{x} \mathcal{V}\bigl[\tilde{M}\cdot \tilde{s}\bigr](x),
\]
with $\tilde{s}$ like $s$ of Section~\ref{eqmessec} but with
$\tilde{a}$s instead of $a$s, i.e., a choice of square root of
$\prod_{h = 0}^{g}(x - \tilde{a}_{h})(x - b_{h})$. Let us detail this
claim. We introduce the effective energy
$\tilde{U}(x) = \tilde{V}(x) - 2\int_{\tilde{S}}\ln|x - \xi| \dd \tilde{\mu}(\xi)$
associated to $\tilde{\mu}$. It satisfies both
$\tilde{U}(\tilde{a}_{h}) - \tilde{U}(b_{h}) = 0$
(because the second component of the image of
$(\tilde{a}_{h}, \tilde{z}_{h})_{h=1}^{g}$ by $\Pi$ is zero) and
$\tilde{U}'(x) = 0$ for $x$ in the support of $\tilde{\mu}$ (because the Stieltjes
transform of $\tilde{V}'$ is the polynomial part of the density of
$\tilde{\mu}$, up to a factor $2\pi$). It satisfies \eqref{eq:charc-eq},
and is the equilibrium measure of the $\beta$-ensemble with potential
$\tilde{V}$. This equilibrium measure has vector of filling fractions
$\bigl(\tilde{\epsilon}_h/\tilde{\mu}\big(\tilde{S}\big)\bigr)_{h = 1}^{g}$, whose components
remain $\mathbb{Q}$-linearly independent.
\end{proof}

In a second step, we will leave the realm of spectral curves of
$\beta$-ensembles and rather consider their complex deformations. Here it becomes important to keep track of the marking.
The equation of a hyperelliptic curve $s^2 = \prod_{h = 0}^{g} (x - a_h)(x - b_h)$ is parameterised by the set~$\bm{\Delta}_{2g + 2}$ of $(2g + 2)$-tuple $(a_h,b_h)_{h = 0}^{g}$ of pairwise distinct complex numbers. Its universal cover~$\widetilde{\bm{\Delta}}_{2g + 2}$ based at a tuple of strictly increasing real numbers parametrises the equation of the hyperelliptic curve together with a choice of marking: at the base point it is the one described in Section~\ref{sec:constr-spectr-curve}, and there is a unique way to get from there a marking for any other point in~$\widetilde{\bm{\Delta}}_{2g + 2}$ by performing continuous deformations of the representatives of the homology cycles. The outcome is an analytic family $\hat{\bm{C}} \rightarrow \widetilde{\bm{\Delta}}_{2g + 2}$ of marked hyperelliptic curves, which coincide with the one described in Section~\ref{sec:constr-spectr-curve} above the connected component of the base point in the real locus of~$\widetilde{\bm{\Delta}}_{2g + 2}$. Concretely, in other real connected components, the symplectic basis of homology has changed by an $\text{Sp}_{2g}(\mathbb{Z})$-transformation compared to Section~\ref{sec:constr-spectr-curve}, and so must have the matrix of periods. Let us denote likewise $\widetilde{\bm{C}}$ the family of universal covers over $\widetilde{\bm{\Delta}}_{2g + 2}$. We will rely on the following basic fact in complex geometry, see, e.g., \cite[Chapter
1]{period-book}.
\begin{Lemma}\label{lem:holo-periods}
 The period matrix \eqref{tauhk} is a holomorphic function on $\widetilde{\bm{\Delta}}_{2g + 2}$. The Abel map based at $\infty_+$ is a holomorphic function $\widetilde{\bm{C}} \rightarrow \mathbb{C}^{g}$.
\end{Lemma}

\section{Asymptotics of the partition function and the kernels}

\label{sec:expans-part}
\subsection{Expansion of the partition function and generalised central limit theorem}
\label{sec:expans-part-funct}

The large $N$ asymptotic expansion of the partition function of the
$\beta$-ensembles in the multi-cut regime was established in \cite{BGmulti}, under
assumptions which are satisfied for the potentials that we consider in
Theorem~\ref{thm:equilibrium}. In particular, the off-criticality assumption on $A$ corresponds to $M$ of Lemma~\ref{lemMmm} having no zeros on $A$. We reproduce here the formulae for these asymptotics, which will be our starting point.

\begin{Theorem}[{\cite[Theorem 1.5]{BGmulti}}]
\label{Asymppart}Let $g \geq 1$, let $V$ as in Theorem~$\ref{thm:equilibrium}$ and assume $M$ from Lemma~$\ref{lemMmm}$ has no zeros on $A$. The partition function has the following expansion
as $N \to \infty$
\[
 Z^{V}_{N} \sim N^{\frac{\beta}{2}N + \varkappa}{\rm e}^{N^{2}\mathcal{E}[\mu_{{\rm eq}}] + N\mathcal{S}[\mu_{{\rm eq}}] + \mathcal{G}[\mu_{{\rm eq}}]}\,\vartheta_{-N\bm{\epsilon}^{*},\bm{0}}\bigl(\bm{v}_{{\rm eq}}\big| \tfrac{\beta}{2}\bm{\tau}\bigr).
\]
Here
\[
\mathcal{S}[\mu] = \biggl(1 - \frac{\beta}{2}\biggr)({\rm Ent}[\mu] - \ln(\beta/2)) + (\beta/2)\ln(2\pi/e) - \ln \Gamma(\beta/2),
\]
where ${\rm Ent}[\mu]$ is the von Neumann entropy of the probability measure $\mu$, $\varkappa$ is a known universal constant depending only on $g$ and $\beta$, $\mathcal{G}[\mu]$ is a continuous functional whose expression is irrelevant for our purposes, and
\begin{equation}
\label{vnabl}\bm{v}_{{\rm eq}} = \frac{\nabla_{\epsilon} \mathcal{S}[\mu_{{\rm eq},\bm{\epsilon}}]}{2{\rm i}\pi}\bigg|_{\bm{\epsilon} = \bm{\epsilon}^*} .
\end{equation}
\end{Theorem}
In \eqref{vnabl} one differentiates with respect to $\bm{\epsilon} = (\epsilon_1,\ldots,\epsilon_g)$,
keeping in mind that the filling fraction $\epsilon_{0}$, associated to the
component of the support $[a_{0}, b_{0}]$, satisfies
$\epsilon_0 = 1 - (\epsilon_1 + \cdots + \epsilon_g)$ and thus depends on
$\epsilon_{1}, \ldots, \epsilon_{g}$. The kernels appearing in the determinantal and
Pfaffian formulae of Borodin and Strahov (Theorems~\ref{beta2form}--\ref{th:idbeta4}) can be
estimated using the following generalised central limit theorem. We
use the notation $\oint_{S}$ for a sum of contour integrals in the
positive direction around the connected components of the support $S$
of the equilibrium measure.

\begin{Theorem}[{\cite[Theorem 1.6]{BGmulti}}]\label{thm:clt}
 Let $f$ be a holomorphic function in a complex neighborhood of $A$. Under the same assumptions as Theorem~$\ref{Asymppart}$, we have as $N \to \infty$
\[
\bigl\langle\ee^{\sum_{i=1}^{N}f(\lambda_i)}\bigr\rangle^{V}_{N} \sim{\rm e}^{N\mathcal{L}[f] + \mathcal{H}[f] + \frac{1}{2}\mathcal{Q}[f,f]} \frac{\vartheta_{-N\bm{\epsilon}^{*},\bm{0}}\bigl(\bm{v}_{{\rm eq}} + \bm{\mathcal{U}}[f] \big| \frac{\beta}{2}\bm{\tau}\bigr)}{\vartheta_{-N\bm{\epsilon}^{*},\bm{0}}\bigl(\bm{v}_{{\rm eq}} \big| \frac{\beta}{2}\bm{\tau}\bigr)} .
\]
Here
\[
\mathcal{L}[f] = \oint_{S} W_1(\xi)\frac{f(\xi)\dd\xi}{2{\rm i}\pi}, \qquad
\mathcal{Q}[f,f] = \frac{2}{\beta} \oint_{S^2} W_{2}(\xi_1,\xi_2) \frac{f(\xi_1)\dd \xi_1}{2{\rm i}\pi} \frac{f(\xi_2)\dd \xi_2}{2{\rm i}\pi},
\]
where $W_1$ and $W_2$ are calculated in a model with fixed filling fraction tending to $\bm{\epsilon}^*$,
$\mathcal{H}[f]$ is a~linear form whose expression is irrelevant, and
\[
\bm{\mathcal{U}}[f] = \oint_{S} f(X(z)) \frac{ \dd \bm{u}(z)}{2{\rm i}\pi}.
\]
\end{Theorem}

\subsection{Three explicit formulae}

We will establish in Section~\ref{sec:can} the following expression for the equilibrium energy in terms of the geometry of the spectral curve.
\begin{prop} \label{lem:Ef0} The equilibrium energy is
\[
-\mathcal{E}[\mu_{{\rm eq}}] = \frac{\beta}{2}\mathcal{L}[V] + \frac{\beta^2}{8}\mathcal{Q}[V,V] + {\rm i}\pi \beta \bm{\epsilon}^* \cdot  (\bm{\tau}(\bm{\epsilon}^*) + \bm{u}(\infty_-) ).
\]
\end{prop}
The right-hand side is proportional to $\frac{\beta}{2}$, since $W_2$ contains a factor $\frac{2}{\beta}$, and so is the left-hand side. Accordingly this is an identity between $\beta$-independent quantities. There is a classical link between random matrix theory and the theory
of Frobenius manifolds: the free energy at leading order,
namely $\mathcal{E}[\mu_{{\rm eq}}]$ coincides with the
prepotential of the Hurwitz--Frobenius manifold associated to the
spectral curve of the random matrix ensemble. A formula for this free energy was established in the more general context of the 2-matrix model in \cite{Bertola}, involving only the geometry of the spectral curve. It involves the same ingredient but does not have exactly the same form as Proposition~\ref{lem:Ef0}, which is the formula we need.

 Before going further, we give two extra formulae. The first one evaluates the argument $\bm{v}_{{\rm eq}}$ of the theta functions in Theorems~\ref{Asymppart} and \ref{thm:clt}; it is not necessary for Section~\ref{sec:derivation-formulae}, but we include it for completeness. The second one will be used in the proof of Lemma~\ref{Kasymint}.

 \begin{prop}\label{lem:compute-nu}
 Assume $M$ in \eqref{W1SM} is of the form
 $M(x) = t_{2g + 2}\prod_{h = 1}^g (x - z_h)$ with roots
 $z_{h}\notin A$ satisfying $b_{h - 1} < z_h < a_h$ for any
 $h \in [g]$, and denote by $(\bm{e}_1,\ldots,\bm{e}_g)$ the
 canonical basis of $\mathbb{C}^g$.

 The function $\operatorname{Im}\bm{u}(z) = \int_{\infty_+}^{z} \dd \bm{u}$ is a
 single-valued function of $z \in \widehat{C}_+$, and we have
\begin{equation*}
\bm{v}_{{\rm eq}} = 2\pi\biggl(1 - \frac{\beta}{2}\biggr)\Biggl[ \frac{g + 1}{2{\rm i}}\bm{u}(\infty_-) + \sum_{k = 1}^{g} \biggl(\operatorname{Im}\bm{u}(z_k) + \frac{g + 1 - k}{2{\rm i}} \bm{\tau}(e_k)\biggr)\Biggr].
\end{equation*}
\end{prop}
\begin{proof} See Appendix~\ref{AppA}. As we explain in Section~\ref{sec:can}, $\bm{u}(\infty_-)$ and $\bm{\tau}$ are purely imaginary, so all terms in the right-hand side are real, as it should be. The domain $\widehat{C}_+$ is homeomorphic to the non simply-connected domain $\widehat{\mathbb{C}} \setminus S$, so $\int_{\infty_+}^{z} \dd \bm{u}$ is multi-valued. However, the ambiguities are $\mathcal{A}$-periods of $\dd\bm{u}$ which are real, so $\operatorname{Im}\bm{u}(z)$ is single-valued.
\end{proof}
\begin{Remark}
 If we remove the assumption on the roots of $M$ in Lemma
 \ref{lem:compute-nu}, we can still compute $ \bm{v}_{{\rm eq}}$ and
 obtain a formula of a similar form.
\end{Remark}

 \begin{Lemma}\label{lem:compute-nu-infty}
\[
\bm{\mathcal{U}}[V] = - \Res_{\infty_+} V\dd\bm{u} = \bm{\tau}(\bm{\epsilon}^*) + \bm{u}(\infty_{-}).
\]
\end{Lemma}

\begin{proof}
 The first equality is obtained by moving the contour around $S$ to $\infty_+$. We focus on the second equality. Since $\dd V = 2(\phi - Y\dd X)$, we have
 \begin{equation*}
	 \Res_{\infty_{+}}V\dd \bm{u} = -\Res_{\infty_{+}}\bm{u}\dd V = -2\Res_{\infty_{+}}\bm{u}(\phi - Y\dd X).
 \end{equation*}
 Since $\infty_+$ is the base point for the Abel map, we have
 $\bm{u}(\infty_{+}) = 0$ and since $\phi$ has only a simple pole at $\infty_+$,
 the first term gives a vanishing residue. The hyperelliptic involution preserves $X$, sends $Y$ and $\dd\bm{u}$ to their opposite, and $\infty_+$ to $\infty_-$. Hence, it sends $\bm{u}$ to $\bm{u}(\infty_-) - \bm{u}$, where $\infty_- \in \widehat{C}^0$. Using the involution as a change of variables, we get
 \begin{equation*}
\Res_{\infty_{+}}\bm{u}Y\dd X = \Res_{\infty_{-}}(\bm{u} - \bm{u}(\infty_{-}))Y\dd X = \Res_{\infty_{-}}\bm{u}Y\dd X - \bm{u}(\infty_{-}).
 \end{equation*}
For the last equality, since the only poles of $Y \dd X$ are $\infty_{\pm}$ we could evaluate
\[
 \Res_{\infty_-} Y \dd X = -\Res_{\infty_+} Y \dd X = -\Res_{\infty_+} \phi = 1.
\]
We then write
\[
\Res_{\infty_+} V \dd\bm{u} =  (\Res_{\infty_+} + \Res_{\infty_-} )Y \dd X - \bm{u}(\infty_-).
\]
The first term can be computed with the Riemann bilinear identity \cite[equation~(3.0.2)]{Farkas-Kra}
\begin{equation}
\label{uYdX}
	( \Res_{\infty_{+}} + \Res_{\infty_-}) \bm{u}Y\dd X = \frac{1}{2{\rm i}\pi} \sum_{h = 1}^{g} \biggl(\oint_{\A_{h}}\dd\bm{u} \cdot \oint_{\B_{h}}Y\dd X - \oint_{\B_{h}}\dd\bm{u}\cdot \oint_{\A_{h}}Y\,\dd X \biggr).
\end{equation}
Taking into account that for any $h,k \in [g]$, we have $\oint_{\B_{h}}Y\dd X = 0$ (Remark~\ref{remB0}), and $\oint_{\B_h}\dd u_{k} = \tau_{h,k}$ and $\oint_{\A_h} Y \dd X = 2{\rm i}\pi \epsilon^*_h$, we find that \eqref{uYdX} is equal to $-\bm{\tau}(\bm{\epsilon}^*)$.
\end{proof}

\subsection{Kernel asymptotics: intermediate computations}

We need to compute an asymptotic equivalent as $K \rightarrow \infty$ of kernels of the form
\begin{equation*}
 \frac{Z_{M}^{\frac{K}{M}V}}{Z_K^V} \Biggl\langle\prod_{j=1}^{m}\det(x_{j} - \Lambda)^{c_{j}}\Biggr\rangle^{\frac{K}{M}V}_{M},
\end{equation*}
where $x_{1}, \ldots, x_{m} \notin A$, $c_{1}, \ldots, c_{m} \in \mathbb{Z}$, and $(M - K) = p$ is a fixed integer. Notice that
\begin{equation}
 \label{Kasmim}
 \frac{Z_{M}^{\frac{K}{M}V}}{Z_K^V}\Biggl\langle\prod_{j=1}^{m}\det(x_{j} - \Lambda)^{c_{j}}\Biggr\rangle^{\frac{K}{M}V}_{M} = \frac{Z_M^{V}}{Z_K^{V}} \bigl\langle{\rm e}^{\sum_{i=1}^{M} f_{\bm{c}}(\lambda_{i}) + \frac{\beta}{2}pV(\lambda_i)}\bigr\rangle^{V}_{M},
\end{equation}
where we used the holomorphic function on a neighborhood of $A$
\begin{equation*}
 f_{\bm{c}}(\lambda) = \sum_{j=1}^{m}c_{j}\ln(x_{j} - \lambda)
\end{equation*}
and for $x_j \notin A$ we choose the cut of the logarithm away from $A$. In order to access the asymptotics of \eqref{Kasmim} via Theorem~\ref{thm:clt}, we first have to evaluate the following quantities.

\begin{Lemma}
 \label{L44}
 Let $z,z_1,z_2 \in C_+$. We have
 \begin{gather*}
	 \mathcal{L}\bigl[\ln(X(z) - \bullet)\bigr]   = 2{\rm i}\pi\bm{\epsilon}^*\cdot\bm{u}(z) - \ln\bigl(\eta(z)E(z, \infty_{+})^2\dd\zeta_{\infty_{+}}(\infty_{+})\bigr) - \sum_{k = 1}^{d} \frac{t_k}{k} \int_{\infty_+}^{z} \dd B_{\infty_-,k}, \\
	 \mathcal{Q}\bigl[\ln(X(z) - \bullet),V\bigr]   = \frac{2}{\beta} \sum_{k = 1}^{d} \frac{t_k}{k} \int_{\infty_+}^{z} \dd B_{\infty_+,k}, \\
	 \mathcal{Q}\bigl[\ln(X(z_1) - \bullet), \ln(X(z_2) - \bullet)\bigr]   = \frac{2}{\beta}\ln \biggl(\!\frac{E(z_1, z_2)}{E(z_1,\!\infty_{+})E(z_2,\!\infty_+)(X(z_2) \!-\! X(z_1))\dd \zeta_{\infty_+}\!(\infty_+)}\!\biggr), \\
	 \mathcal{Q}\bigl[\ln(X(z) - \bullet),\ln(X(z) - \bullet)\bigr]   = \frac{2}{\beta}\ln\biggl(\frac{1}{E(z,\infty_+)^2\dd \zeta_{\infty_+}(\infty_+)}\biggr),
\end{gather*}
where $\eta(z)$ is defined in \eqref{etadef}. In particular, we
observe the simplification
\begin{gather} \nonumber
 \mathcal{L}\bigl[\ln(X(z) - \bullet)\bigr] + \frac{\beta}{2} \mathcal{Q}\ln\bigl[\ln(X(z) - \bullet),V\bigr] \\
 \qquad{}= 2{\rm i}\pi \bm{\epsilon}^{*} \cdot \bm{u}(z) - \ln\bigl(\eta(z)E(z, \infty_{+})^{2}\dd\zeta_{\infty_{+}}(\infty_{+})\bigr).\label{simplifLQ}
\end{gather}
\end{Lemma}
\begin{proof}
 Let $x = X(z)$. The first two formulas are a consequence of the decomposition (\ref{eq:decomp-omega})
 of~$\phi$. We have
 \begin{align*}
 \mathcal{L}\bigl[\ln(x - \bullet)\bigr]
 &= \oint_{S} \frac{\dd\xi}{2{\rm i}\pi}\,\ln(x - \xi)W_{1}(\xi)
  = \oint_{S} \frac{\dd \xi}{2{\rm i}\pi} \biggl[\int_{\infty}^{x}\biggl(\frac{1}{\xi' - \xi} - \frac{1}{\xi'}\biggr) \dd \xi' + \ln x \biggr]W_{1}(\xi)\\
 &= \ln(x) + \int_{\infty}^{x} \dd \xi' \biggl(\oint_{S} \frac{\dd \xi}{2{\rm i}\pi}\,\frac{W_1(\xi)}{\xi' - \xi} - \frac{1}{\xi'}\biggr)
  = \ln(x) + \int_{\infty}^{x} \dd \xi'\biggl(W_{1}(\xi') - \frac{1}{\xi'}\biggr)\\
 &= \ln(x) + \int_{\infty_{+}}^{z}\biggl(\phi - \frac{\dd X}{X}\biggr).
 \end{align*}
 We can expand this as
 \begin{gather}
 \mathcal{L}\bigl[\ln(x - \bullet)\bigr]
  = \ln x + 2{\rm i}\pi \bm{\epsilon}^* \cdot \bm{u}(z) + \int_{\infty_+}^{z} \biggl(\dd S_{\infty_+,\infty_-} + \frac{\dd \zeta_{\infty_+}}{\zeta_{\infty_+}}\biggr) - \sum_{k = 1}^{d} \frac{t_k}{k} \int_{\infty_+}^{z} \dd B_{\infty_-,k} \nonumber\\
 \qquad{} = \ln x + 2{\rm i}\pi \bm{\epsilon}^* \cdot \bm{u}(z) + \int_{\infty_+}^z \dd_{z'} \ln\biggl(\frac{E(z',\infty_-)\zeta_{\infty_+}(z')}{E(z',\infty_+)}\biggr) - \sum_{k = 1}^{d} \frac{t_k}{k} \int_{\infty_+}^{z} B_{\infty_-,k}\label{Lxnu}\\
 \qquad{} = \ln x + 2{\rm i}\pi \bm{\epsilon}^* \cdot \bm{u}(z) + \ln\biggl(\frac{E(z,\infty_-)\zeta_{\infty_+}(z)}{E(z,\infty_+)E(\infty_+,\infty_-)\dd \zeta_{\infty_+}(\infty_+)}\biggr) - \sum_{k = 1}^{d} \frac{t_k}{k} \int_{\infty_+}^{z} B_{\infty_-,k}\nonumber
 \end{gather}
and $\ln x$ cancels with $\ln \zeta_{\infty_+}(z)$. We introduce the $1$-form on $\widetilde{C}$
\begin{equation}
\label{etadef}\eta(z) = \frac{E(\infty_+,\infty_-)}{E(z,\infty_+)E(z,\infty_-)}
\end{equation}
and use it to get rid of $\infty_-$ in \eqref{Lxnu}. This leads to the claimed formula.

For the second formula
\begin{equation*}
\mathcal{Q}\bigl[\ln(x - \bullet),V\bigr] = \frac{2}{\beta} \oint_{S^2} \frac{\dd \xi_1}{2{\rm i}\pi} \frac{\dd\xi_2}{2{\rm i}\pi}\,V(\xi_1)\biggl[\ln x + \int_{\infty_+}^{x} \biggl(\frac{1}{\xi' - \xi_2} - \frac{1}{\xi'}\biggr)\dd \xi'\biggr] W_2(\xi_1,\xi_2).
\end{equation*}
We move the $\xi_2$-contour to $\infty$. Since $W_2$ has no residue at $\infty$ the first and third term disappear and we get
\begin{align*}
\mathcal{Q}\bigl[\ln(x - \bullet),V] &{} = \frac{2}{\beta} \oint_{S} \frac{\dd \xi_1}{2{\rm i}\pi} V(\xi_1) \int_{\infty_+}^{x} W_2(\xi_1,\xi') \,\dd \xi' \\
&{} =\frac{2}{\beta} \oint_{S} \frac{\dd \xi_1}{2{\rm i}\pi} V(\xi_1) \int_{\infty_+}^{x} \biggl(W_2(\xi_1,\xi') + \frac{1}{(\xi_1 - \xi')^2}\biggr)\dd \xi',
\end{align*}
where in the second line the shift does not affect the integral around $S$ as it is holomorphic near~$S$. Then, we write $X(w) = \xi_1$ and $X(w') = \xi'$, consider these integrals as integrals on $\widetilde{C}$, and recognise via \eqref{BW2} the fundamental bidifferential $B$. Since the path on which we integrate~$w'$ remains in the first sheet away from the cut, we can move the integral over $w$ to surround the pole $\infty_-$ of $V(X(w))$. We get
\begin{align*}
\mathcal{Q}\bigl[\ln(x - \bullet),V] &{} = \frac{2}{\beta} \int_{w' = \infty_+}^{z} \mathop{{\rm Res}}_{w = \infty_-} \dd X(w)V(X(w))B(w,w') \\
&{} = \frac{2}{\beta} \sum_{k = 1}^{d} \frac{t_k}{k} \int_{w' = \infty_+}^{z} \mathop{{\rm \Res}}_{w = \infty_-} \zeta_{\infty_-}(w)^{-k} B(w,w') \\
&{} = \frac{2}{\beta} \sum_{k = 1}^{d} \frac{t_k}{k} \int_{\infty_+}^{z} \dd B_{\infty_-,k},
\end{align*}
as desired.

For the third formula, let $x_i = X(z_i)$. Using Section \ref{sec:second-corr-fund}, we find
\begin{align*}
 \mathcal{Q}\bigl[\ln(x_1 - \bullet), \ln(x_2 - \bullet)\bigr] &{} = \frac{2}{\beta} \oint_{S^2} \frac{\dd \xi_1\dd \xi_2}{(2{\rm i}\pi)^2}\,W_2(\xi_1,\xi_2) \ln(x_1 - \xi_1)\ln(x_2 - \xi_2) \\
 &{} = \frac{2}{\beta} \int_{\infty_+}^{z_1} \int_{\infty_+}^{z_2} \mathcal{W}_2(w_1,w_2)\,\dd X(w_1) \dd X(w_2)
\end{align*}
after we handle the logarithms like in the previous proofs. Then
\begin{align*}
 \mathcal{Q}\bigl[\ln(x_1 - \bullet), \ln(x_2 - \bullet)\bigr] &{} = \frac{2}{\beta} \int_{\infty_+}^{z_1} \int_{\infty_+}^{z_2} \biggl(B(w_1,w_2) - \frac{\dd X(w_1)\dd X(w_2)}{(X(w_1) - X(w_2))^2}\biggr) \\
 &{} = \frac{2}{\beta} \ln\biggl(\frac{E(z_1,z_2)E_0(z_1,\infty_+)E_0(\infty_+,z_2)}{E(z_1,\infty_+)E(\infty_+,z_2)E_0(z_1,z_2)}\biggr),
\end{align*}
where we used Lemma~\ref{IntB} both for $\widetilde{C}$ and $\widehat{\mathbb{C}}$, and \eqref{limEtilde} to get rid of the ratio of the relative prime form with two arguments $\infty_+$. Since the prime forms are antisymmetric in their two variables, we can arrange the formula to have $\infty_+$ always in the second argument. The presence of $\infty_+$ in~$E_0$ factors can be understood by first replacing it by a point $\tilde{z}$, and then letting $\tilde{z} \rightarrow \infty_+$. Due to Lemma~\ref{LemE0in}, the product of the two $E_0$-factors involving $\infty_+$ only gives a sign when we use the local coordinate $\zeta_{\infty_+}$, and using $E_0(z_1,z_2) = (X(z_1) - X(z_2))/\sqrt{\dd X(z_1)\dd X(z_2)}$ leads to the claim.
\end{proof}

We are now in position to evaluate the asymptotics of the kernels. We will mainly be interested in a situation with only two variables
\begin{align}
\nonumber
 \mathcal{K}_{M}^{\frac{K}{M}V}\left(\begin{matrix} c & \tilde{c} \\ x & \tilde{x} \end{matrix}\right) :={}& \bigl\langle \det(x - \Lambda)^c \det(\tilde{x} - \Lambda)^{\tilde{c}}{\rm e}^{\frac{\beta}{2}p{\rm Tr} V(\Lambda)}\bigr\rangle_M^V \\
={} & \bigl\langle{\rm e}^{\sum_{i = 1}^M c\ln(x - \lambda_i) + \tilde{c}\ln(\tilde{x} - \lambda_i) + \frac{\beta}{2}pV(\lambda_i)}\bigr\rangle_M^V,\label{Kmcc}
\end{align}
where $p = M - K$. It will be used in the form
\begin{equation}
\label{Krat}
\frac{Z_M^{\frac{K}{M}V}}{Z_K^V} \bigl\langle\det(x - \Lambda)^c\det(\tilde{x} - \Lambda)^{\tilde{c}}\bigr\rangle_M^{\frac{K}{M}V} = \frac{Z_M^V}{Z_K^V} \mathcal{K}_{M}^{\frac{K}{M}V}\left(\begin{matrix} c & \tilde{c} \\ x & \tilde{x}\end{matrix}\right).
\end{equation}

\begin{Lemma}\label{Kasymint}\allowdisplaybreaks Let $z,\tilde{z} \in C_+$ and $c,\tilde{c} \in \mathbb{Z}$. We have as $K \rightarrow \infty$ and $p$ is a fixed integer:
\begin{gather*}
	  \mathcal{K}^{\frac{K}{K + p}V}_{K + p}\left(\begin{matrix}
										 c & \tilde{c} \\
										 X(z) & X(\tilde{z})
									 \end{matrix}\right) \\
	\qquad{} \sim\exp\biggl\{Kc \mathcal{L}[\ln(X(z) - \bullet)] + K\tilde{c}\mathcal{L}[\ln(X(\tilde{z}) - \bullet)] + c\mathcal{H}[\ln(X(z) - \bullet)]\\
\qquad\quad{} + \tilde{c}\mathcal{H}[\ln(X(\tilde{z}) - \bullet)] + K\frac{\beta}{2}p\mathcal{L}[V] + \frac{\beta}{2}p^2 \mathcal{L}[V] + \frac{\beta}{2}p\mathcal{H}[V] + \frac{\beta^2}{8}p^2 \mathcal{Q}[V,V]\biggr\} \\
\qquad\quad{} \times \exp\{2{\rm i}\pi p \bm{\epsilon}^* \cdot (c \bm{u}(z) + \tilde{c} \bm{u}(\tilde{z}))\} \bigl(\eta(z)E(z,\infty_+)^{2}\dd \zeta_{\infty_+}(\infty_+)\bigr)^{-pc} \\
\qquad\quad{} \times
\bigl(\eta( \tilde{z})E(\tilde{z},\infty_+)^{2}\dd \zeta_{\infty_+}(\infty_+)\bigr)^{-p\tilde{c}} \\
\qquad\quad{} \times\exp\left\{\frac{1}{2}c^2 \mathcal{Q}[\ln(X(z) - \bullet),\ln(X(z) - \bullet)] + \frac{1}{2} \tilde{c}^2 \mathcal{Q}[\ln(X(\tilde{z}) - \bullet),\ln(X(\tilde{z}) - \bullet)]\right\} \\
\qquad\quad{} \times\left(\frac{E(z,\tilde{z})}{E(z,\infty_+)E(\tilde{z},\infty_+)(X(\tilde{z}) - X(z))\dd \zeta_{\infty_+}(\infty_+)}\right)^{\frac{2}{\beta}c\tilde{c}} \\
\qquad\quad{}
 \times \frac{\vartheta_{-(K + p)\bm{\epsilon}^*,\bm{0}}\bigl(\bm{v}_{{\rm eq}} + c\bm{u}(z) + \tilde{c}\bm{u}(\tilde{z}) + \frac{\beta}{2} p(\bm{\tau}(\bm{\epsilon}^{*}) + \bm{u}(\infty_-)) \big| \frac{\beta}{2} \bm{\tau}\bigr)}{\vartheta_{-M\bm{\epsilon}^{*},\bm{0}}\bigl(\bm{v}_{{\rm eq}}\big|\frac{\beta}{2}\bm{\tau}\bigr)}.
\end{gather*}
In particular, if $p = 0$, we have as $M \rightarrow \infty$
\begin{align*}
	\mathcal{K}^{V}_{M}\left(\begin{matrix} c & \tilde{c} \\ X(z) & X(\tilde{z}) \end{matrix}\right) &{}\sim{\rm e}^{Mc \mathcal{L}[\ln(X(z) - \bullet)] + M\tilde{c}\mathcal{L}[\ln(X(\tilde{z}) - \bullet)] + c\mathcal{H}[\ln(X(z) - \bullet)] + \tilde{c}\mathcal{H}[\ln(X(\tilde{z}) - \bullet)]} \\
	& \quad{} \times{\rm e}^{\frac{1}{2}c^2 \mathcal{Q}[\ln(X(z) - \bullet),\ln(X(z) - \bullet)] + \frac{1}{2} \tilde{c}^2 \mathcal{Q}[\ln(X(\tilde{z}) - \bullet),\ln(X(\tilde{z}) - \bullet)]} \\
												 & \quad{} \times \biggl(\frac{E(z,\tilde{z})}{(X(\tilde{z}) - X(z))E(z,\infty_+)E(\tilde{z},\infty_+) \dd \zeta_{\infty_+}(\infty_+)}\biggr)^{\frac{2}{\beta}c\tilde{c}} \\
&\quad{}\times \frac{\vartheta_{-M\bm{\epsilon}^*,\bm{0}}\bigl(\bm{v}_{{\rm eq}} + c\bm{u}(z) + \tilde{c}\bm{u}(\tilde{z}) \big| \frac{\beta}{2} \bm{\tau}\bigr)}{\vartheta_{-(K+p)\bm{\epsilon}^{*},\bm{0}}\bigl(\bm{v}_{{\rm eq}}\big|\frac{\beta}{2}\bm{\tau}\bigr)}.
\end{align*}
\end{Lemma}

\begin{proof}
Let $x = X(z)$ and $\tilde{x} = X(\tilde{z})$. Applying
 Theorem~\ref{thm:clt} to the definition \eqref{Kmcc}, we get as $M,K \rightarrow \infty$ while $ M - K = p$ is fixed
 \begin{gather*}
	 \mathcal{K}^{\frac{K}{M}V}_{M}\left(\begin{matrix} c & \tilde{c} \\ x & \tilde{x} \end{matrix}\right) \\
{} \sim \frac{Z^{V}_{M}}{Z^{\frac{K}{M}V}_{M}}{\rm e}^{Mc\mathcal{L}[\ln(x - \bullet)] + M\tilde{c}\mathcal{L}[\ln(\tilde{x} - \bullet)] + c\mathcal{H}[\ln(x - \bullet)] + \tilde{c}\mathcal{H}[\ln(\tilde{x} - \bullet)] + M\frac{\beta}{2}p\mathcal{L}[V] + \frac{\beta}{2}\mathcal{H}[V] + \frac{\beta^2}{8}p^2\mathcal{Q}[V,V]} \\
 {}\times{\rm e}^{ \frac{1}{2}c^2 \mathcal{Q}[\ln(x - \bullet),\ln(x - \bullet)] + c \tilde{c} \mathcal{Q}[\ln(x - \bullet),\ln(\tilde{x} - \bullet)] + \frac{1}{2}\tilde{c}^2\mathcal{Q}[\ln(\tilde{x} - \bullet),\ln(\tilde{x} - \bullet)]+\frac{\beta}{2}pc \mathcal{Q}[\ln(x - \bullet),V] + \frac{\beta}{2}p\tilde{c}\mathcal{Q}[\ln(\tilde{x} - \bullet),V]}\! \\
 {}\times \frac{\vartheta_{-M\bm{\epsilon}^{*}, \bm{0}}\bigl(\bm{v}_{{\rm eq}} + \frac{1}{2{\rm i}\pi} \oint_{S}\bigl[c\ln(x - \bullet) + \tilde{c}\ln(\tilde{x} - \bullet) + \frac{\beta}{2}p V\bigr] \dd \bm{u} \big| \frac{\beta}{2}\bm{\tau}\bigr)}{\vartheta_{-M\bm{\epsilon}^{*},\bm{0}}\bigl(\bm{v}_{{\rm eq}}\big|\frac{\beta}{2}\bm{\tau}\bigr)}.
\end{gather*}
We split $M = K + p$ in the exponential and combine the newly created
$p$-terms with the $\frac{\beta}{2}\mathcal{Q}[\ln,V]$ terms which is
evaluated thanks to \eqref{simplifLQ}. We also replace the term
$c\tilde{c}\mathcal{Q}[\ln,\ln]$ by its evaluation from
Lemma~\ref{L44}, but refrain from doing so for the $c^2$ and the
$\tilde{c}^2$ terms. Finally, writing the logarithm in the arguments of the
theta function as a primitive, we get an expression in terms of the
Abel map (a similar manipulation was carried out in the proof of Lemma~\ref{lem:compute-nu}), and Lemma~\ref{lem:compute-nu-infty} tells us $\frac{1}{2{\rm i}\pi} \oint_{S} V \dd \boldsymbol{u} = \bm{\tau}(\bm{\epsilon}^*) + \bm{u}(\infty_-)$ which we need to multiply by $\frac{\beta p}{2}$ in the last argument of the theta function. Together with Lemma~\ref{lem:compute-nu-infty},
this implies the claimed formula.
\end{proof}

The ratio of partition functions appearing in \eqref{Krat} will only be needed through the following combination.
\begin{Lemma}\label{ratioZ}
 For fixed $p \in \mathbb{Z}$, we have as $K \rightarrow \infty$
 \begin{gather*}
	 \frac{Z_{K + p}^{V}Z_{K - p}^{V}}{(Z_K^V)^2}\sim{\rm e}^{2p^2 \mathcal{E}[\mu_{{\rm eq}}] + {\rm i}\pi \beta p^2 \bm{\epsilon}^* \cdot \bm{\tau}(\bm{\epsilon}^*)} \\
\hphantom{\frac{Z_{K + p}^{V}Z_{K - p}^{V}}{(Z_K^V)^2}\sim}{}\times\frac{\vartheta_{-K\bm{\epsilon}^*,\bm{0}}\bigl(\bm{v}_{{\rm eq}} - \tfrac{\beta}{2}p\bm{\tau}(\bm{\epsilon}^*)\big|\tfrac{\beta}{2}\bm{\tau}\bigr) \vartheta_{-K\bm{\epsilon}^*,\bm{0}}\bigl(\bm{v}_{{\rm eq}} + \tfrac{\beta}{2} p \bm{\tau}(\bm{\epsilon}^*)\big|\tfrac{\beta}{2}\bm{\tau}\bigr)}{\vartheta_{-K\bm{\epsilon}^*,\bm{0}}\bigl(\bm{v}_{{\rm eq}}\big|\frac{\beta}{2}\bm{\tau}\bigr)^{2}}.
 \end{gather*}
\end{Lemma}
\begin{proof}
Using Theorem~\ref{Asymppart}, we find
\[
\frac{Z_{K + p}^{V}}{Z_K^V} \sim{\rm e}^{(2Kp + p^2)\mathcal{E}[\mu_{{\rm eq}}] + p \mathcal{S}[\mu_{{\rm eq}}] + \frac{\beta}{2}p(\ln K + 1)}\,\frac{\vartheta_{-(K + p)\bm{\epsilon}^*,\bm{0}}\bigl(\bm{v}_{{\rm eq}}\big|\frac{\beta}{2}\bm{\tau}\bigr)}{\vartheta_{-K\bm{\epsilon}^*,\bm{0}}\bigl(\bm{v}_{{\rm eq}}\big|\frac{\beta}{2}\bm{\tau}\bigr)}.
\]

Multiplying this expression and the same with $p \rightarrow -p$ yields
\begin{equation}
 \label{Zdoubler}
	\frac{Z_{K + p}^{V}Z_{K - p}^{V}}{(Z_K^V)^2} \sim{\rm e}^{2p^2\mathcal{E}[\mu_{{\rm eq}}]}\,\frac{\vartheta_{-(K + p)\bm{\epsilon}^*,\bm{0}}\bigl(\bm{v}_{{\rm eq}}\big|\frac{\beta}{2}\bm{\tau}\bigr)\vartheta_{-(K - p)\bm{\epsilon}^*,\bm{0}}\bigl(\bm{v}_{{\rm eq}}\big|\frac{\beta}{2}\bm{\tau}\bigr)}{\vartheta^2_{-K\bm{\epsilon}^*,\bm{0}}\bigl(\bm{v}_{{\rm eq}}\big|\frac{\beta}{2}\bm{\tau}\bigr)}.
\end{equation}
We would like to rewrite all theta functions with a characteristic $-K\bm{\epsilon}^*$ instead of $-(K \pm p)\bm{\epsilon}^*$. For this, we come back to the definitions in Section~\ref{sec:theta-function} and find
\begin{equation}
\label{thetashift}
\vartheta_{-(K + p)\bm{\epsilon}^*,\bm{0}}\bigl(\bm{v}\big|\tfrac{\beta}{2}\bm{\tau}\bigr) ={\rm e}^{{\rm i}\pi \frac{\beta}{2} p^2 \bm{\epsilon}^* \cdot \bm{\tau}(\bm{\epsilon}^*) - 2{\rm i}\pi p \bm{\epsilon}^* \cdot \bm{v}}\,\vartheta_{-K\bm{\epsilon}^*,\bm{0}}\bigl(\bm{v} - p\tfrac{\beta}{2}\bm{\tau}(\bm{\epsilon}^*)\big|\tfrac{\beta}{2}\bm{\tau}\bigr),
\end{equation}
which holds for any $\bm{v} \in \mathbb{C}^g$. We multiply the outcome with the same factor for $p \rightarrow -p$ and obtain
\begin{gather*}
 \vartheta_{-(K + p)\bm{\epsilon}^*,\bm{0}}\bigl(\bm{v}\big|\tfrac{\beta}{2}\bm{\tau}\bigr)\vartheta_{-(K - p)\bm{\epsilon}^*,\bm{0}}\bigl(\bm{v}\big|\tfrac{\beta}{2}\bm{\tau}\bigr) \\
\qquad{} ={\rm e}^{{\rm i}\pi \beta p^2 \bm{\epsilon}^* \cdot \bm{\tau}(\bm{\epsilon}^*)}\vartheta_{-K\bm{\epsilon}^*,\bm{0}}\bigl(\bm{v} - \tfrac{\beta}{2}p\bm{\tau}(\bm{\epsilon}^*)\big|\tfrac{\beta}{2}\bm{\tau}\bigr) \vartheta_{-K\bm{\epsilon}^*,\bm{0}}\bigl(\bm{v} + \tfrac{\beta}{2} p \bm{\tau}(\bm{\epsilon}^*)\big|\tfrac{\beta}{2}\bm{\tau}\bigr).
\end{gather*}
Inserting in \eqref{Zdoubler} yields the claim.
\end{proof}

\section{Derivation of the theta identities}
\label{sec:derivation-formulae}

We shall now state and prove the main theorems. For each case
$\beta \in \{1,2,4\}$, we give both a proof based on the
analysis in the previous sections, and a second, direct proof based on geometric arguments. We recall the definition of the meromorphic $1$-form $\eta$ on $\widetilde{C}$ that appeared in~\eqref{etadef}
\[
\eta(z) = \frac{E(\infty_+,\infty_-)}{E(z,\infty_+)E(z,\infty_-)}.
\]
Throughout this section, the Abel map will always be based at $\infty_+$, and $\bm{u}(\infty_-)$ is computed using a path from $\infty_+$ to $\infty_-$ which does not cross any of the representatives $(\mathcal{A}_h,\mathcal{B}_h)_{h = 1}^{g}$ obtained by analytic continuation from the ones in Section~\ref{sec:constr-spectr-curve}, see the discussion above Lemma~\ref{lem:holo-periods}.

\subsection[The beta = 2 formula]{The $\boldsymbol{\beta = 2}$ formula}

\begin{Theorem}\label{thm:formula_beta_2}
 Consider a hyperelliptic curve $\widehat{C}$, and let
 $z,z',w,w' \in \widetilde{C}$, and
 $\bm{\mu}, \bm{\nu} \in \mathbb{R}^g$. Then, we have
 \begin{gather}
  (X(w) - X(z'))(X(z) - X(w'))\frac{E(z,w)E(z',w')}{E(w,z')E(z,w')}\nonumber\\
  \qquad \quad{}\times\frac{\vartheta_{\bm{\mu},\bm{\nu}}\bigl(\bm{u}(z) - \bm{u}(z') + \bm{u}(w) - \bm{u}(w')\big|\bm{\tau}\bigr)}{E(z,z')E(w,w')}\vartheta_{\bm{\mu}, \bm{\nu}}(\bm{0}\big|\bm{\tau}) \nonumber\\
  \qquad\quad {}-(X(z) - X(w))(X(z') - X(w'))\frac{\vartheta_{\bm{\mu},\bm{\nu}}(\bm{u}(z) - \bm{u}(z')\big|\bm{\tau}\bigr)}{E(z,z')}\frac{\vartheta_{\bm{\mu},\bm{\nu}}\bigl(\bm{u}(w) - \bm{u}(w')\big|\bm{\tau}\bigr)}{E(w,w')}\nonumber\\
  \qquad {}= E(z,w)E(z',w') \biggr( \prod_{p \in \{z, z', w, w'\}} \eta(p)\biggr)\vartheta_{\bm{\mu},\bm{\nu}}\bigl(\bm{u}(z) + \bm{u}(w) -\bm{u}(\infty_-)\big|\bm{\tau}\bigr)\nonumber\\
  \qquad \quad{}\times\vartheta_{\bm{\mu},\bm{\nu}}\bigl(- \bm{u}(z') - \bm{u}(w') + \bm{u}(\infty_-)\big|\bm{\tau}\bigr).\label{eq:formula-beta-2}
 \end{gather}
\end{Theorem}
Notice that each of the three terms is actually a $\frac{1}{2}$-form on the universal cover $\widetilde{C}$ in each of the variable $z$, $z'$, $w$, $w'$. Even if $\bm{\mu} = 0$, it does not descend to the curve itself. As the proof will show, the formula holds equally well for arbitrary $\bm{\mu}$, $\bm{\nu}$ complex, in particular, if we shift all arguments of the theta functions by the same arbitrary but common $\bm{\nu} \in \mathbb{C}^g$.

\begin{proof}
Consider first the case where the Weierstra\ss{} points of $\widehat{C}$ are
 real. By Proposition~\ref{propallhyp}, we can find a $\beta$-ensemble whose spectral curve has $\widehat{C}$ as underlying Riemann surface and for which the results of Section~\ref{sec:expans-part} apply (by construction the potential is off-critical).
 We then express the identity of Theorem~\ref{beta2form} for $m_1 = m_2 = 1$, for $x = X(z)$,
 $x' = X(z')$, $\tilde{x} = X(w)$ and $\tilde{x}' = X(w')$ pairwise distinct points in $\mathbb{C} \setminus A$, that determine unique points $z,z',w,w' \in C_+$. Taking into account the
 definition of $\mathcal{K}$ in \eqref{Krat}, we get
\begin{gather}
 \mean{ \frac{\det(x - \Lambda)\det(\tilde{x} - \Lambda)}{\det(x' - \Lambda) \det(\tilde{x}' - \Lambda)}}_{N}^{V}\nonumber \\
 \qquad{} = (x - x')(\tilde{x} - \tilde{x}')\frac{N}{N + 1} \frac{Z_{N - 1}^{V}Z_{N + 1}^{V}}{(Z_N^V)^2} \mathcal{K}^{\frac{N}{N - 1}V}_{N - 1}\left(\begin{matrix} 1 & 1 \\ x & \tilde{x} \end{matrix}\right) \mathcal{K}^{\frac{N}{N + 1}V}_{N + 1}\left(\begin{matrix} - 1 & -1 \\ \tilde{x}' & x' \end{matrix}\right)\nonumber\\
  \qquad\quad{} + \mathcal{K}_{N}^{V}\left(\begin{matrix} 1 & -1 \\ x & x' \end{matrix}\right)\mathcal{K}_{N}^{V}\left(\begin{matrix} 1 & -1 \\ \tilde{x} & \tilde{x}' \end{matrix}\right).\label{smallrelb2}
\end{gather}
Let us first consider the asymptotics of left-hand side as $N \rightarrow \infty$. We have $\bm{v}_{{\rm eq}} = 0$ since $\beta = 2$. Coming back to Theorem~\ref{thm:clt} and using Lemma~\ref{L44} for the $\mathcal{Q}$-terms involving two different variables, we obtain
\begin{gather}
 \mean{ \frac{\det(x - \Lambda)\det(\tilde{x} - \Lambda)}{\det(x' - \Lambda) \det(\tilde{x}' - \Lambda)}}_{N}^{V} \nonumber\\
\qquad {} \sim{}{\rm e}^{N\mathcal{L}\bigl[\ln\bigl(\frac{(x - \bullet)(\tilde{x} - \bullet)}{(x' - \bullet)(\tilde{x}' - \bullet)}\bigr)\bigr] + \mathcal{H}\bigl[\ln\bigl(\frac{(x - \bullet)(\tilde{x} - \bullet)}{(x' - \bullet)(\tilde{x}' - \bullet)}\bigr)\bigr] + \frac{1}{2} \sum_{\xi \in \{x,\tilde{x},x',\tilde{x}'\}} \mathcal{Q}[\ln(\xi - \bullet),\ln(\xi - \bullet)]} \nonumber\\
\qquad\quad{} \times{}\frac{E(z,w)E(z',w')E(z,\infty_+)E(z',\infty_+)E(w,\infty_+)E(w',\infty_+)\bigl(\dd \zeta_{\infty_+}(\infty_+)\bigr)^2}{E(z,z')E(z,w')E(w,z')E(w,w')}\nonumber\\
\qquad\quad{} \times{}\frac{(x - x')(x - \tilde{x}')(\tilde{x} - x')(\tilde{x} - \tilde{x}')}{(x - \tilde{x})(x' - \tilde{x}')} \nonumber\\
\qquad\quad{} \times{} \frac{\vartheta_{-N\bm{\epsilon}^*,\bm{0}}\bigl(\bm{u}(z) - \bm{u}(z') + \bm{u}(w) - \bm{u}(w')\big|\bm{\tau}\bigr)}{\vartheta_{-N\bm{\epsilon}^*,\bm{0}}\bigl(\bm{0}\big|\bm{\tau}\bigr)}.\label{T1}
\end{gather}

For the asymptotics of the first term of the right-hand side of \eqref{smallrelb2}, we use the $\beta = 2$ specialisation of Lemma~\ref{ratioZ} with $p = 1$ for the ratio of partition functions, and Lemma~\ref{Kasymint} for the $\mathcal{K}$-factors with $K = N$, $M = N \mp 1$ and $(c,\tilde{c}) = (\pm 1,\pm 1)$, that is $p = \mp 1$. The outcome is
\begin{gather}
  (x - x')(\tilde{x} - \tilde{x}')\frac{N}{N + 1} \frac{Z_{N - 1}^{V}Z_{N + 1}^{V}}{(Z_N^V)^2} \mathcal{K}^{\frac{N}{N - 1}V}_{N - 1}\left(\begin{matrix} 1 & 1 \\ x & \tilde{x} \end{matrix}\right) \mathcal{K}^{\frac{N}{N + 1}V}_{N + 1}\left(\begin{matrix} -1 & -1 \\ \tilde{x}' & x' \end{matrix}\right) \nonumber \\
 \qquad \sim{\rm e}^{2\mathcal{E}[\mu_{{\rm eq}}] + 2\mathcal{L}[V] + \mathcal{Q}[V,V] - 2{\rm i}\pi \bm{\epsilon}^* \cdot(\bm{u}(z) + \bm{u}(z') + \bm{u}(w) + \bm{u}(w'))} \nonumber \\
\qquad\quad {}\times{\rm e}^{N\mathcal{L}\bigl[\ln\bigl(\frac{(x - \bullet)(\tilde{x} - \bullet)}{(x' - \bullet)(\tilde{x}' - \bullet)}\bigr)\bigr] + \mathcal{H}\bigl[\ln\bigl(\frac{(x - \bullet)(\tilde{x} - \bullet)}{(x' - \bullet)(\tilde{x}' - \bullet)}\bigr)\bigr] + \frac{1}{2} \sum_{\xi \in \{x,x',\tilde{x},\tilde{x}'\}} \mathcal{Q}[\ln(\xi - \bullet),\ln(\xi - \bullet)]} \nonumber\\
\qquad\quad {}  \times \frac{E(z,w)E(z',w')(x - x')(\tilde{x} - \tilde{x}')}{(x - \tilde{x})(x' - \tilde{x}')} \prod_{p \in \{z,z',w,w'\}} \eta(p)E(p,\infty_+) \sqrt{\dd \zeta_{\infty_+}(\infty_+)}\nonumber\\
\qquad\quad {}  \times \frac{1}{\vartheta^2_{-N\bm{\epsilon}^{*},\bm{0}} (\bm{0} |\bm{\tau} )}\bigl(\vartheta_{-(N - 1)\bm{\epsilon}^*,\bm{0}} (\bm{u}(z) + \bm{u}(w) - \bm{u}(\infty_-) - \bm{\tau}(\bm{\epsilon}^{*})  | \bm{\tau} )\nonumber\\
\qquad\quad {} \times
\vartheta_{-(N + 1)\bm{\epsilon}^*,\bm{0}} ( - \bm{u}(z') - \bm{u}(w') + \bm{u}(\infty_-) + \bm{\tau}(\bm{\epsilon}^{*})  | \bm{\tau} )\bigr).\label{Line2}
\end{gather}
In the last line, we can restore a characteristic $-N\bm{\epsilon}^*$ in the theta functions thanks to \eqref{thetashift}, which we have to use respectively with $\bm{v} = \bm{u}(z) + \bm{u}(w) - \bm{u}(\infty_-) - \bm{\tau}(\bm{\epsilon}^{*})$ and $\bm{v} = - \bm{u}(z') - \bm{u}(w') + \bm{u}(\infty_-) + \bm{\tau}(\bm{\epsilon}^{*})$. The outcome is
\begin{gather*}
	 \vartheta_{-(N - 1)\bm{\epsilon}^*,\bm{0}} ( \bm{u}(z) + \bm{u}(w) - \bm{u}(\infty_-) - \bm{\tau}(\bm{\epsilon}^{*})  | \bm{\tau} )\\
\qquad\quad{}\times\vartheta_{-(N + 1)\bm{\epsilon}^*,\bm{0}} ( - \bm{u}(z') - \bm{u}(w') + \bm{u}(\infty_-) + \bm{\tau}(\bm{\epsilon}^*)  | \bm{\tau} ) \\
	\qquad{} ={}{\rm e}^{2{\rm i}\pi \bm{\epsilon}^* \cdot \bm{\tau}(\bm{\epsilon}^*) + 2{\rm i}\pi \bm{\epsilon}^*\cdot(\bm{u}(z) + \bm{u}(z') + \bm{u}(w) + \bm{u}(w') - 2\bm{u}(\infty_-) - 2\bm{\tau}(\bm{\epsilon}^*))} \\
	\qquad\quad{} \times{} \vartheta_{-N\bm{\epsilon}^*,\bm{0}} ( \bm{u}(z) + \bm{u}(w) - \bm{u}(\infty_-)  | \bm{\tau} )\vartheta_{-N\bm{\epsilon}^*,\bm{0}} ( - \bm{u}(z') - \bm{u}(w') + \bm{u}(\infty_-)  | \bm{\tau} ).
\end{gather*}
Inserting this in \eqref{Line2} gives an exponential prefactor
\[
{\rm e}^{2\mathcal{E}[\mu_{{\rm eq}}] + 2\mathcal{L}[V] + \mathcal{Q}[V,V] + 4{\rm i}\pi \bm{\epsilon}^* \cdot( \bm{\tau}(\bm{\epsilon}^*) + \bm{u}(\infty_-))},
\]
which would be equal to $1$ if we had Proposition~\ref{lem:Ef0}. We will establish Proposition~\ref{lem:Ef0} as a~byproduct in Section~\ref{sec:can}; for the moment, we proceed assuming it holds. We get
\begin{gather}
 (x - x')(\tilde{x} - \tilde{x}')\frac{N}{N + 1} \frac{Z_{N - 1}^{V}Z_{N + 1}^{V}}{(Z_N^V)^2} \mathcal{K}^{\frac{N}{N - 1}V}_{N - 1}\left(\begin{matrix} 1 & 1 \\ x & \tilde{x} \end{matrix}\right) \mathcal{K}^{\frac{N}{N + 1}V}_{N + 1}\left(\begin{matrix} -1 & -1 \\ \tilde{x}' & x' \end{matrix}\right) \nonumber\\
 \qquad{}\sim{\rm e}^{N\mathcal{L}\bigl[\ln\bigl(\frac{(x - \bullet)(\tilde{x} - \bullet)}{(x' - \bullet)(\tilde{x}' - \bullet)}\bigr)\bigr] + \mathcal{H}\bigl[\ln\bigl(\frac{(x - \bullet)(\tilde{x} - \bullet)}{(x' - \bullet)(\tilde{x}' - \bullet)}\bigr)\bigr] + \frac{1}{2} \sum_{\xi \in \{x,x',\tilde{x},\tilde{x}'\}} \mathcal{Q}[\ln(\xi - \bullet),\ln(\xi - \bullet)]} \nonumber\\
 \qquad\quad{}\times \frac{E(z,w)E(z',w')(x - x')(\tilde{x} - \tilde{x}')}{(x - \tilde{x})(x' - \tilde{x}')} \prod_{p \in \{z,z',w,w'\}} \eta(p)E(p,\infty_+)\sqrt{\dd \zeta_{\infty_+}(\infty_+)} \nonumber\\
\qquad\quad{} \times \frac{\vartheta_{-N\bm{\epsilon}^*,\bm{0}} (\bm{u}(z) + \bm{u}(w) - \bm{u}(\infty_-) |\bm{\tau} )\vartheta_{-N \bm{\epsilon}^*,\bm{0}} ( - \bm{u}(z') - \bm{u}(w') + \bm{u}(\infty_+) |\bm{\tau} )}{\vartheta_{-N\bm{\epsilon}^*}^2 (\bm{0} |\bm{\tau} )}.\label{T2}
\end{gather}

The asymptotics of the second term is simpler as we just need to use the $p = 0$ case of Lemma~\ref{Kasymint}. The outcome is
\begin{gather}
\mathcal{K}_{N}^{V}\left(\begin{matrix} 1 & -1 \\ x & x' \end{matrix}\right)\mathcal{K}_{N}^{V}\left(\begin{matrix} 1 & -1 \\ \tilde{x} & \tilde{x}' \end{matrix}\right) \nonumber\\
\qquad {} \sim{\rm e}^{N\mathcal{L}\bigl[\ln\bigl(\frac{(x - \bullet)(\tilde{x} - \bullet)}{(x' - \bullet)(\tilde{x}' - \bullet)}\bigr)\bigr] + \mathcal{H}\bigl[\ln\bigl(\frac{(x - \bullet)(\tilde{x} - \bullet)}{(x' - \bullet)(\tilde{x}' - \bullet)}\bigr)\bigr] + \frac{1}{2} \sum_{\xi \in \{x,x',\tilde{x},\tilde{x}'\}} \mathcal{Q}[\ln(\xi - \bullet),\ln(\xi - \bullet)]} \nonumber\\
\qquad\quad{} \times \frac{(x - x')(\tilde{x} - \tilde{x}')}{E(z,z')E(w,w')} \prod_{p \in \{z,z',w,w'\}} E(p,\infty_+)\sqrt{\dd \zeta_{\infty_+}(\infty_+)} \nonumber\\
\qquad\quad{} \times \frac{\vartheta_{-N\bm{\epsilon}^*,\bm{0}} (\bm{u}(z) - \bm{u}(z')  |\bm{\tau} )\vartheta_{-N\bm{\epsilon}^*,\bm{0}} (\bm{u}(w) - \bm{u}(w') |\bm{\tau} )}{\vartheta_{-N\bm{\epsilon}^*,\bm{0}}^2 (\bm{0} |\bm{\tau} )}.\label{T3}
\end{gather}

We also observe that \eqref{T1}, \eqref{T2} and \eqref{T3} all have the same exponential factor involving~$\mathcal{L}$,~$\mathcal{H}$ and~$\mathcal{Q}$, the same factor
\[
\prod_{p \in \{z,z',w,w'\}} E(p,\infty_+)\sqrt{\dd \zeta_{\infty_+}(\infty_+)},
\]
and the same squared $\vartheta$ in the denominator, except for \eqref{T1} where the latter is not squared. After we cancel those and multiply further by
\[
\frac{(x - \tilde{x})(x' - \tilde{x}')}{(x - x')(\tilde{x} - \tilde{x}')},
\]
the identity \eqref{smallrelb2} as $N \rightarrow \infty$ then becomes
\begin{gather}
 \frac{E(z,w)E(z',w')\,(x - \tilde{x}')(\tilde{x} - x')}{E(z,z')E(z,w')E(w,z')E(w,w')}\vartheta_{-N\bm{\epsilon}^*,\bm{0}} (\bm{u}(z) - \bm{u}(z') + \bm{u}(w) - \bm{u}(w') |\bm{\tau} )\vartheta_{-N\bm{\epsilon}^*,\bm{0}} (\bm{0} |\bm{\tau} ) \nonumber\\
 \qquad{}= E(z,w)E(z',w')\eta(z)\eta(z')\eta(w)\eta(w')\vartheta_{-N\bm{\epsilon}^*,\bm{0}} (\bm{u}(z) + \bm{u}(w) - \bm{u}(\infty_-) |\bm{\tau} )\nonumber\\
  \qquad\quad{}\times\vartheta_{-N\bm{\epsilon}^*,\bm{0}} (- \bm{u}(z') - \bm{u}(w') + \bm{u}(\infty_-) |\bm{\tau} ) \nonumber\\
 \qquad\quad{}+ \frac{(x - \tilde{x})(x' - \tilde{x}')}{E(z,z')E(w,w')}\,\vartheta_{-N\bm{\epsilon}^*,\bm{0}} (\bm{u}(z) - \bm{u}(z') |\bm{\tau} )\vartheta_{-N\bm{\epsilon}^*,\bm{0}} (\bm{u}(w) - \bm{u}(w') |\bm{\tau} ) + o(1).\label{pre-id}
\end{gather}

Let $\bm{\mu} \in \mathbb{R}^g$. Let us assume temporarily that
$\epsilon_{1}^*, \ldots, \epsilon_{g}^*$ are $\mathbb{Q}$-linearly independent. Then,
by Kronecker's theorem, one can find an increasing sequence $(N^{(n)})_{n \geq 1}$
such that
\begin{equation*}
 \lim_{n \rightarrow \infty}\bigl(N^{(n)} \bm{\epsilon}^* + \lfloor -N^{(n)}\bm{\epsilon}^* \rfloor\bigr) = - \bm{\mu},
\end{equation*}
where the integer part is applied to each component of the vector.
Using $N = N^{(n)}$ in \eqref{pre-id} and letting $n \rightarrow \infty$
we get the same identity without the $o(1)$ and with characteristic
$\bm{\mu}$, $\bm{0}$ for all theta functions, which is \eqref{eq:formula-beta-2} with $\bm{\nu} = \bm{0}$ and pairwise distinct points $z,z',w,w' \in C_+$.

If $\epsilon_1^*,\ldots,\epsilon_g^*$ are not $\mathbb{Q}$-linearly independent, thanks to Corollary~\ref{codense} we can take a sequence of $\beta$-ensembles whose spectral curve admits as underlying Riemann surface a hyperelliptic curve~$\widehat{C}^{(n)}$ with real Weierstra\ss{} points converging to those of $\widehat{C}$, and whose filling fractions are $\mathbb{Q}$-linearly independent. By the previous argument, we know \eqref{eq:formula-beta-2} with arbitrary $\bm{\mu} \in \mathbb{R}^g$ and $\bm{\nu} = \bm{0}$ and $z,z',w,w' \in C^{(n)}_+$ pairwise distinct. Since all the members of this identity are continuous in the real Weierstra\ss{} points $a_0 < b_0 < \cdots < a_g < b_g$ while the values $X(z),X(z'),X(w),X(w') \in \mathbb{C}$ are fixed away from them, taking $n \rightarrow \infty$ shows that the formula also holds for $\widehat{C}$, $\bm{\nu} = 0$ and pairwise distinct points $z,z',w,w' \in C_+$. Let us call $(\star)$ this formula. From there, we can derive the desired identity in full generality by using repeatedly analytic continuation, as follows.

Firstly, all terms in ($\star$) are holomorphic functions of $\bm{\mu} \in \mathbb{C}^{g}$, and the identity holds for real $\bm{\mu}$. Therefore, it must hold as well for complex $\bm{\mu}$. In particular, we can replace $\bm{\mu}$ with $\bm{\mu} + \bm{\tau}^{-1}(\bm{\nu})$ for arbitrary $\bm{\mu},\bm{\nu} \in \mathbb{R}^{g}$, and rewrite all theta functions as
\[
\vartheta_{\bm{\mu} + \bm{\tau}(\bm{\nu}),\bm{0}}(\bm{z}|\bm{\tau}) ={\rm e}^{{\rm i}\pi\bm{\tau}^{-1}(\bm{\nu}) \cdot \bm{\nu} + 2{\rm i}\pi \bm{\tau}^{-1}(\bm{\nu}) \cdot \bm{z}} \vartheta_{\bm{\mu},\bm{\nu}}(\bm{z}|\bm{\tau}).
\]
The resulting phase is common to the three terms of the identity, therefore ($\star$) is valid for arbitrary $\bm{\mu},\bm{\nu} \in \mathbb{C}^{g}$, and a fortiori for arbitrary real $\bm{\mu}$, $\bm{\nu}$.

Secondly, fix $R > 0$ large enough, and let $\widetilde{\bm{\Delta}}_{2g + 2}(R)$ be the subset of points in our para\-meter space of marked hyperelliptic curves $\widetilde{\bm{\Delta}}_{2g + 2}$ such that the $X$-image of the Weierstra\ss{} points have moduli $\leq R$. Having fixed and pairwise distinct values $x,x',\tilde{x},\tilde{x}' \in \mathbb{C}$ such that $\max(|x|,|x'|,|\tilde{x}|,|\tilde{x}'|) \geq 2R$ determines a unique quadruple of analytic sections $z,z',w,w'$: $\widetilde{\bm{\Delta}}_{2g + 2}(R) \rightarrow \hat{\bm{C}}$ such that $x = X(z)$, $x' = X(z')$, $\tilde{x} = X(w)$ and $\tilde{x}' = X(w')$. These sections represent points in the (moving with parameters) hyperelliptic curve. Due to Lemma~\ref{lem:holo-periods} and the discussion preceding it, all terms in ($\star$) are holomorphic functions on $\widetilde{\bm{\Delta}}_{2g + 2}(R)$, but since we know that ($\star$) holds in the connected component of the base point in the real locus of $\widetilde{\bm{\Delta}}_{2g + 2}(R)$, it must also hold over the whole $\widetilde{\bm{\Delta}}_{2g + 2}(R)$ with $\max(|x|,|x'|,|\tilde{x}|,|\tilde{x}'|) \geq 2R$. Now rather fixing a~marked hyperelliptic curve corresponding to a point in $\widetilde{\bm{\Delta}}_{2g + 2}(R)$, the identity ($\star$) is valid for points~$z$,~$z'$,~$w$,~$w'$ in a neighborhood of $\infty_+$, but since it can be seen as an identity involving only meromorphic functions of $z,z',w,w' \in \widetilde{C}$, it must hold for arbitrary quadruple of points in~$\widetilde{C}$. Eventually as $R$ is arbitrary in the argument, we get~\eqref{eq:formula-beta-2} over the whole parameter space~$\widetilde{\bm{\Delta}}_{2g + 2}$ and quadruples of points in the universal cover of the associated hyperelliptic curve.
\end{proof}

This formula implies the Fay identity, as we now show.
\begin{prop}\label{corol:fay-from-beta=2}
 Formula \eqref{eq:formula-beta-2} implies the Fay identity \eqref{Fayid} for hyperelliptic curves.
\end{prop}
\begin{proof}\allowdisplaybreaks
 The key remark is that $z$ and $w$ play almost symmetric roles in
 \eqref{eq:formula-beta-2}. We write the two equations obtained when exchanging $z$ and $w$, specialised to $\bm{\mu} = 0$ and $\bm{\nu} \in \mathbb{C}^g$ arbitrary but rather transferred to the argument of the theta functions, so that everything is expressed in terms of $\theta = \vartheta_{\bm{0},\bm{0}}$:  \begin{gather*}
	 (X(w) - X(z'))(X(z) - X(w'))\frac{E(z,w)E(z',w')}{E(w,z')E(z,w')}\\
\qquad\quad{}\times \frac{\theta (\bm{\nu} + \bm{u}(z) - \bm{u}(z') + \bm{u}(w) - \bm{u}(w') |\bm{\tau} )}{E(z,z')E(w,w')}\theta (\bm{\nu} |\bm{\tau} )\\
\qquad\quad{} -(X(z) - X(w))(X(z') - X(w')) \frac{\theta(\bm{\nu} + \bm{u}(z) - \bm{u}(z') |\bm{\tau} )}{E(z,z')}\frac{\theta (\bm{\nu} + \bm{u}(w) - \bm{u}(w') |\bm{\tau} )}{E(w,w')}\\
\qquad{} = E(z,w)E(z',w')\eta(z)\eta(z')\eta(w)\eta(w')\theta (\bm{\nu} + \bm{u}(z) + \bm{u}(w) - \bm{u}(\infty_-) |\bm{\tau} )\\
\qquad\quad{}\times \theta (\bm{\nu} - \bm{u}(z') - \bm{u}(w') + \bm{u}(\infty_-) |\bm{\tau} )\\
\qquad{}  = (X(z) - X(z'))(X(w) - X(w'))\frac{E(z,w)E(z',w')}{E(z,z')E(w,w')}\\
\qquad\quad{}\times	 \frac{\theta (\bm{\nu} + \bm{u}(z) - \bm{u}(z') + \bm{u}(w) - \bm{u}(w') |\bm{\tau} )}{E(w,z')E(z,w')} \theta (\bm{\nu} |\bm{\tau} )\\
\qquad\quad{}  -(X(z)- X(w))(X(z') - X(w'))\frac{\theta (\bm{\nu} + \bm{u}(w) - \bm{u}(z') |\bm{\tau} )}{E(w,z')}\frac{\theta (\bm{\nu} + \bm{u}(z) - \bm{u}(w') |\bm{\tau} )}{E(z,w')}.
 \end{gather*}

 Subtracting the first member from the third member of the equalities, grouping the terms together and dividing by $(X(z) - X(w))(X(z') - X(w'))$ yields the identity
 \begin{gather*}
 0 =\frac{E(z,w)E(z',w')}{E(w,z')E(z,w')E(z,z')E(w,w')} \theta (\bm{\nu} + \bm{u}(z) - \bm{u}(z') + \bm{u}(w) - \bm{u}(w') |\bm{\tau} )\theta (\bm{\nu} |\bm{\tau} ) \\
\hphantom{0 =}{}
- \frac{\theta (\bm{\nu} + \bm{u}(w) - \bm{u}(z') |\bm{\tau} )\theta (\bm{\nu} + \bm{u}(z) - \bm{u}(w') |\bm{\tau} )}{E(w,z')E(z,w')}\\
\hphantom{0 =}{}
+ \frac{\theta (\bm{\nu} + \bm{u}(z) - \bm{u}(z') |\bm{\tau} )\theta (\bm{\nu} + \bm{u}(w) - \bm{u}(w') |\bm{\tau} )}{E(z,z')E(w,w')},
 \end{gather*}
 which is exactly the Fay identity \eqref{Fayid} after we replace the prime form with its expression \eqref{Eprimtheta} and take $(z_1,z_2,z_3,z_4) = (z,w,z',w')$.
 \end{proof}

Finally, we provide a direct proof of \eqref{eq:formula-beta-2} based on complex analysis.
This proof is based on a classical theorem of Riemann which we recall for the
convenience of the reader.

\begin{Theorem}[{\cite[Theorem 3.1]{Tata1}}]\label{thm:Riemann-zero}
 There is a vector $\bm{k} \in \C^{g}$ such that for all
 $\boldsymbol{\nu}' \in \C^{g}$, the function $z \mapsto \theta (\bm{\nu}' + \bm{u}(z) | \bm{\tau} )$ of $z \in \widetilde{C}$ either vanishes identically or has $g$ zeroes $w_1,\ldots,w_g$ in a~fundamental domain, satisfying
\[
\sum_{h=1}^{g}\bm{u}(w_{h}) = - \bm{\nu} + \bm{k} \mod \mathbb{L}.
\]
\end{Theorem}
The vector $\bm{k}$ is called vector of Riemann constants.

\begin{proof}[Direct geometric proof of Theorem~\ref{thm:formula_beta_2}]
It suffices to prove the identity for $\bm{\mu} = 0$, since we still have $\bm{\nu} \in \mathbb{C}^g$ arbitrary which allow reconstructing arbitrary characteristics. Let $\bm{\nu}' \in \C^{g}$. Riemann's
 theorem \ref{thm:Riemann-zero} implies that seen as a function of $z \in \widetilde{C}$, the theta
 function $z \mapsto \theta (\bm{\nu}' + \bm{u}(z) |\bm{\tau} )$ is either identically zero of has $g$ zeroes in a fundamental domain. Let $\mathcal{D}_{\bm{\nu}'}$ be its zero divisor. We apply this to $\bm{\nu}' = \bm{\nu} + \bm{u}(w) - \bm{u}(z') - \bm{u}(w')$ with $\bm{\nu}$, $z'$, $w$, $w'$ generic such that it is not in the theta divisor. A classical consequence of Riemann's theorem
 \cite[Theorem~VI.3.3]{Farkas-Kra} is that meromorphic functions on $C$ with pole divisor at most $\mathcal{D}_{\bm{\nu}'}$ are constant. For convenience, write
 \begin{gather*}
	 \mathfrak{c}_{1} = \frac{(X(w) - X(z'))(X(z) - X(w'))E(z,w)E(z',w')}{E(w,z')E(z,w')E(z,z')E(w,w')},\\
	 \mathfrak{c}_{2} = -\frac{(X(z) - X(w))(X(z') - X(w'))}{E(z,z')E(w,w')},\qquad
	 \mathfrak{c}_{3} = E(z,w)E(z',w')\eta(z)\eta(z')\eta(w)\eta(w')
 \end{gather*}
and consider
\begin{align*}
 \Psi(z) ={}& \frac{\mathfrak{c}_{2}}{\mathfrak{c}_1}\,\frac{\theta(\bm{\nu} + \bm{u}(z) - \bm{u}(z') |\bm{\tau} )\theta (\bm{\nu} + \bm{u}(w) - \bm{u}(w') |\bm{\tau} )}{\theta (\bm{\nu} + \bm{u}(z) - \bm{u}(z') + \bm{u}(w) - \bm{u}(w') |\bm{\tau} )\theta (\bm{\nu} |\bm{\tau} )} \\
 &{}+ \frac{\mathfrak{c}_{3}}{\mathfrak{c}_1}\,\frac{\theta (\bm{\nu} +\bm{u}(z) + \bm{u}(w) - \bm{u}(\infty_-) |\bm{\tau} )\theta (\bm{\nu}- \bm{u}(z') - \bm{u}(w') + \bm{u}(\infty_-) |\bm{\tau} )}{\theta (\bm{\nu} + \bm{u}(z) - \bm{u}(z') + \bm{u}(w) - \bm{u}(w') |\bm{\tau} )\theta (\bm{\nu} |\bm{\tau} )}.
\end{align*}
This is a meromorphic function of $z \in \widehat{C}$. We have seen that the theta function in the denominator has $g$ zeroes,
which are thus poles of $\Psi$. We now consider the poles of the
 other factors: the zeroes of $\mathfrak{c}_{1}$ and the poles of
 $\mathfrak{c}_{2}$ and $\mathfrak{c}_{3}$. The coefficient $\mathfrak{c}_1$ has a simple zero at $z = \jmath(w')$, where $\jmath$ is the hyperelliptic involution, and at $z = w$. The coefficients $\mathfrak{c}_2$ and $\mathfrak{c}_3$ have simple poles at $z = z'$. Accordingly, both ratios $\frac{\mathfrak{c}_{2}}{\mathfrak{c}_{1}}$ and $\frac{\mathfrak{c}_{3}}{\mathfrak{c}_{1}}$ have a simple pole only at
 $z = \jmath(w')$. A careful computation of the residues taking into account $\bm{u}(\jmath(w')) = \bm{u}(\infty_-) - \bm{u}(w')$ shows that $\Psi$ does not have a pole when $z = \jmath(w')$. Notice that there are no pole at $\infty_{\pm}$ as poles coming from linear terms are cancelled by other linear terms or the form $\eta$.
 We conclude that the divisor of poles of $\Psi$ is at most $\mathcal{D}_{\bm{\nu}'}$. Thus, it is a constant function of $z$. A similar argument with the other variables would show that it is a constant function of $z$, $z'$, $w$, $w'$. By sending the points~$z'$,~$w$,~$w'$ to~$z$ one after the other, and we arrive to $\Psi = 1$.
\end{proof}

\subsection[The beta = 1 formula]{The $\boldsymbol{\beta = 1}$ formula}

\begin{Theorem}\label{thm:formula_beta_1}
 Consider a marked hyperelliptic curve $\widehat{C}$ and let
 $z_1,z_1',z_2,z_2' \in \widetilde{C}$ and $\bm{\mu}, \bm{\nu} \in \mathbb{R}^g$. Writing
 $x_i = X(z_i)$ and $x_i' = X\big(z_i'\big)$, we have
 \begin{gather}
  \biggl(\frac{E(z_1, z_2)E\bigl(z'_1, z_2'\bigr)}{E\bigl(z_1, z_1'\bigr)E\bigl(z_1, z_2'\bigr)E\bigl(z_2, z_1'\bigr)E\bigl(z_2, z_2'\bigr)}\biggr)^{2}\nonumber\\
  \qquad \quad{}\times\vartheta_{\bm{\mu}, \bm{\nu}}\bigl( \bm{u}(z_1) - \bm{u}\bigl(z'_1\bigr) + \bm{u}(z_2) - \bm{u}\bigl(z_2'\bigr)\big|\tfrac{\bm{\tau}}{2}\bigr)\vartheta_{\bm{\mu},\bm{\nu}}\bigl(\bm{0}\big|\tfrac{\bm{\tau}}{2}\bigr)\nonumber \\
 \qquad {}= \frac{(x_1 - x_2)\bigl(x'_1 - x_2'\bigr)}{\bigl(x_1 - x_1'\bigr)\bigl(x_2 - x_2'\bigr)}\frac{\vartheta_{\bm{\mu},\bm{\nu}}\bigl( \bm{u}(z_1) - \bm{u}\bigl(z_2'\bigr)\big|\tfrac{\bm{\tau}}{2}\bigr)}{E\bigl(z_1,z_2'\bigr)^{2}}\frac{\vartheta_{\bm{\mu},\bm{\nu}}\bigl(\bm{u}(z_2) - \bm{u}\bigl(z_1'\bigr)\big|\tfrac{\bm{\tau}}{2}\bigr)}{E\bigl(z_1',z_2\bigr)^{2}}\nonumber\\
 \qquad\quad{}-\frac{(x_1 - x_2)\bigl(x_1' - x_2'\bigr)}{\bigl(x_1 - x_2'\bigr)\bigl(x_2 - x_1'\bigr)}\frac{\vartheta_{\bm{\mu},\bm{\nu}}\bigl(\bm{u}(z_1) - \bm{u}\bigl(z_1'\bigr)\big|\tfrac{\bm{\tau}}{2}\bigr)}{E\bigl(z_1, z'_1\bigr)^{2}}\frac{\vartheta_{\bm{\mu},\bm{\nu}}(\bm{u}(z_2) - \bm{u}\bigl(z_2'\bigr)\big|\tfrac{\bm{\tau}}{2}\bigr)}{E\bigl(z_2', z_2\bigr)^{2}}\nonumber\\
 \qquad\quad{}+\frac{\bigl(E(z_1,z_2)E\bigl(z'_1, z'_2\bigr) \eta(z_1)\eta\bigl(z_1'\bigr)\eta(z_2)\eta\bigl(z_2'\bigr)\bigr)^2}{\bigl(x_1 - x'_1\bigr)\bigl(x_1 -x_2'\bigr)\bigl(x_2 - x'_1\bigr)\bigl(x_2 - x_2'\bigr)} \nonumber\\
 \qquad \quad{} \times \vartheta_{\bm{\mu},\bm{\nu}}\bigl( \bm{u}(z_1) + \bm{u}(z_2) - \bm{u}(\infty_-)\big|\tfrac{\bm{\tau}}{2}\bigr) \vartheta_{\bm{\mu},\bm{\nu}}\bigl(-\bm{u}\bigl(z'_1\bigr) - \bm{u}\bigl(z'_2\bigr) + \bm{u}(\infty_-)\big|\tfrac{\bm{\tau}}{2}\bigr).\label{eq:formula-beta-1}
 \end{gather}
\end{Theorem}
\begin{proof}
 The strategy is similar to the proof for $\beta = 2$ in Theorem~\ref{thm:formula_beta_2}. In particular, we first prove an asymptotic identity for hyperelliptic curves arising from $\beta = 1$ ensembles, use approximations to get arbitrary characteristic $\bm{\mu}$, $\bm{0}$, and then analytic continuation to get the identity for marked hyperelliptic curves with arbitrary complex Weierstra\ss{} points and characteristics $\bm{\mu},\bm{\nu} \in \mathbb{R}^g$.

The starting point is the exact identity of Theorem~\ref{th:idbeta1} in the simplest
 non-trivial case, i.e., $m = 2$ (Pfaffian of size $4$). Taking $x_1,x_1',x_2,x_2' \in \mathbb{C} \setminus A$ pairwise distinct, this gives
 \begin{gather}
	  \mean{\frac{\det(x_1- \Lambda)\det(x_2- \Lambda)}{\det(x_1'- \Lambda)\det(x_2'- \Lambda)}}_{2N}^{V} \nonumber\\
	\qquad {}  = \frac{2N(2N - 1)}{(2N + 2)(2N + 1)} (x_1- x_1')(x_2- x_2')(x_1- x_2')(x_2- x_1') \nonumber\\
\qquad \quad{}\times\frac{Z_{2N - 2}^{V}Z_{2N + 2}^{V}}{(Z_{2N}^{V})^2} \mathcal{K}_{2N - 2}^{\frac{2N}{2N - 2}V}\left(\begin{matrix} 1 & 1 \\ x_1& x_2\end{matrix}\right)\mathcal{K}_{2N + 2}^{\frac{2N}{2N + 2}V}\left(\begin{matrix} -1 & -1 \\ x'_1& x_2'\end{matrix}\right) \nonumber\\
\qquad\qquad{}	 - \frac{(x_1- x_2')(x_2- x_1')}{(x_1- x_2)(x_1'- x_2')} \mathcal{K}_{2N}^V\left(\begin{matrix} 1 & -1 \\ x_1& x_1'\end{matrix}\right)\mathcal{K}_{2N}^{V}\left(\begin{matrix} 1 & -1 \\ x_2& x_2'\end{matrix}\right) \nonumber\\
\qquad\quad{}+ \frac{(x_1- x_1')(x_2- x_2')}{(x_1- x_2)(x_1'- x_2')}\mathcal{K}_{2N}^{V}\left(\begin{matrix} 1 & -1 \\ x_1& x_2'\end{matrix}\right)\mathcal{K}_{2N}^{V}\left(\begin{matrix} 1 & -1 \\ x_2& x_1'\end{matrix}\right). \label{beta1un}
 \end{gather}
 As the computation is similar to $\beta = 2$ we only streamline it. Let $z_1,z_1',z_2,z_2' \in C_+$ be such that $x_i = X(z_i)$ and $x_i' = X(z_i')$. The asymptotic equivalent
 of the left-hand side of \eqref{beta1un} as $N \rightarrow \infty$ is obtained from Theorem~\ref{thm:clt}:
\begin{gather}
	{\rm e}^{(2N\mathcal{L} + \mathcal{H})\bigl[\ln\bigl(\frac{(x_1- \bullet)(x_2- \bullet)}{(x_1'- \bullet)(x_2'- \bullet)}\bigr)\bigr] + \frac{1}{2}\sum_{\xi \in \{x_1,x_1',x_2,x_2'\}} \mathcal{Q}[\ln(\xi - \bullet),\ln(\xi - \bullet)]} \nonumber\\
\qquad{}	 \times \biggl(\prod_{p \in \{z_1,z'_1,z_2,z_2'\}} E(p,\infty_+)\sqrt{\dd \zeta_{\infty_+}(\infty_+)}\biggr)^2 \nonumber \\
\qquad{}	 \times \Biggl(\frac{E(z_1,z_2)E\bigl(z'_1,z_2'\bigr)}{E\bigl(z_1,z_1'\bigr)E\bigl(z_1,z_2'\bigr)E\bigl(z_2,z'_1\bigr)E\bigl(z_2,z_2'\bigr)}\,\frac{\bigl(x_1 - x_1'\bigr)\bigl(x_1 - x_2'\bigr)\bigl(x_2 - x_1'\bigr)\bigl(x_2 - x_2'\bigr)}{(x_1 - x_2)\bigl(x_1' - x_2'\bigr)}\Biggr)^2 \nonumber\\
\qquad{}	 \times \frac{\vartheta_{-2N\bm{\epsilon}^*,\bm{0}}\bigl(\bm{v}_{{\rm eq}} + \bm{u}(z_1) - \bm{u}\bigl(z'_1\bigr) + \bm{u}(z_2) - \bm{u}\bigl(z'_2\bigr)\big|\frac{\bm{\tau}}{2}\bigr)}{\vartheta_{-2N\bm{\epsilon}^*,\bm{0}}\bigl(\bm{v}_{{\rm eq}}\big|\frac{\bm{\tau}}{2}\bigr)}.\label{firstune}
\end{gather}
We have used Lemma~\ref{L44} to evaluate the $\mathcal{Q}[\ln(\xi - \bullet),\ln(\xi' - \bullet)]$ terms that appear with $\xi \neq \xi'$. In the right-hand side of \eqref{beta1un}, we use the asymptotics of the $2$-point kernel from Lemma~\ref{Kasymint} with $K = 2N$ and $p = \mp 2$.

Consider the asymptotics of the first term in the right-hand side of \eqref{beta1un}. It contains a~product of theta functions with characteristic $-(2N \pm 2)\bm{\epsilon}^*$, which we can replace by two theta functions with same characteristic $-2N\bm{\epsilon}^*$ using \eqref{firstune} up to an extra exponential factor. The latter combines with the asymptotics of the ratio of partition functions of shifted size from Lemma~\ref{ratioZ} to reproduce a factor
\[
{\rm e}^{8\mathcal{E}[\mu_{{\rm eq}}] + 4\mathcal{L}[V] + \mathcal{Q}[V,V] + 8{\rm i}\pi \bm{\epsilon}^* \cdot (\bm{\tau}(\bm{\epsilon}^*) + \bm{u}(\infty_-))}
\]
which is equal to $1$ due to Proposition~\ref{lem:Ef0}, and to kill the factor of ${\rm e}^{-4{\rm i}\pi \bm{\epsilon}^*\cdot(\bm{u}(z_1)+\bm{u}(z_2) + \bm{u}(z_1') + \bm{u}(z_2'))}$ coming from the use of \eqref{simplifLQ}. The other factors are the first line of \eqref{firstune} multiplied by
\begin{gather*}
 \Biggl(\frac{E(z_1,z_2)E\bigl(z_1',z_2'\bigr)}{(x_1 - x_2)\bigl(x_1' - x_2'\bigr)} \prod_{p \in \{z_1,z'_1,z_2,z_2'\}} \eta(p)E(p,\infty_+)\sqrt{\dd\zeta_{\infty_+}(\infty_+)}\Biggr)^2 \\
 \times \vartheta_{-2N\bm{\epsilon}^*,\bm{0}}\bigl(\bm{v}_{{\rm eq}} - \bm{\tau}(\epsilon^*)\big|\tfrac{\bm{\tau}}{2}\bigr) \vartheta_{-2N\bm{\epsilon}^*,\bm{0}}\bigl(\bm{v}_{{\rm eq}} + \bm{\tau}(\bm{\epsilon}^*)\big|\tfrac{\bm{\tau}}{2}\bigr) \\
 \times \vartheta_{-2N\bm{\epsilon}^*,\bm{0}}\bigl(\bm{v}_{{\rm eq}} + \bm{u}(z_1) + \bm{u}(z_2) - \bm{u}(\infty_-)\big|\tfrac{\bm{\tau}}{2}\bigr) \vartheta_{-2N\bm{\epsilon}^*,\bm{0}}\bigl(\bm{v}_{{\rm eq}} - \bm{u}\bigl(z_1'\bigr) - \bm{u}\bigl(z_2'\bigr) + \bm{u}(\infty_-)\big|\tfrac{\bm{\tau}}{2}\bigr),
\end{gather*}
where the last line was already explained.

The asymptotics of the second and third terms in the right-hand side of \eqref{beta1un} are more straightforward to get. They both contain the first line and the third line of \eqref{firstune}, and the other asymptotic factors are
\begin{gather*}
 - \frac{\bigl(x_1 - x_2'\bigr)\bigl(x_2 - x_1'\bigr)}{(x_1 - x_2)\bigl(x_1' - x_2'\bigr)}\Biggl(\frac{\bigl(x_1 - x_1'\bigr)\bigl(x_2 - x_2'\bigr)}{E\bigl(z_1,z_1'\bigr)E\bigl(z_2,z_2'\bigr)} \prod_{p \in \{z_1,z_1',z_2,z_2'\}} E(p,\infty_+)\sqrt{\dd \zeta_{\infty_+}(\infty_+)}\Biggr)^2 \\
 \qquad\quad{}\times \vartheta_{-2N\bm{\epsilon}^*,\bm{0}}\bigl(\bm{v}_{{\rm eq}} + \bm{u}(z_1) - \bm{u}\bigl(z_1'\bigr)\big|\tfrac{\bm{\tau}}{2}\bigr) \vartheta_{-2N\bm{\epsilon}^*,\bm{0}}\bigl(\bm{v}_{{\rm eq}} + \bm{u}(z_2) - \bm{u}\bigl(z_2'\bigr)\big|\tfrac{\bm{\tau}}{2}\bigr)
\end{gather*}
for the second term (including its sign), and
\begin{gather*}
 \frac{\bigl(x_1 - x_1'\bigr)\bigl(x_2 - x_2'\bigr)}{(x_1 - x_2)\bigl(x_1' - x_2'\bigr)} \Biggl(\frac{\bigl(x_1 - x_2'\bigr)\bigl(x_2 - x_1'\bigr)}{E\bigl(z_1,z_2'\bigr)E\bigl(z_2,z_1'\bigr)} \prod_{p \in \{z_1,z_1',z_2,z_2'\}} E(p,\infty_+) \sqrt{\dd\zeta_{\infty_+}(\infty_+)} \Biggr)^2 \\
 \qquad{}\times \vartheta_{-2N\bm{\epsilon}^*,\bm{0}}\bigl(\bm{v}_{{\rm eq}} + \bm{u}(z_1) - \bm{u}\bigl(z_2'\bigr)\big|\tfrac{\bm{\tau}}{2}\bigr) \vartheta_{-2N\bm{\epsilon}^*,\bm{0}}\bigl(\bm{v}_{{\rm eq}} + \bm{u}(z_2) - \bm{u}\bigl(z_1'\bigr)\big|\tfrac{\bm{\tau}}{2}\bigr).
\end{gather*}
We then divide all terms by the first and second line of \eqref{firstune} (common
factor to all terms) and~by
\[
\Biggl(\frac{\bigl(x_1 - x_1'\bigr)\bigl(x_1 - x_2'\bigr)\bigl(x_2 - x_1'\bigr)\bigl(x_2 - x_2'\bigr)}{(x_1 - x_2)\bigl(x_1' - x_2'\bigr)}\Biggr)^2
\]
to arrive to \eqref{eq:formula-beta-1} with $\bm{\mu} = -2N \bm{\epsilon}^*$, $\bm{\nu} = 0$ and an extra $\bm{v}_{{\rm eq}}$ added to the arguments of all theta functions. Then, we repeat the end of the proof of Theorem~\ref{thm:formula_beta_2} to get exactly and in full generality the claimed \eqref{thm:formula_beta_2}.
\end{proof}

Theorem~\ref{thm:formula_beta_1} can be reformulated as an identity for theta functions with matrix of periods $\bm{\tau}$ instead of $\frac{\bm{\tau}}{2}$.

\begin{Lemma}
There is an equivalence between \eqref{eq:formula-beta-1} for any $\bm{\mu},\bm{\nu} \in \mathbb{R}^g$, and the formula
\begin{gather}
	0  = \mathfrak{c}_{1} \vartheta_{\frac{\bm{\alpha}}{2}, \bm{0}}\bigl(\bm{u}(z_1) - \bm{u}\bigl(z_1'\bigr) + \bm{u}(z_2) - \bm{u}\bigl(z_2'\bigr)\big|\bm{\tau}\bigr) \nonumber\\ \hphantom{0 =}{}
+ \mathfrak{c}_{2} \vartheta_{\frac{\bm{\alpha}}{2},\bm{0}}\bigl(\bm{u}(z_1) + \bm{u}\bigl(z_1'\bigr) - \bm{u}(z_2) - \bm{u}\bigl(z_2'\bigr)\big|\bm{\tau}\bigr) \nonumber\\ \hphantom{0 =}{}
+\mathfrak{c}_{3} \vartheta_{\frac{\bm{\alpha}}{2},\bm{0}}\bigl(\bm{u}(z_1) - \bm{u}\bigl(z'_1\bigr) - \bm{u}(z_2) + \bm{u}\bigl(z_2'\bigr)\big|\bm{\tau}\bigr) \nonumber\\ \hphantom{0 =}{}
+ \mathfrak{c}_{4} \vartheta_{\frac{\bm{\alpha}}{2},\bm{0}}\bigl(\bm{u}(z_1) + \bm{u}\bigl(z_1'\bigr)+\bm{u}(z_2)+\bm{u}\bigl(z_2'\bigr) - 2\bm{u}(\infty_-)\big|\bm{\tau}\bigr)\label{eq:equiv-beta-1}
\end{gather}
for any $\bm{\alpha} \in \mathbb{Z}^{g}/2\mathbb{Z}^{g}$, where
\begin{gather*}
\mathfrak{c}_1  = -\Biggl(\frac{E(z_1, z_2)E\bigl(z_1', z_2'\bigr)}{E\bigl(z_1, z_1'\bigr)E\bigl(z_1, z_2'\bigr)E\bigl(z_2, z_1'\bigr)E\bigl(z_2, z_2'\bigr)}\Biggr)^{2}, \\
\mathfrak{c}_2  = \frac{(x_1-x_2)\big(x_1'-x_2'\big)}{\bigl(x_1-x_1'\bigr)\bigl(x_2-x_2'\bigr)}\frac{1}{\bigl(E\bigl(z_1,z_2'\bigr)E\big(z_1', z_2\big)\bigr)^2}, \\
\mathfrak{c}_3  = -\frac{(x_1 - x_2)\big(x_1' -x_2'\big)}{\big(x_1 -x_2'\big)\bigl(x_2-x_1'\bigr)}\frac{1}{\bigl(E\bigl(z_1, z_1'\bigr)E\big(z_2', z_2\big)\bigr)^2}, \\
\mathfrak{c}_4  = \frac{\bigl(E(z_1, z_2)E\bigl(z_1', z_2'\bigr)\eta(z_1)\eta\bigl(z_1'\bigr)\eta(z_2)\eta\bigl(z_2'\bigr)\bigr)^{2}}{\bigl(x_1-x_1'\bigr)\bigl(x_1-x_2'\bigr)\bigl(x_2-x_1'\bigr)\bigl(x_2-x_2'\bigr)}.
\end{gather*}
\end{Lemma}
\begin{proof}
The trick is to use Riemann's binary addition theorem, see, e.g., \cite[equation~(6.6)]{Tata1}. It states that for any $\bm{\mu},\bm{\nu},\bm{\mu}',\bm{\nu}' \in \mathbb{R}^g$ and $\bm{z}_1,\bm{z}_2 \in \mathbb{C}^g$
\begin{gather}
	 \vartheta_{\bm{\mu}, \bm{\nu}}\bigl(\bm{z}_{1} + \bm{z}_{2}\big| \tfrac{\bm{\tau}}{2}\bigr)\vartheta_{\bm{\mu}', \bm{\nu}'}\bigl(\bm{z}_{1} - \bm{z}_{2}\big| \tfrac{\bm{\tau}}{2}\bigr) \nonumber\\
\qquad{}	 = \sum_{\bm{\alpha} \in \Z^{g}/2\Z^{g}}\vartheta_{\frac{\bm{\mu}+\bm{\mu'}+\bm{\alpha}}{2} , \bm{\nu} + \bm{\nu}'} (2\bm{z}_{1} |\bm{\tau} )\vartheta_{\frac{\bm{\mu}-\bm{\mu}'+\bm{\alpha}}{2} , \bm{\nu} - \bm{\nu}'} (2\bm{z}_{2} |\bm{\tau} ).\label{binaryR}
\end{gather}
We apply the transformation \eqref{binaryR} with $\bm{\mu}' = \bm{\mu}$ and $\bm{\nu} = \bm{\nu}'$ to each term in \eqref{eq:formula-beta-1}, writing it in the equivalent form
\begin{gather*}
\sum_{\bm{\alpha} \in \mathbb{Z}^{g}/2\mathbb{Z}^{g}} \vartheta_{\bm{\mu} + \frac{\bm{\alpha}}{2},\bm{\nu}}\bigl(\bm{u}(z_1) - \bm{u}\bigl(z_1'\bigr) + \bm{u}(z_2) - \bm{u}\bigl(z_2'\bigr)\big|\bm{\tau}\bigr) \Biggl(\sum_{i = 1}^{4} \mathfrak{c}_i \vartheta_{\frac{\bm{\alpha}}{2},\bm{0}}\bigl(\bm{w}_i\bigl(z_1,z_1',z_2,z_2'\bigr)\big|\bm{\tau}\bigr)\Biggr),
\end{gather*}
where the $\bm{w}_{i}$ are exactly the four arguments of the theta functions appearing in \eqref{eq:equiv-beta-1}. In this form, the converse implication is clear. The direct implication follows from the observation that~$\bm{\nu}$ is arbitrary, and the family of functions $T_{\bm{\alpha}}(\bm{\nu}) = \vartheta_{\frac{\bm{\alpha}}{2},\bm{0}}\bigl(\bm{u}(z_1) - \bm{u}\bigl(z_1'\bigr) + \bm{u}(z_2) - \bm{u}\bigl(z_2'\bigr)\big|\bm{\tau}\bigr)$ indexed by $\bm{\alpha} \in \mathbb{Z}^{g}/2\mathbb{Z}^{g}$ are linearly independent, forcing the sum inside the bracket to vanish for each individual $\bm{\alpha}$.
\end{proof}

This is an identity involving only Riemann theta functions, for which we can offer a direct geometric proof, in a slightly more general form.

\begin{Theorem}
Equation~\eqref{eq:equiv-beta-1} holds for any marked hyperelliptic curve for any characteristic $\bm{\mu},\bm{\nu} \in \mathbb{R}^g$ $($instead of just half-integer characteristics$)$.
\end{Theorem}
\begin{proof}[Direct geometric proof of Theorem~\ref{thm:formula_beta_1}]
 The strategy is similar of the direct proof of Theorem~\ref{thm:formula_beta_2}. We start without lack of generality to set $\bm{\mu} = 0$ and for $\bm{\nu} \in \mathbb{C}^g$ arbitrary, consider
 \begin{align*}
	 \Psi(z_1)
	 &= \frac{\mathfrak{c}_{2}}{\mathfrak{c}_{1}}\frac{\theta\bigl(\bm{\nu} + \bm{u}(z_1) + \bm{u}\bigl(z_1'\bigr) - \bm{u}(z_2) - \bm{u}\bigl(z_2'\bigr)\big|\bm{\tau}\bigr)
		}{\theta\bigl(\bm{\nu} + \bm{u}(z_1) - \bm{u}\bigl(z'_1\bigr) + \bm{u}(z_2) - \bm{u}\bigl(z_2'\bigr)\big|\bm{\tau}\bigr)}\\
	 & \quad +\frac{\mathfrak{c}_{3}}{\mathfrak{c}_{1}}\frac{\theta\bigl(\bm{\nu} + \bm{u}(z_1) - \bm{u}\bigl(z_1'\bigr) - \bm{u}(z_2) + \bm{u}\bigl(z_2'\bigr)\big|\bm{\tau}\bigr)
		}{\theta\bigl(\bm{\nu} + \bm{u}(z_1) - \bm{u}\bigl(z_1'\bigr) + \bm{u}(z_2) - \bm{u}\bigl(z_2'\bigr)\big|\bm{\tau}\bigr)}\\
	 & \quad +\frac{\mathfrak{c}_{4}}{\mathfrak{c}_{1}}\frac{\theta\bigl(\bm{\nu} + \bm{u}(z_1) + \bm{u}\bigl(z_1'\bigr)+\bm{u}(z_2)+\bm{u}\bigl(z_2'\bigr) - 2\bm{u}(\infty_-)\big|\bm{\tau}\bigr)}{\theta\bigl(\bm{\nu} + \bm{u}(z_1) - \bm{u}\bigl(z_1'\bigr) + \bm{u}(z_2) - \bm{u}\bigl(z_2'\bigr)\big|\bm{\tau}\bigr)}.
 \end{align*}
This is a meromorphic function of $z_1 \in \widehat{C}$. We first analyse the poles that may come from the ratios of coefficients. The ratio
 \smash{$\frac{\mathfrak{c}_{2}}{\mathfrak{c}_{1}}$} has simple poles at
 $z_1 = z_2$ and $z_1 = \jmath\bigl(z'_1\bigr)$, where $\jmath$ is the hyperelliptic involution. The
 ratio \smash{$\frac{\mathfrak{c}_{3}}{\mathfrak{c}_{1}}$} has simple poles only at
 $z_1 = \jmath\bigl(z_2'\bigr)$ and $z_1 = z_2$. The ratio
 \smash{$\frac{\mathfrak{c}_{4}}{\mathfrak{c}_{1}}$} has simple poles only at
 $z_1 = \jmath\bigl(z_2'\bigr)$ and $z_1 = \jmath\bigl(z_2'\bigr)$. However, careful computation of the
 residues show that $\Psi$ has none of these poles. Thus, the only poles of $\Psi$ are the zeros of $z_1 \mapsto \theta\bigl(\bm{\nu} + \bm{u}(z_1) - \bm{u}\bigl(z_1'\bigr) + \bm{u}(z_2) - \bm{u}\bigl(z_2'\bigr)\big|\bm{\tau}\bigr)$. As in the direct proof of Theorem~\ref{thm:formula_beta_2}, Riemann's theorem implies that if we choose the points $z_1'$, $z_2$, $z_2'$ and the vector $\bm{\nu}$ generically, there are no nonconstant meromorphic function whose poles are the zeroes of this theta function. We deduce that $\Psi(z_1)$ does not depend on $z_1$. A similar argument shows that $\Psi(z_1)$ is independent of all points
 points $z_1$, $z_1'$, $z_2$, $z_2'$. Sending $z_1'$, $z_2$, $z_2'$ successively to~$z_1$, we find that the constant is~$1$.
\end{proof}

\subsection[The beta = 4 formula]{The $\boldsymbol{\beta = 4}$ formula}

The case $\beta = 4$ has the same structure as the $\beta = 1$ case of Theorem~\ref{thm:formula_beta_1}, except that the argument of the theta functions are doubled while we use the matrix $2\bm{\tau}$. This similarity is already manifest in the exact formulae of Theorems~\ref{th:idbeta1} and \ref{th:idbeta4}.

\begin{Theorem}\label{thm:formula_beta_4}
 Consider a marked hyperelliptic curve $\widehat{C}$, and let
 $z_1,z_1',z_2,z_2' \in \widetilde{C}$ and $\bm{\mu}, \bm{\nu} \in \mathbb{C}^g$. Writing $x_i = X(z_i)$ and $x_i' = X(z_i')$, we have
 \begin{gather*}
  \Biggl(\frac{E(z_1, z_2)E\bigl(z_1', z_2'\bigr)}{E\bigl(z_1, z_1'\bigr)E\bigl(z_1, z_2'\bigr)E\bigl(z_2, z_1'\bigr)E\bigl(z_2, z_2'\bigr)}\Biggr)^{2}\\
  \qquad\quad{}\times\vartheta_{\bm{\mu}, \bm{\nu}}\bigl(2\bigl(\bm{u}(z_1) - \bm{u}\bigl(z'_1\bigr) + \bm{u}(z_2) - \bm{u}\bigl(z_2'\bigr)\bigr)\big|2\bm{\tau}\bigr)\vartheta_{\bm{\mu},\bm{\nu}}\bigl(\bm{0}\big|2\bm{\tau}\bigr) \\
  \qquad\quad{}-\frac{(x_1-x_2)\bigl(x_1'-x_2'\bigr)}{\bigl(x_1-x_1'\bigr)\bigl(x_2-x_2'\bigr)}\frac{\vartheta_{\bm{\mu},\bm{\nu}}\bigl( 2\bigl(\bm{u}(z_1) - \bm{u}\bigl(z_2'\bigr)\bigr)\big|2\bm{\tau}\bigr)}{E\bigl(z_1,z_2'\bigr)^{2}}\frac{\vartheta_{\bm{\mu},\bm{\nu}}\bigl(2\bigl(\bm{u}(z_2) - \bm{u}\bigl(z_1'\bigr)\bigr)\big|2\tau\bigr)}{E\bigl(z_1',z_2\bigr)^{2}}\\
 \qquad\quad{}+\frac{(x_1-x_2)\bigl(x_1'-x_2'\bigr)}{\bigl(x_1-x_2'\bigr)\bigl(x_2 - x_1'\bigr)}\frac{\vartheta_{\bm{\mu},\bm{\nu}}\bigl(2\bigl(\bm{u}(z_1) - \bm{u}\bigl(z'_1\bigr)\bigr)\big|2\bm{\tau}\bigr)}{E\bigl(z_1, z_1'\bigr)^{2}}\frac{\vartheta_{\bm{\mu},\bm{\nu}}(2\bigl(\bm{u}(z_2) - \bm{u}\bigl(z_2'\bigr)\bigr)\big|2\bm{\tau}\bigr)}{E\bigl(z_2',z_2\bigr)^{2}}\\
 \qquad{}=\frac{\bigl(E(z_1, z_2)E\bigl(z_1', z_2'\bigr) \eta(z_1)\eta(z_2)\eta\bigl(z_1'\bigr)\eta\bigl(z_2'\bigr)\bigr)^2}{\bigl(x_1-x_1'\bigr)\bigl(x_1-x_2'\bigr)\bigl(x_2-x_1'\bigr)\bigl(x_2-x_2'\bigr)} \\
  \qquad\quad{}\times \vartheta_{\bm{\mu},\bm{\nu}}\bigl(2(\bm{u}(z_1) + \bm{u}(z_2) - \bm{u}(\infty_-))\big|2\bm{\tau}\bigr) \vartheta_{\bm{\mu},\bm{\nu}}\bigl(2(-\bm{u}\bigl(z'_1\bigr) - \bm{u}\bigl(z_2'\bigr) + \bm{u}(\infty_-))\big|2\bm{\tau}\bigr),
 \end{gather*}
\end{Theorem}

\begin{proof}\allowdisplaybreaks
The starting point is the simplest non-trivial identity of Theorem~\ref{th:idbeta4}, namely $m = 2$ (Pfaffian of size $4$), which gives
\begin{gather*}
 \Biggl\langle\frac{\det(x_1 - \Lambda)^2\det(x_2 - \Lambda)}{\det\bigl(x_1' - \Lambda\bigr)^2\det\bigl(x_2' - \Lambda\bigr)^2}\Biggr\rangle_{N}^{V} \\
 \qquad{}= \frac{N}{N + 1}\bigl(x_1 - x_1'\bigr)\bigl(x_2 - x_2'\bigr)\bigl(x_1 - x_2'\bigr)\bigl(x_2 - x_1'\bigr) \\
 \qquad\quad{}\times  \frac{Z_{N + 1}^{V}Z_{N - 1}^{V}}{(Z_N^V)^2}\,\mathcal{K}_{N - 1}^{\frac{N}{N - 1}V}\left(\begin{matrix} 2 & 2 \\ x_1 & x_2 \end{matrix}\right)\mathcal{K}_{N + 1}^{\frac{N}{N + 1}V}\left(\begin{matrix} -2 & -2 \\ x_1' & x_2' \end{matrix}\right) \\
\qquad\quad{} - \frac{\bigl(x_1 - x_2'\bigr)\bigl(x_2 - x_1'\bigr)}{(x_1 - x_2)\bigl(x_1' - x_2'\bigr)} \mathcal{K}_{N}^{V}\left(\begin{matrix} 2 & -2 \\ x_1 & x_1' \end{matrix}\right)\mathcal{K}_{N}^{V}\left(\begin{matrix} 2 & -2 \\ x_2 & x_2' \end{matrix}\right) \\
\qquad\quad{}
 + \frac{\bigl(x_1 - x_1'\bigr)\bigl(x_2 - x_2'\bigr)}{(x_1 - x_2)\bigl(x_1' - x_2'\bigr)} \mathcal{K}_{N}^{V}\left(\begin{matrix} 2 & -2 \\ x_1 & x_2' \end{matrix}\right)\mathcal{K}_N^{V}\left(\begin{matrix} 2 & -2 \\ x_2 & x_1' \end{matrix}\right).
\end{gather*}
We omit the details of the asymptotic analysis based on Lemmas~\ref{Kasymint} and \ref{ratioZ}: it is very similar to the $\beta = 1$ case. Instead of using them for $K = 2N$, $p = \pm 2$ and $c,\tilde{c} \in \{-1,1\}$, now we rather use them with $K = N$ and $p = \pm 1$ and $c,\tilde{c} \in \{- 2,2\}$.
\end{proof}

\begin{Lemma}
Theorem~$\ref{thm:formula_beta_4}$ is equivalent to Theorem~$\ref{thm:formula_beta_1}$.
\end{Lemma}
\begin{proof} We apply Theorem~\ref{thm:formula_beta_1} to the hyperelliptic curve with
 matrix of periods $\bm{\tau}' = -\bm{\tau}^{-1}$. Then, \eqref{eq:formula-beta-1} is an identity
 involving theta functions with matrix
 $\frac{\bm{\tau'}}{2} = -\frac{\bm{\tau}^{-1}}{2}$. On the other hand, the
 modular transformation of the theta function is (see
 \cite[equation~(5.1)]{Tata1}), for any $\bm{z},\bm{\mu},\bm{\nu} \in \mathbb{R}^g$
 \begin{equation*}
 \begin{split}
	 \vartheta_{\bm{\nu}, -\bm{\mu}}\bigl(\bm{z}\big | -\tfrac{\bm{\tau}^{-1}}{2}\bigr) = D_{\bm{\tau}} \cdot \ee^{2{\rm i}\pi \bm{z} \cdot \bm{\tau}^{-1}(\bm{z})} \vartheta_{\bm{\mu}, \bm{\nu}}(2\bm{z} | 2\bm{\tau})
 \end{split}
 \end{equation*}
for some constant $D_{\bm{\tau}} \in \mathbb{C}^*$. Applying this to each term in Theorem~\ref{thm:formula_beta_4}, all terms get the same prefactor and we are left with Theorem~\ref{thm:formula_beta_4}. The operation is reversible.
\end{proof}

\subsection{Formula for the multi-cut equilibrium energy (Proof of Proposition~\ref{lem:Ef0})}
\label{sec:can}
In the proof of Theorem~\ref{thm:formula_beta_2}, if we did not use Proposition~\ref{lem:Ef0} to simplify the exponential in~\eqref{Line2}, the rest of the arguments would prove the identity \eqref{eq:formula-beta-2} with a prefactor
\begin{equation}
\label{aphsee}{\rm e}^{2\mathcal{E}[\mu_{{\rm eq}}] + 2\mathcal{L}[V] + \mathcal{Q}[V,V] + 4{\rm i}\pi \bm{\epsilon}^* \cdot (\bm{\tau}(\bm{\epsilon}^*) + \bm{u}(\infty_-))}
\end{equation}
in the right-hand side, valid for any hyperelliptic curve with real Weierstra\ss{} points and the equilibrium measure $\mu_{{\rm eq}}$ of the associated (unconstrained) $\beta = 2$ ensemble. Taking all points $z$, $z'$, $w$, $w'$ to $\infty_+$ in this modified identity implies that this extra factor \eqref{aphsee} must be equal to $1$. The argument of the exponential is manifestly real, except perhaps or the last term. As the curve is hyperelliptic, a basis of the space of holomorphic forms is given by $\dd \pi_k = \frac{x^{k} \dd x}{s}$ for $k \in [g]$. Recall that $s$ takes imaginary values on the segments $[a_h,b_h]$ for each $h \in [0,g]$, and real values between the segments. This implies that the matrix $Q_{k,h} = \oint_{\mathcal{A}_{h}} \dd \pi_k$ has purely imaginary entries. Since $(\dd u_h)_{h = 1}^{g}$ is the basis dual to $\mathcal{A}$-cycle integration, we have
\[
\dd u_h = \sum_{k = 1}^{g} Q^{-1}_{h,k} \dd \pi_k,\qquad \text{with}\ Q^{-1}\ \text{purely imaginary}.
\]
Integrating this on the $\mathcal{B}$-cycles which only run between segments (Section~\ref{sec:constr-spectr-curve}) yields a purely imaginary matrix of periods $\bm{\tau}$. A path from $\infty_+$ to $\infty_-$ that does not cross any of the $\mathcal{A}$- and $\mathcal{B}$-cycles described in Section~\ref{sec:constr-spectr-curve} is for instance the path travelling along the real axis in $\widehat{C}_+$ from $-\infty$ to $a_0$, then along the real axis in $\widehat{C}_-$ from $a_0$ to $-\infty_-$. In this range $s$ is real-valued, so $\bm{u}(\infty_-)$ is also purely imaginary. All in all, \eqref{aphsee} only involves the real exponential, and we conclude that
\[
2\mathcal{E}[\mu_{{\rm eq}}] + 2\mathcal{L}[V] + \mathcal{Q}[V,V] + 4{\rm i}\pi \bm{\epsilon}^* \cdot (\bm{\tau}(\bm{\epsilon}^*) + \bm{u}(\infty_-)) = 0.
\]
This argument was for $\beta = 2$, but we retrieve Proposition~\ref{lem:Ef0} in full generality since it is simply the $\beta = 2$ identity multiplied by $\frac{\beta}{2}$ and taking into account the prefactor $\frac{2}{\beta}$ in the definition of $\mathcal{Q}$, while $\mu_{{\rm eq}}$ and $\mathcal{L}$ are independent of $\beta$. So, it was justified (without loop in the logic) to proceed with Proposition~\ref{lem:Ef0} in the proofs of Section~\ref{sec:derivation-formulae}. In fact, the same argument would establish Proposition~\ref{lem:Ef0} as a byproduct of the proof of the $\beta = 1$ Theorem~\ref{thm:formula_beta_1} or of the $\beta = 4$ Theorem~\ref{thm:formula_beta_4} instead of Theorem~\ref{thm:formula_beta_2}.

\appendix

\section[Variation of the entropy with respect to filling fractions (Proof of Proposition~\ref{lem:compute-nu})]{Variation of the entropy with respect to filling fractions\\ (Proof of Proposition~\ref{lem:compute-nu})}
\label{AppA}

Consider the equilibrium measure $\mu_{{\rm eq},\bm{\epsilon}}$ of a $\beta$-ensemble with fixed filling fractions $\bm{\epsilon}$ such that $M(x) = t_{2g + 2}\prod_{h = 1}^{g} (x - z_h)$ with $z_h \in (b_{h - 1},a_h)$ in the notations of Section~\ref{eqmessec}. The density of~$\mu_{{\rm eq},\bm{\epsilon}}$~is
\[
\rho(x) = \frac{t_{2g + 2}}{2\pi} \prod_{h = 1}^{g} |x - z_h| \prod_{h = 0}^{g} \sqrt{|x - a_h||x - b_h|} \cdot \ind_S(x).
\]

We need to compute for each $h \in [g]$
\begin{equation*}
v_{{\rm eq},h} = \biggl(\frac{\beta}{2} - 1\biggr) \int_{S} \partial_{\epsilon_h} (\rho(x) \ln \rho(x) ) =  \biggl(\frac{\beta}{2} - 1\biggr)\int_{S} (\partial_{\epsilon_h} \rho(x) ) \ln \rho(x)\,\dd x.
\end{equation*}
For the last equality, we used that $\int_{S} \rho(x)\dd x = 1$ has vanishing $\epsilon_h$-derivative. The density $\rho$ can be expressed as a jump of
$W_{1}$ to rewrite
\begin{align}
v_{{\rm eq},h} &{} = \biggl(\frac{\beta}{2} - 1\biggr) \int_{S}\partial_{\epsilon_{h}}\frac{W_{1}(x - {\rm i}0) - W_{1}(x + {\rm i}0)}{2i\pi}\ln\rho(x)\,\dd x \nonumber\\
&{} = \biggl(\frac{\beta}{2} - 1\biggr)\Biggl(\sum_{k = 1}^{g} \Upsilon_h(z_k) + \frac{1}{2} \sum_{h = 0}^{g}  (\Upsilon_h(a_h) + \Upsilon_h(b_h) )\Biggr)\label{nueqh}
\end{align}
in terms of the integrals
\begin{equation}
\label{Ihxi} \forall \xi \in \mathbb{R} \qquad \Upsilon_h(\xi) := \int_{S} \partial_{\epsilon_h}\biggl(\frac{W_1(x - {\rm i}0) - W_1(x + {\rm i}0)}{2{\rm i}\pi}\biggr) \ln|x - \xi| \, \dd x.
\end{equation}

It is well-known (see, e.g., \cite[Appendix A]{BGmulti}) that
\[
\forall z \in \widehat{C}_+\qquad \partial_{\epsilon_h} W_1(X(z))\dd X(z) = 2{\rm i}\pi \dd u_h(z).
\]
For $x \in \mathbb{C} \setminus S$ or in $S \pm {\rm i}0$, we define $\mathfrak{z}(x)$ to be the unique point in $\overline{\widehat{C}_+}$ such that $X(\mathfrak{z}(x)) = x$. Then
\[
\Upsilon_h(\xi) = \int_{S}  (\dd u_h(\mathfrak{z}(x - {\rm i}0)) - \dd u_h(\mathfrak{z}(x + {\rm i}0)) ) \ln|x - \xi| = 2\int_{S} \dd u_h(\mathfrak{z}(x - {\rm i}0)) \ln|x - \xi|.
\]
This is a differentiable function of $\xi$. For $\xi \notin S$, we can compute
\[
\partial_{\xi} \Upsilon_h(\xi) = \int_{S}  (\dd u_h(\mathfrak{z}(x - {\rm i}0) - \dd u_h(\mathfrak{z}(x + {\rm i}0)) ) \frac{1}{\xi - x} = \oint_{S} \frac{\dd u_h(z)}{\xi - X(z)} = 2{\rm i}\pi \frac{\dd u_h}{\dd X}(\mathfrak{z}(\xi)).
\]
For $\xi \in \mathring{S}$, we rather have
\[
\partial_{\xi} \Upsilon_h(\xi) = 2 \fint_{S} \frac{\dd u_h(\mathfrak{z}(x - {\rm i}0))}{\xi - x} = -\frac{\dd u_h}{\dd X}(\mathfrak{z}(\xi + {\rm i}0)) - \frac{\dd u_h}{\dd X}(\mathfrak{z}(\xi - {\rm i}0)) = 0.
\]
We will integrate this starting along the real line starting from $\xi = - \infty + {\rm i}0$ and using the continuity of $\Upsilon_h$ on the real axis shifted by $+{\rm i}0$. From the definition \eqref{Ihxi}, we can see that $\lim_{\xi \rightarrow \infty} \Upsilon_h(\xi) = 0$. Therefore,
\begin{equation}
\label{Upsilonh}
\frac{\Upsilon_h(\xi)}{2{\rm i}\pi} = \begin{cases}\displaystyle u_h(\mathfrak{z}(\xi)) + \sum_{l = 0}^{k - 1}  (u_h(a_l) - u_h(b_l) ) & \text{if} \ \xi \in (b_{k - 1},a_k), \\
\displaystyle u_h(a_k) + \sum_{l = 0}^{k- 1}  (u_h(a_k) - u_h(b_k) ) & \text{if} \ \xi \in [a_k,b_k], \end{cases}
\end{equation}
with the conventions $b_{-1} = -\infty$ and $a_{g + 1} = +\infty$. Note that we could start integrating along the real line coming from $+\infty$, but we would get an equivalent expression because
\begin{equation*}
\sum_{k = 0}^{g} \bm{u}(a_k) = \sum_{k = 0}^{g} \bm{u}(b_k).
\end{equation*}
The primitive $\bm{u}$ of $\dd \bm{u}$ in $(\mathbb{C} \setminus S)$ is multivalued, because this domain is not simply-connected. Yet, for the previous computation, it suffices to define it by integration based at $\infty_+$ in the simply-connected domain $\mathbb{H} \setminus S$, and it is extended to $S$ and hence $\overline{\mathbb{H}}$ by continuity. Inserting the formula \eqref{Upsilonh} in \eqref{nueqh}, we arrive to
\begin{gather}
\bm{v}_{{\rm eq},h} = 2{\rm i}\pi \biggl(\frac{\beta}{2} - 1\biggr)\label{sumv}\\
\hphantom{\bm{v}_{{\rm eq},h} =}{}
\times\Biggl[\sum_{k = 1}^{g} \bigl(\bm{u}(z_k) + \bm{u}(a_0) - \bm{u}(b_0) + \cdots + \bm{u}(a_{k - 1}) - \bm{u}(b_{k - 1})\bigr) + \sum_{k = 0}^g \frac{\bm{u}(a_k) + \bm{u}(b_k)}{2}\Biggr].\nonumber
\end{gather}

We now compute $\bm{u}(a_k)$ and $\bm{u}(b_k)$ as defined above. Denote $(\bm{e}_1,\ldots,\bm{e}_g)$ the canonical basis of $\mathbb{C}^g$. Due to the description of the representatives of the $\mathcal{A}$- and $\mathcal{B}$-cycles in Section~\ref{sec:constr-spectr-curve} and the fact that the hyperelliptic involution changes the sign of $\dd \bm{u}$, we have
\begin{equation}
\label{ub0a0}
\bm{u}(b_0) - \bm{u}(a_0) = - \frac{1}{2} \oint_{\mathcal{A}_0} \dd \bm{u} = \frac{1}{2} \sum_{l = 1}^{g} \bm{e}_l,
\end{equation}
and for any $k \in [g]$
\begin{gather}
\bm{u}(b_k) - \bm{u}(a_k)  = - \frac{1}{2} \oint_{\mathcal{A}_k} \dd \bm{u} = -\frac{1}{2}\bm{e}_k, \nonumber\\
\bm{u}(a_k) - \bm{u}(b_{k - 1})  = \frac{1}{2} \oint_{\mathcal{B}_{k} - \mathcal{B}_{k - 1}} \dd \bm{u} = \frac{1}{2} (\bm{\tau}(\bm{e}_k) - \bm{\tau}(e_{k - 1})), \label{uakbk}
\end{gather}
with the conventions $\mathcal{B}_0 = 0$ and $\bm{e}_0 = 0$. Since $a_0$ is the only Weierstra\ss{} point that does not belong to the $\mathcal{A}$- and $\mathcal{B}$-cycles specified in Section~\ref{sec:constr-spectr-curve}, $\bm{u}(\infty_-)$ can be obtained by integrating~$\dd \bm{u}$ in the first sheet $-\infty$ on the real line to $a_0$, and then to $a_0$ from $-\infty$ on the real line in the second sheet. Therefore,
\[
\bm{u}(a_0) = \frac{1}{2} \bm{u}(\infty_-).
\]
From \eqref{ub0a0} and \eqref{uakbk}, we deduce
\[
\bm{u}(b_0) = \frac{1}{2}\Biggl(\bm{u}(\infty_-) + \sum_{l = 1}^g \bm{e}_l\Biggr),
\]
and for $k \in [g]$
\begin{gather*}
\bm{u}(a_k) = \frac{1}{2}\Biggl(\bm{u}(\infty_-) + \sum_{l = k}^{g} \bm{e}_l + \sum_{l = 1}^{k} \bm{\tau}(\bm{e}_l)\Biggr), \\
\bm{u}(b_k) = \frac{1}{2}\Biggl(\bm{u}(\infty_-) + \sum_{l= k + 1}^{g} \bm{e}_l + \sum_{l = 1}^{k} \bm{\tau}(\bm{e}_l)\Biggr).
\end{gather*}
Therefore,
\begin{align*}
\sum_{k = 0}^{g} \bm{u}(a_k) &{} = \sum_{k = 0}^{g} \bm{u}(b_k) = \frac{1}{2}\Biggl[(g + 1)\bm{u}(\infty_-) + \sum_{k = 1}^{g}\Biggl( \sum_{l = k}^{g} \bm{e}_l + \sum_{l = 1}^k \bm{\tau}(\bm{e}_l)\Biggr)\Biggr] \\
 &{} = \frac{1}{2}\Biggl((g + 1)\bm{u}(\infty_-) + \sum_{l = 1}^{g} (l \bm{e}_l + (g + 1 -l)\bm{\tau}(\bm{e}_l))\Biggr).
\end{align*}

We can return to the computation of $\bm{v}_{{\rm eq}}$. By definition in \eqref{vnabl} it is real, so we can replace~$\bm{u}$ by $\operatorname{Im} \bm{u}$ in \eqref{sumv}. Since $\bm{u}(b_l) - \bm{u}(a_l)$ is real for any $l \in [0,g]$, we get
\[
\bm{v}_{{\rm eq}} = 2\pi\biggl(1 - \frac{\beta}{2}\biggr)\Biggl[ \sum_{k = 1}^{g} \biggl(\operatorname{Im}\bm{u}(z_k) + \frac{g + 1 - k}{2}\,\operatorname{Im}\bm{\tau}(e_k)\biggr)+ \frac{g + 1}{2} \operatorname{Im}\bm{u}(\infty_-)\Biggr].
\]
Since we already know that $\bm{\tau}$ and $\bm{u}(\infty_-)$ are purely imaginary, we can drop imaginary part and divide by ${\rm i}$ instead, and this is the final formula.

\subsubsection*{Acknowledgments}

G.B.\ thanks Semyon Klevtsov for his encouragement to carry out this study, Angela Ortega for discussions on quadrisecants, and IHES (where this article was revised) for hospitality and excellent work conditions..
T.B.\ is supported by ERC Project LDRAM : ERC-2019-ADG Project 884584.

\pdfbookmark[1]{References}{ref}
\LastPageEnding

\end{document}